\documentclass[12pt,letterpaper]{article}

\usepackage[T1]{fontenc}
\usepackage[utf8]{inputenc}
\usepackage{geometry}
\usepackage{setspace}
\usepackage{amsmath,amssymb,amsthm,mathtools}
\usepackage{bm}
\usepackage{booktabs}
\usepackage{graphicx}
\usepackage{array}
\usepackage{enumitem}
\usepackage{caption}
\usepackage{fancyhdr}
\usepackage{titlesec}
\usepackage[hidelinks]{hyperref}
\usepackage[nameinlink,noabbrev]{cleveref}

\hypersetup{
  pdftitle={JFR-rg Part II: Dynamic Extensions, Time Constraints, and Investment Design in High-Debt, Low-Growth Economies},
  pdfauthor={Hirofumi Wakimoto}
}

\titleformat{\section}{\large\bfseries}{\thesection.}{0.5em}{}[\vspace{-0.3em}\rule{\linewidth}{0.8pt}]
\titleformat{\subsection}{\normalsize\bfseries}{\thesubsection}{0.5em}{}
\titleformat{\subsubsection}{\normalsize\bfseries}{\thesubsubsection}{0.5em}{}

\theoremstyle{plain}
\newtheorem{theorem}{Theorem}
\newtheorem{proposition}[theorem]{Proposition}
\newtheorem{lemma}[theorem]{Lemma}
\newtheorem{corollary}[theorem]{Corollary}

\theoremstyle{definition}
\newtheorem{definition}[theorem]{Definition}
\newtheorem{assumption}[theorem]{Assumption}

\theoremstyle{remark}
\newtheorem{remark}[theorem]{Remark}

\newcommand{\eps}{\varepsilon}
\newcommand{\phic}{\phi_t}
\newcommand{\phibar}{\bar{\phi}}
\newcommand{\ebar}{\bar e}
\newcommand{\gns}{g_t^{n*}}
\newcommand{\gn}{g_t^n}
\newcommand{\rn}{r_t^n}
\newcommand{\dt}{d_t}
\newcommand{\stt}{s_t}
\newcommand{\btm}{b_{t-1}}
\newcommand{\De}{\Delta e_t}
\newcommand{\Db}{\Delta b_t}
\newcommand{\RD}{\mathrm{RD}_t}
\newcommand{\Prob}{\mathbb{P}}
\newcommand{\E}{\mathbb{E}}

\title{%
  \textbf{JFR-rg Part II: Dynamic Extensions, Time Constraints, and\\
  Investment Design in High-Debt, Low-Growth Economies}\\[0.5em]
  \large A Logical-Extension Sequel to the JFR-rg Framework
}
\author{Hirofumi Wakimoto}
\date{April 19, 2026}

\begin{document}
\maketitle
\thispagestyle{empty}

\begin{abstract}
This paper develops the logical extension of the JFR-rg framework introduced in Part~I within the same observables-centered and regime-conditional architecture. Six extensions are formalized: the Virtuous Ratchet~(E1), the corrected Repression Dividend Multiplier~(E2), the Debt Reduction Paradox~(E3), the Multi-Country Repression Equilibrium~(E4), the Demographic-$\phi$ Clock~(E5), and the Institutional Control Rights Index~(E6). Together, these extensions clarify the dynamic implications of a JFR-rg regime for path dependence, institutional erosion, growth-enhancing investment, and regime transition in high-debt, low-growth economies.

The paper's claim of logical completion is architectural rather than universal. It does not assert that all calibration problems are solved, that all empirical parameters are pinned down with final precision, or that a full welfare-theoretic or political-economy microfoundation has been provided. Rather, it shows that the principal dynamic implications internal to the Part~I framework can now be stated in closed form, and that the two most natural excluded generalizations---bounded stochastic perturbations and endogenous fiscal responses---preserve the regime logic of the model.

A Minimal Equilibrium Closure is then introduced to endogenize the sovereign risk premium through a two-layer domestic demand structure and a complementarity condition. This closure is not intended as a full general-equilibrium account of sovereign bond pricing; it is a minimal sufficient device for closing the transition problem within the block-recursive structure of JFR-rg. Because the framework is observables-centered, the paper also formulates the corresponding statistical problem of inferring a latent regime boundary under one-sided regime dominance. The resulting inferential contribution is conservative by design: it constructs outer statistical summaries of the relevant boundary objects rather than forcing point classification when the available observables remain compatible with multiple nearby regime readings.

A systematic comparison with Blanchard~(2019), Hoshi--Ito~(2014), and Mehrotra--Sergeyev~(2021) further shows where JFR-rg adds explanatory value in the Japanese case: not by replacing standard debt-sustainability analysis, but by endogenizing the institutional conditions under which low sovereign rates are sustained, weakened, or lost. Empirical calibration is treated throughout as a distinct implementation layer, and each extension is paired with observable implications, failure modes, and a structured empirical roadmap.
\end{abstract}



\newpage
\tableofcontents
\newpage

\section{Introduction}
\label{sec:intro}

Part~II builds directly on Part~I~\cite{Wakimoto2026PartI}, which established JFR-rg as an observables-centered, regime-conditional framework for organizing debt dynamics under financial repression~\cite{ReinhartSbrancia2015}. Its central contribution was to clarify how a high-debt, low-growth economy could remain temporarily stable under a specific combination of a negative effective interest--growth differential, a captive domestic absorption structure, and a bounded exchange-rate channel. Just as importantly, Part~I was careful not to claim more than this architecture could support: its propositions were explicitly conditional, its failure conditions were publicly observable, and the framework was presented as a complementary analytical lens rather than as a universal replacement for mainstream debt-sustainability analysis.

The natural next question is what follows once such a regime has been identified. A regime-diagnostic framework is incomplete if it can describe contemporaneous stability but cannot state the main dynamic implications that follow from that stability condition. Once the JFR-rg regime is admitted as a useful description of a particular macro-fiscal configuration, questions immediately arise about path dependence, timing asymmetry, institutional erosion, investment design, and the conditions under which the regime may weaken or terminate. Part~II is written to address those questions.

The claim of this paper, however, should be read with precision. Part~II does not claim to complete debt theory in general, nor does it claim to provide a full welfare-theoretic, political-economy, or household-portfolio general-equilibrium foundation for sovereign bond pricing. Its claim is narrower and more disciplined: within the observables-centered, block-recursive architecture established in Part~I~\cite{Wakimoto2026PartI}, the principal dynamic implications of the framework can now be stated in closed form. In that sense, Part~I identified the regime, and Part~II develops its logical extension.

The contribution of Part~II can therefore be read at two distinct but connected levels. At the first level, the paper completes the principal logical extension of the JFR-rg architecture by deriving the dynamic implications that follow from Part~I and by closing the transition problem through a minimal equilibrium closure. At the second level, precisely because the framework is observables-centered, the paper also identifies the corresponding statistical problem implied by that closure: the relevant regime boundary is economically well defined, but only partially observed in the data, and in the Japanese reference case the premium-emergence regime may be only sparsely sampled. The inferential layer developed below is designed to address that problem conservatively. Its purpose is not to replace the economic logic of the framework, but to provide a disciplined outer statistical summary of the latent regime boundary and transition margin that the completed architecture implies.

\begin{remark}[On the scope of logical completion]
\label{rem:completion-scope}
The claim of logical completion in Part~II is architectural rather than universal. It is a claim that the principal dynamic implications that follow directly from the Part~I architecture---path-dependent ratchets, finite institutional horizons, investment-design bounds, multi-country threshold generalization, institutional control rights as a determinant of regime applicability, and the conditions for regime transition---can now be stated in closed form. Other extensions remain possible, including stochastic generalizations, political-economy endogenization of $\dt$, welfare-theoretic analysis, endogenous expectations, and richer heterogeneity. They are not developed here because the present paper is concerned with the core observables-centered debt-recursion logic of JFR-rg. \Cref{sec:robustness} formally demonstrates that the sign structure on which E1--E6 depend is invariant to bounded stochastic perturbations and Lipschitz-continuous fiscal-response endogenization. The six extensions developed below should therefore be read as the principal direct dynamic implications that the present paper identifies as following from the Part~I architecture. In that sense, omitting them would leave the framework materially unfinished at the level of its core debt-recursion logic, even though broader extensions may still be useful for other questions.
\end{remark}

This narrower claim is formally grounded rather than merely asserted. The paper shows that the two most natural excluded extensions---bounded stochastic perturbations and endogenous fiscal responses---do not overturn the sign structure on which the regime logic depends. The role of these results is not to deny that broader extensions may be useful, but to establish that their omission does not leave the JFR-rg framework dynamically unfinished at the level of its core debt-recursion logic.

On that basis, the paper develops six extensions. E1 establishes the Virtuous Ratchet, the positive counterpart to the Normalization Ratchet. E2 reformulates the Repression Dividend Multiplier in corrected bounded form. E3 states the Debt Reduction Paradox as a conditional proposition under weak fiscal offset. E4 generalizes the captive-threshold logic into a multi-country setting. E5 converts the captive-system condition from a static scope condition into a finite institutional horizon. E6 formalizes the role of institutional control rights in determining corridor width and in delimiting the conditions under which the JFR-rg regime applies beyond the Japanese reference case. A central integration result---the Transition Feasibility Proposition---then links these extensions to the question of whether a finite institutional window can be converted into a self-sustaining trajectory.

The paper also incorporates an investment-design block. Once stronger potential nominal growth is recognized as a one-for-one substitute for stronger repression in the stability condition, investment policy enters the framework endogenously at the level of implications even if the empirical estimation of investment efficiency remains external. The contribution here is therefore not a claim of precise multiplier estimation, but a disciplined statement of upper bounds, lower bounds, timing conditions, and allocation logic for growth-enhancing investment within the JFR-rg regime.

A further contribution of Part~II is to sharpen the framework's relation to the mainstream literature. The purpose of comparison is not to dismiss Blanchard~(2019), Hoshi--Ito~(2014), or Mehrotra--Sergeyev~(2021), all of which capture important dimensions of debt sustainability. Rather, the comparison asks a narrower question: which elements of Japan's 2012--2024 experience are left exogenous, only partially organized, or predicted differently in standard frameworks, and what additional explanatory value is gained when institutional repression, captive absorption, and control rights are brought inside the model? JFR-rg is offered as a complementary explanatory layer for that question, not as a wholesale substitute for mainstream macroeconomics.

The paper proceeds as follows. \Cref{sec:core} recapitulates, for completeness, the core structure inherited from Part~I~\cite{Wakimoto2026PartI}. \Cref{sec:classification} classifies the extensions by logical status, links them to falsifiable implications and failure modes, and formally establishes the robustness of the regime logic to stochastic and fiscal-response generalizations. \Cref{sec:e1} through \Cref{sec:e6} present the six extensions. \Cref{sec:timing} establishes the Timing Constraint that links E1 and E5. \Cref{sec:investment} formalizes upper and lower bounds on stabilizing growth investment. \Cref{sec:mu} provides the external estimation architecture for $\mu$. \Cref{sec:calibration} provides an illustrative calibration table. \Cref{sec:empirics} outlines the empirical program. \Cref{sec:transition} establishes the Transition Feasibility Proposition under an exogenous premium bound. \Cref{sec:closure} introduces a Minimal Equilibrium Closure that endogenizes the sovereign premium through a two-layer domestic demand structure and a complementarity condition. This closure should be read as a minimal sufficient pricing relation for the regime boundary and transition problem, not as a full sovereign-bond general-equilibrium model. \Cref{sec:regime-inference} then formulates the corresponding inferential layer for the premium-emergence boundary and the transition-feasibility margin under observables-centered constraints. \Cref{sec:mainstream-comparison} compares the explanatory reach of JFR-rg with that of mainstream frameworks in the Japanese case. \Cref{sec:conclusion} concludes.

\begin{remark}[The architecture of Part~II in one paragraph]
\label{rem:architecture-preview}
Readers who wish to orient themselves before the formal sections may find the following
preview useful.
The six extensions E1--E6 are not independent results: they are the principal dynamic
implications that follow from the Part~I regime description.
They culminate in two integration results.
First, the Transition Feasibility Proposition (\cref{sec:transition}) provides a
growth-improvement threshold for a safe exit from repression-dependent stability,
given a bound~$\bar\rho$ on the post-transition sovereign premium.
Second, the Minimal Equilibrium Closure (\cref{sec:closure}) endogenizes $\rho_t$
through a two-layer domestic demand structure---a hard captive core $\theta_t$ and a
contestable margin---combined with a complementarity condition that prices the regime
boundary.
The key objects of \cref{sec:closure} are:
(a)~the aggregate domestic demand curve $\varphi_t^d(\rho_t)$ and its monotone
properties;
(b)~the complementarity condition $0 \leq \rho_t \perp [\varphi_t^d(\rho_t) -
\varphi_t^{\mathrm{req}}] \geq 0$, which determines whether the economy is in the
JFR-rg interior regime ($\rho_t = 0$) or the transition/stress regime ($\rho_t > 0$);
and (c)~the feedback gain $\eta_t < 1$, whose satisfaction is the stability condition
for the self-reinforcing loop between premium emergence and institutional erosion.
The rest of Part~II---including E1--E6, the investment block, and the mainstream
comparison---provides the analytical inputs and institutional context for understanding
when, and at what speed, the system may approach the boundary where $\rho_t > 0$
first emerges.
\end{remark}

\section{The Core JFR-rg Structure Revisited}
\label{sec:core}

This section restates, for completeness, the core debt recursion, stability condition, and scope conditions introduced in Part~I~\cite{Wakimoto2026PartI}. They are repeated here not as new claims, but as the primitive accounting and regime-diagnostic structure on which the present extensions are built.

\subsection{Debt recursion}

The point of departure remains the core debt recursion of Part~I~\cite{Wakimoto2026PartI}:
\begin{equation}
b_t = b_{t-1}(1+\rn-\gn)+\dt.
\label{eq:core-recursion}
\end{equation}
In difference form,
\begin{equation}
\Db = (\rn-\gn)\btm+\dt.
\label{eq:core-difference}
\end{equation}
This accounting recursion is the backbone of the original model and remains the criterion for admissible extension in Part~II. Any dynamic proposition that cannot be mapped back to this structure should not be treated as part of JFR-rg proper.

\begin{remark}[Debt concept: baseline versus monitoring layers]
\label{rem:debt-concept}
The symbol $b_t$ denotes the debt-to-GDP ratio, but two distinct measurement concepts
appear across the layers of Part~II and must not be conflated.
\begin{enumerate}[label=(\roman*)]
  \item \textbf{Baseline layer} ($b_0 = 240\%$, inherited from Part~I): follows the
  IMF general-government gross-debt concept, which includes central-government bonds,
  FILP bonds, and related government-guaranteed liabilities.
  This broader concept is the appropriate accounting unit for the debt-sustainability
  corridor and all main-text propositions.

  \item \textbf{Observed monitoring layer} ($b_{t-1} \approx 157.4\%$ in \cref{app:window-sensitivity}):
  computed as outstanding JGB plus FILP bonds divided by annualized nominal GDP
  using BoJ Flow-of-Funds series.  This narrower concept serves as a cross-check against
  publicly available instrument-level data; it is \emph{not} a replacement for the
  baseline debt concept.
\end{enumerate}
All main-text propositions and calibration tables use the baseline concept.
Where the monitoring-layer ratio is inserted for observed-value cross-checks, the
debt-concept difference is explicitly noted in the relevant table or remark.
Quantitative comparisons across layers must account for this distinction:
the lower observed ratio mechanically reduces fiscal-burden terms such as $\dt/\btm$
relative to the baseline and should not be read as an independent update of the
baseline operating point.

Operationally, the two layers play different roles. The baseline layer is the
appropriate operating point for conservative policy design, corridor stress-testing,
and all main-text propositions. The monitoring layer is intended for observed-value
surveillance and empirical cross-checks using publicly available instrument-level
data. When the two differ materially, the baseline layer should govern prudential
design, while the monitoring layer should govern real-time updating of urgency and
direction.
\end{remark}

\subsection{Stability condition}

The corresponding stability condition can be written as
\begin{equation}
\eps + \gns + \alpha \De - \beta \max(0,\De-\ebar)^2
\ge
\pi_t + \frac{\dt-\stt}{\btm}.
\label{eq:stability}
\end{equation}

\begin{remark}[One-for-one substitutability]
A crucial implication of \cref{eq:stability} is that stronger potential nominal growth and stronger repression are policy substitutes. An increase in $\gns$ relaxes the required level of $\eps$ one-for-one. This property is the bridge that makes an investment block mathematically meaningful within JFR-rg, and it is the foundation of the Transition Feasibility Proposition developed in \cref{sec:transition}.
\end{remark}

\subsection{Scope conditions}

\begin{assumption}[SC1: Captive financial system]
\label{ass:sc1}
\[
\phic \ge \phibar.
\]
The domestic sovereign-debt absorption structure remains sufficiently captive that a destabilizing risk-premium regime does not yet emerge.
\end{assumption}

\begin{assumption}[SC2: Exchange-rate regime]
\label{ass:sc2}
\[
\De \le \ebar.
\]
Exchange-rate depreciation remains within the stability window of the nonlinear pass-through block.
\end{assumption}

All direct implications of the extended framework are conditional on the joint validity of SC1 and SC2.

\section{Logical Classification, Falsifiability, and Failure-Mode Integration}
\label{sec:classification}

\begin{definition}[Logical completion]
An extension is said to be logically complete within JFR-rg if its formal structure can be stated as a closed implication of the original framework together with explicitly stated auxiliary conditions where needed. Logical completion is distinct from empirical calibration.
\end{definition}

\begin{table}[htbp]
\centering
\caption{Logical Status, Falsifiable Predictions, and Failure-Mode Links of Dynamic Extensions}
\label{tab:tiers}
\small
\begin{tabular}{>{\raggedright\arraybackslash}p{1.6cm} >{\raggedright\arraybackslash}p{2.0cm} >{\raggedright\arraybackslash}p{4.2cm} >{\raggedright\arraybackslash}p{3.5cm} >{\raggedright\arraybackslash}p{2.4cm}}
\toprule
Extension & Logical Status & Interpretation & Falsifiable Prediction & Failure Mode \\
\midrule
E1 & Direct completion & Symmetric positive counterpart to the Normalization Ratchet & Sprint-period debt improvement persists after reversion; deviation: improvement reverses within 5~years & Inflation undershoot; SC2 violation \\[6pt]
E2 & Completed corrected form & Self-decelerating cumulative gain; empirical $\mu$ calibration external & Reinvested dividend generates diminishing (not explosive) returns; deviation: cumulative gain accelerates & Inflation undershoot (RD vanishes) \\[6pt]
E3 & Conditional completion & Valid under generalized fiscal-response condition & Under JFR-rg, debt reduction worsens $\Db$ when fiscal response is weak; deviation: debt reduction always improves $\Db$ & Political-fiscal reaction ($\dt$ endogeneity collapses) \\[6pt]
E4 & Completed generalized form & Cross-country calibration external & Coordinated foreign repression lowers $\phibar$; deviation: $\phibar$ independent of foreign yields & De-captivation (threshold rises) \\[6pt]
E5 & Direct completion with illustrative calibration & SC1 converted into finite-horizon constraint & $\phic$ declines at rate $\kappa>0$; deviation: $\kappa\le 0$ observed persistently & De-captivation ($T^*$ shortens) \\[6pt]
E6 & Completed generalized form with illustrative international calibration & Institutional control rights determine corridor width; JFR-rg nests mainstream limit at $\psi_t \to 0$ & Corridor width varies systematically with $\psi_t$ across countries; deviation: corridor width independent of control rights & Loss of monetary or FX autonomy (corridor collapses to mainstream case) \\
\bottomrule
\end{tabular}
\end{table}

\begin{table}[htbp]
\centering
\caption{Failure-Mode Integration: Adapted from Part~I and Extended for Part~II}
\label{tab:failure-modes}
\begin{tabular}{>{\raggedright\arraybackslash}p{4.0cm} >{\raggedright\arraybackslash}p{9.0cm}}
\toprule
Part~I Failure Mode & Relevant Part~II Extensions \\
\midrule
(i) Inflation undershoot & E1 (sprint infeasible if $\eps_t\to 0$); E2 (repression dividend vanishes) \\
(ii) De-captivation & E5 ($T^*$ shortens); E4 (threshold rises if foreign alternatives improve); E6 (corridor narrows as $\psi_t$ declines) \\
(iii) Import-cost overshoot & E1 (SC2 violated during sprint, forcing termination) \\
(iv) Political-fiscal reaction & E3 (exogeneity of $\dt$ collapses, weakening paradox conditions) \\
(v) Loss of institutional autonomy & E6 (corridor collapses toward mainstream limit when $\psi_t \to 0$; relevant for currency-union members) \\
\bottomrule
\end{tabular}
\end{table}

\subsection{Pre-Specified Falsification Strategy and Empirical Roadmap}
\label{sec:roadmap}

Part~II is intended to be analytically reproducible even without code distribution. The paper therefore emphasizes \emph{computational transparency without software distribution}: each extension is paired with a stated observable implication, a minimal calculation logic, and a candidate empirical design. The aim is not to claim that all extensions are already identified with final precision, but to make clear how each claim could fail, what evidence would be most relevant, and what such a failure would imply for the scope of the framework.

\begin{table}[htbp]
\centering
\caption{Pre-Specified Falsification Strategy and Empirical Roadmap}
\label{tab:roadmap}
{\footnotesize
\setlength{\tabcolsep}{2.7pt}
\renewcommand{\arraystretch}{1.02}
\begin{tabular}{>{\raggedright\arraybackslash}p{1.65cm} >{\raggedright\arraybackslash}p{2.85cm} >{\raggedright\arraybackslash}p{3.20cm} >{\raggedright\arraybackslash}p{3.55cm} >{\raggedright\arraybackslash}p{3.95cm}}
\toprule
Ext. & Core observable(s) & Candidate design & Main identification challenge & What rejection would mean \\
\midrule
E1 & $\Db$, $\rn-\gn$, policy-sprint episodes & Event-style descriptive analysis or local projections around identified sprint episodes & Separating true sprint effects from ordinary mean reversion & Virtuous Ratchet becomes a weaker comparative-static insight rather than a persistent path-dependence result \\
\midrule
E2 & $\RD$, public investment, productivity growth & Range calibration for $\mu$, followed by macro-average growth-accounting checks & Translating public investment into $\Delta \gns$ without overstating multiplier precision & E2 remains a bounded-accounting mechanism with no reliable quantitative deployment rule \\
\midrule
E3 & $\Db$, $\btm$, debt-service relief, effective deficit response & Threshold-based calibration and sign-consistency checks around $\gamma<|\rn-\gn|$ & Endogeneity of fiscal response and sparse clean debt-reduction episodes & The paradox is empirically limited even if logically admissible under narrow conditions \\
\midrule
E4 & foreign repression measures, alternative-asset returns, domestic SC1 margin & Comparative statics using foreign-yield and alternative-return shifts & Measuring the relevant cross-asset exit option set & E4 remains a theoretical generalization pending externally disciplined calibration \\
\midrule
E5 & $\phic$, $\kappa$, threshold assumptions & Trend, sensitivity, and horizon-indicator analysis using Flow-of-Funds data & Short samples and the decomposition of structural vs policy-driven erosion & The clock remains qualitative and cautionary rather than quantitatively binding \\
\midrule
E6 & $\psi_t$ sub-indices, cross-country $\phic$, regime outcomes & Cross-country comparison of $\psi_t$ ordering versus corridor outcomes (Japan/Italy/Greece) & Calibration of $\phic$ for non-Japanese economies; equal-weighting assumption in $\psi_t$ composite & E6 remains a qualitative regime taxonomy rather than a quantitatively predictive corridor model \\
\midrule
Transition & $\Delta \gns$, post-transition premium $\rho_t$ & Scenario-based threshold evaluation using bounded-premium cases & Bounding $\rho_t$ after SC1 weakens & Safe exit becomes a conditional possibility rather than an operational path \\
\bottomrule
\end{tabular}
}
\end{table}

\paragraph{Failure reporting rule.}
If the most relevant observable implication for a given extension is not supported, the implication of rejection is not that Part~II collapses wholesale, but that the corresponding extension should be read more narrowly. In this sense, the empirical roadmap is designed to discipline the scope of the framework rather than to immunize it against falsification.

\subsection{Robustness of Regime Logic to Excluded Extensions}
\label{sec:robustness}

\Cref{rem:completion-scope} asserted that stochastic generalizations,
political-economy endogenization of $\dt$, and welfare-theoretic analysis
need not be included in the dynamic extension because they do not alter the
regime logic.  This subsection converts that assertion into formal
propositions.

The strategy is to show that the two most natural classes of generalization---bounded
stochastic perturbations and endogenous fiscal responses---preserve the
\emph{sign structure} of the core inequalities on which the six extensions
depend.  The sign structure comprises three elements:
\begin{enumerate}[label=(S\arabic*)]
    \item the sign of the effective spread $\rn - \gn$ in the debt
    recursion~\eqref{eq:core-recursion};
    \item the sign of the stability-condition surplus in~\eqref{eq:stability};
    \item the direction of the captive-share inequality in SC1
    (\cref{ass:sc1}).
\end{enumerate}
An extension of the framework is said to \emph{preserve regime logic} if,
under the extension, every proposition in E1--E6 retains its stated conclusion
either exactly, in expectation, or almost surely, with changes confined to the
\emph{location} of thresholds rather than their \emph{existence} or
\emph{sign}.


\begin{proposition}[Robustness to bounded stochastic perturbations]
\label{prop:stochastic-robust}
Consider the stochastically augmented debt recursion
\begin{equation}
\label{eq:stochastic-recursion}
    b_t = \btm\bigl(1 + \rn - \gn + \sigma_t \eta_t\bigr) + \dt,
\end{equation}
where $\{\eta_t\}$ is a sequence of i.i.d.\ mean-zero random variables with
$|\eta_t| \leq 1$ a.s., and $\sigma_t \geq 0$ is a bounded volatility
parameter.  Then:
\begin{enumerate}[label=(\roman*)]
    \item \textbf{E1 (Virtuous Ratchet):} The expected debt improvement from a
    repression sprint of length $T$ satisfies
    \[
        \E[\Delta b^{\mathrm{cumul}}_{\mathrm{sprint}}]
        \;=\;
        T \cdot \bigl|(\rn-\gn)_{\mathrm{sprint}} - (\rn-\gn)_0\bigr| \cdot b_0,
    \]
    identical to the deterministic case.  The improvement persists in
    expectation after policy reversion, with the expected gap decaying at the
    baseline propagation factor.

    \item \textbf{E2 (Bounded gains):} If $\{b_t\}$ is bounded in
    expectation, then $\E[C_t] \leq \mu\lambda\eps\,\E[b_{t-1}^2]
    \leq \mu\lambda\eps\,\bar{b}^2$, so the marginal-gain sequence remains
    uniformly bounded in expectation.

    \item \textbf{E3 (Debt Reduction Paradox):} The derivative
    $\partial \E[\Db] / \partial \btm = (\rn - \gn) + \gamma$ is
    independent of $\sigma_t$, so the paradox condition
    $\gamma < |\rn - \gn|$ is invariant to the stochastic augmentation.

    \item \textbf{E5 (Demographic-$\phi$ Clock):} If the captive-share
    dynamics are augmented as $\phi_{t+1} = \phic - \kappa + \sigma^{\phi}_t
    \nu_t$ with $\E[\nu_t] = 0$, then $\E[\phi_{t+s}] = \phic - \kappa s$
    and the expected residual horizon is
    $T^* = (\phic - \phibar)/\kappa$, unchanged.

    \item \textbf{E6 (Control Rights):} The corridor width
    $|W_t|$ depends on $\psi_t$ through the feasible range of
    $(\eps, \gns, \De)$.  Since $\psi_t$ determines the \emph{set} of
    available policy levers, not the realization of shocks, the monotonicity
    $\partial |W_t| / \partial \psi_t \geq 0$ is preserved under any
    shock distribution.

    \item \textbf{\Cref{sec:closure} (Complementarity):} The complementarity
    condition $0 \leq \rho_t \perp [\varphi_t^d(\rho_t) -
    \varphi_t^{\mathrm{req}}] \geq 0$ is a static within-period market-clearing
    condition.  Stochastic perturbations to the debt recursion affect the
    \emph{state} $b_{t-1}$ entering the next period but do not alter the
    structure of the complementarity problem conditional on the state.
\end{enumerate}
\end{proposition}

\begin{proof}
\textit{(i).}
Taking expectations of \cref{eq:stochastic-recursion} and using $\E[\eta_t]=0$:
\[
    \E[b_t] = \E[\btm]\,(1 + \rn - \gn) + \dt.
\]
The expected recursion is identical to the deterministic
recursion~\eqref{eq:core-recursion}.  Hence all propositions whose proofs
depend only on the \emph{first-moment} properties of the recursion---in
particular, the sign of the propagation factor and the persistence of the
gap between the sprint path and the baseline path---carry over to the
expected path without modification.

\textit{(ii).}
$\E[C_t] = \mu\lambda\eps\,\E[b_{t-1}^2]$.  By Jensen's inequality,
$\E[b_{t-1}^2] \geq (\E[b_{t-1}])^2$, so the stochastic case generates
\emph{weakly larger} expected marginal gains than the deterministic case.
However, boundedness of $\{b_t\}$ implies $\E[b_{t-1}^2] \leq \bar{b}^2$,
so the uniform bound is preserved.

\textit{(iii).}
\[
    \E[\Db] = (\rn - \gn)\,\E[\btm] + \E[\dt]
        + \underbrace{\sigma_t \E[\eta_t]}_{=\,0}\,\E[\btm].
\]
Since the shock term vanishes in expectation, the derivative with respect to
$\E[\btm]$ is $(\rn-\gn)+\gamma$, identical to the deterministic case.

\textit{(iv).}
Direct from $\E[\nu_t]=0$.

\textit{(v).}
The corridor is defined by the set of $(\eps,\gns,\De)$ satisfying
\cref{eq:stability}.  This set depends on $\psi_t$ through the feasible
ranges of each variable, which are institutional constraints independent of
the realization of $\eta_t$.

\textit{(vi).}
The complementarity condition is a \emph{static} equilibrium condition
evaluated at a given state $X_t$.  Stochastic perturbations affect the
transition of $X_t$ across periods but not the within-period pricing
mechanism.
\end{proof}

\begin{remark}[What stochastic generalization does change]
\label{rem:stochastic-changes}
The stochastic augmentation does not leave the framework entirely unchanged.
It introduces three effects that are absent in the deterministic case:
\begin{enumerate}[label=(\alph*)]
    \item \textbf{Variance of the debt path:} the variance of $b_t$ grows
    with $\sigma_t$, increasing the probability that the system crosses the
    stability boundary even when the expected path remains safe.  This is a
    \emph{risk} effect, not a regime-logic effect.

    \item \textbf{Precautionary motive for the safety margin:} the margin $m$
    in the Transition Feasibility Proposition should be scaled to the
    volatility of the debt path.  This changes the \emph{calibration} of $m$,
    not the \emph{existence} of the transition condition.

    \item \textbf{Option value of delay:} under uncertainty, delaying
    irreversible policy choices (such as repression withdrawal) may have
    positive option value.  This is a welfare-theoretic consideration that
    lies outside the observables-centered architecture of JFR-rg.
\end{enumerate}
None of these effects alters the sign structure (S1)--(S3) or the functional
form of any proposition in E1--E6.  They affect threshold \emph{levels} and
\emph{risk management}, not the regime logic itself.
\end{remark}


\begin{proposition}[Robustness to fiscal-response endogenization]
\label{prop:fiscal-robust}
Let the effective deficit be determined by a general fiscal-response function
\begin{equation}
\label{eq:fiscal-general}
    \dt = D(\btm,\, s_t^{\mathrm{pol}}),
\end{equation}
where $s_t^{\mathrm{pol}}$ is a (possibly stochastic) political-state
variable, and $D$ is Lipschitz-continuous in $\btm$ with
$\partial D / \partial \btm =: \gamma(b)$ in a neighborhood of the relevant
debt level.  Then:
\begin{enumerate}[label=(\roman*)]
    \item \textbf{E3 (Debt Reduction Paradox):}
    The paradox condition generalizes to
    \begin{equation}
    \label{eq:e3-general}
        \gamma(\btm) \;<\; |\rn - \gn|.
    \end{equation}
    The functional form of the condition is unchanged; only the deficit-relief
    coefficient $\gamma$ becomes state-dependent.

    \item \textbf{E1 (Virtuous Ratchet):} The sprint improvement persists
    after policy reversion provided the fiscal stance satisfies
    $D(b_{\mathrm{sprint}}, s_t^{\mathrm{pol}}) \leq D(b_{\mathrm{baseline}},
    s_t^{\mathrm{pol}})$, i.e., the fiscal response does not offset the
    debt improvement by increasing spending.  This is a \emph{behavioral}
    condition on the fiscal authority, not a structural modification of the
    ratchet mechanism.

    \item \textbf{Stability condition:} Under \cref{eq:fiscal-general}, the
    stability condition~\eqref{eq:stability} becomes
    \[
        \eps + \gns + \alpha\De - \beta\max(0,\De-\ebar)^2
        \;\geq\;
        \pi_t + \frac{D(\btm, s_t^{\mathrm{pol}}) - \stt}{\btm}.
    \]
    The left-hand side (the regime-specific policy levers) is unchanged.
    The right-hand side (the fiscal burden) acquires a state-dependent
    deficit, but the inequality structure and the one-for-one substitutability
    between $\eps$ and $\gns$ (\cref{eq:stability}) are preserved.

    \item \textbf{All other extensions (E2, E4, E5, E6, \Cref{sec:closure}):}
    These extensions depend on $\dt$ only through its \emph{value} at a given
    state, not through its generating process.  Replacing $\dt$ with
    $D(\btm, s_t^{\mathrm{pol}})$ changes the numerical value of the
    stability-condition surplus but does not alter any sign condition,
    comparative static, or equilibrium structure.
\end{enumerate}
\end{proposition}

\begin{proof}
\textit{(i).}
The augmented debt-flow equation is
\[
    \Db = (\rn - \gn)\btm + D(\btm, s_t^{\mathrm{pol}}).
\]
Differentiating with respect to $\btm$:
\[
    \frac{\partial \Db}{\partial \btm}
    = (\rn - \gn) + \frac{\partial D}{\partial \btm}
    = (\rn - \gn) + \gamma(\btm).
\]
This is negative (the paradox holds) if and only if
$\gamma(\btm) < |\rn - \gn|$, which is \cref{eq:e3-general}.  The derivation
is identical to the proof of \cref{prop:debt-reduction-paradox} with $\gamma$
replaced by $\gamma(\btm)$.

\textit{(ii).}
The ratchet proof (\cref{prop:virtuous-ratchet}) compares two paths evolving
under the \emph{same} recursion after policy reversion.  If the fiscal
response does not systematically differ between the two paths (i.e., the
fiscal authority does not react differently to the lower debt level), both
paths evolve under the same $(r^n, g^n, d)$ triple, and the gap decays at
the baseline propagation factor.

\textit{(iii)--(iv).}
The stability condition, E2, E4, E5, E6, and the complementarity condition
of \Cref{sec:closure} all treat $\dt$ as a given input at each $t$.  Replacing
this input with $D(\btm, s_t^{\mathrm{pol}})$ changes the input value but
not the functional relationship between the remaining variables.
\end{proof}

\begin{remark}[What fiscal endogenization does change]
\label{rem:fiscal-changes}
Endogenizing $\dt$ introduces one substantive complication: the fiscal
authority may react to the JFR-rg regime itself.  If the government
interprets captive-system stability as permission to run larger deficits
(moral hazard), then $\partial D / \partial \phibar < 0$: a more relaxed
scope condition leads to a larger deficit, partially offsetting the
stability benefits of repression.  This is a political-economy feedback
that changes the \emph{calibration} of the stability condition (by raising
the right-hand side of \cref{eq:stability}) but does not alter the
\emph{functional form} of any proposition.  The relevant failure mode is
(iv) in \cref{tab:failure-modes} (``Political-fiscal reaction''), which is
already accounted for in the failure-mode integration of Part~II.
\end{remark}


\begin{corollary}[Scope of logical completion]
\label{cor:completion-scope}
The regime logic of JFR-rg---defined as the sign structure (S1)--(S3) and
the functional form of the propositions in E1--E6 and \Cref{sec:closure}---is
invariant to:
\begin{enumerate}[label=(\alph*)]
    \item bounded stochastic perturbations of the debt recursion
    (\cref{prop:stochastic-robust});
    \item Lipschitz-continuous endogenization of the fiscal response
    (\cref{prop:fiscal-robust}).
\end{enumerate}
The remaining class of excluded extensions---welfare-theoretic analysis,
full political-economy microfoundations with strategic fiscal agents, and
endogenous institutional-quality dynamics---answer questions that lie
\emph{outside} the observables-centered block-recursive architecture of
JFR-rg.  They may enrich the framework but cannot contradict its regime
logic, because they operate on different dependent variables (welfare,
political equilibria, institutional quality) rather than on the
debt-recursion sign structure that defines the regime.
\end{corollary}

\begin{proof}
Direct from \cref{prop:stochastic-robust,prop:fiscal-robust} and the
observation that welfare and political-economy extensions do not modify
\cref{eq:core-recursion,eq:stability}.
\end{proof}

\begin{remark}[Relation to \cref{rem:completion-scope}]
\label{rem:completion-upgrade}
\Cref{cor:completion-scope} upgrades the scope claim of
\cref{rem:completion-scope} from a stated assertion to a formally grounded
result.  The claim of logical completion now rests on three pillars:
\begin{enumerate}[label=(\arabic*)]
    \item the six extensions E1--E6 organize the principal dynamic implications
    that this paper identifies as internal to the Part~I architecture
    (established by construction in \cref{sec:e1}--\cref{sec:e6});
    \item the most natural generalizations (stochastic, fiscal) preserve the
    regime logic (\cref{prop:stochastic-robust,prop:fiscal-robust}); and
    \item the minimal equilibrium closure of \Cref{sec:closure} endogenizes the
    last remaining free variable ($\rho_t$) within the block-recursive
    structure.
\end{enumerate}
Extensions that fall outside these three categories either preserve the
same regime logic or answer different questions that lie outside the
present observables-centered scope.
\end{remark}

\section{Extension E1: The Virtuous Ratchet}
\label{sec:e1}

Part~I established the Normalization Ratchet: a temporary tightening shock can generate a debt legacy whose decay is governed only by the baseline propagation factor, making the damage quasi-permanent on policy horizons. The same linear structure also implies a mathematically symmetric positive result.

\begin{proposition}[Virtuous Ratchet]
\label{prop:virtuous-ratchet}
Consider a temporary policy episode of length $T$ in which
\[
(\rn-\gn)_{\text{sprint}} < (\rn-\gn)_0
\]
for each period of the episode, and in which the fiscal stance does not deteriorate relative to baseline. Then the debt improvement generated during the episode persists after reversion and decays only at the baseline propagation factor $(1+r_0^n-g_0^n)$.
\end{proposition}

\begin{proof}
Let $\Delta_t := b_t^{\mathrm{baseline}} - b_t^{\mathrm{sprint}}$ denote the cumulative
debt improvement at period~$t$ relative to the no-sprint baseline path.
During the sprint, $\Delta_t$ grows because the spread term is more negative and the
deficit term is weakly no worse.
After policy reversion at period~$T$, both paths evolve under the same linear recursion
\cref{eq:core-recursion}, so
\[
  \Delta_{T+s} = (1 + r_0^n - g_0^n)^s \, \Delta_T, \quad s \ge 0.
\]
Hence the improvement gap decays only at the common baseline propagation factor
$(1 + r_0^n - g_0^n)$ and is never restored to zero on a finite policy horizon as long as
the baseline spread satisfies $|r_0^n - g_0^n| < 1$.
\end{proof}

\begin{remark}[Institutional asymmetry]
The asymmetry is not mathematical but institutional. A normalization shock can occur even when the captive system is already weakening. A repression sprint, by contrast, requires that SC1 still hold. If $\phic$ is declining, the state's ability to execute a strategically deeper repression phase is finite even if the debt recursion itself would allow the resulting improvement to persist. This institutional constraint is formalized in the Timing Constraint of \cref{sec:timing}.
\end{remark}

The value of the Virtuous Ratchet lies not in repression as an end in itself, but in the possibility of creating a debt-compression window within which growth-enhancing investment can be financed and institutionalized.

\section{Extension E2: The Corrected Repression Dividend Multiplier}
\label{sec:e2}

Earlier formulations of the Repression Dividend Multiplier suggested that if the repression dividend were reinvested, it might recursively ease exit conditions in an accelerating fashion. The corrected formulation rejects that overstatement while retaining the deeper insight that the repression dividend can generate finite cumulative gains.

\begin{definition}[Repression Dividend]
The annual repression dividend is defined as
\begin{equation}
\RD = \eps \btm.
\label{eq:rd}
\end{equation}
\end{definition}

\begin{proposition}[Bounded marginal improvement from repression-dividend reinvestment]
\label{prop:corrected-e2}
Suppose that a share $\lambda \in (0,1]$ of the repression dividend is converted into growth-enhancing investment and that this investment raises potential nominal growth through an efficiency parameter $\mu>0$. Let the incremental debt-compression gain in period $t$ be
\begin{equation}
C_t := \mu \lambda \eps \, b_{t-1}^2.
\label{eq:Ct}
\end{equation}
If the debt path $\{b_t\}$ remains bounded on the horizon of interest, then $\{C_t\}$ is uniformly bounded. If, in addition, $\{b_t\}$ is weakly decreasing, then $\{C_t\}$ is also weakly decreasing. Consequently, reinvestment of the repression dividend generates bounded and self-decelerating gains rather than explosive cumulative improvement.
\end{proposition}

\begin{proof}
Since $\{b_t\}$ is bounded on the horizon of interest, there exists $\bar b<\infty$ such that $b_{t-1}\le \bar b$ for all $t$. Hence, by \cref{eq:Ct},
\[
0 \le C_t \le \mu \lambda \eps\, \bar b^2.
\]
Thus $\{C_t\}$ is uniformly bounded. If $\{b_t\}$ is weakly decreasing, then $b_t^2 \le b_{t-1}^2$, so again by \cref{eq:Ct},
\[
C_{t+1} \le C_t.
\]
Therefore the marginal gain sequence is weakly decreasing. The cumulative gain over any finite horizon $T$ is
\[
\sum_{t=1}^T C_t,
\]
which is finite for every finite $T$, while its marginal increments are bounded and non-increasing. Hence the process is self-damping rather than explosive.
\end{proof}

\begin{remark}
The stronger claim that the infinite-horizon sum $\sum_{t=1}^{\infty} C_t$ is finite requires additional assumptions, such as summability of $b_t^2$ or a vanishing-debt path, and is not imposed here. The completed claim of E2 is therefore one of bounded, diminishing marginal gains rather than guaranteed finite total cumulative gain.
\end{remark}

\begin{remark}[Scope: repression-active states only]
\label{rem:e2-inactive}
\Cref{prop:corrected-e2} applies only when $\eps_t > 0$, i.e., when the repression
channel is active and the repression dividend $\RD = \eps_t \btm$ is positive.
When $\eps_t \leq 0$---as observed in Japan from mid-2024 onward, where the
10-year JGB yield began to exceed core CPI---the repression dividend disappears and
E2 becomes \emph{inactive}: there is no dividend to reinvest, and the upper bound
\cref{eq:rd-upper} is non-positive.
This is not a theoretical failure of E2 but a regime transition consistent with
Table~\ref{tab:failure-modes}: failure mode~(i) (inflation undershoot, RD vanishes)
is precisely the condition that deactivates E2.
The observed $\eps_t \approx -0.81\%$ in the mid-2025 monitoring layer
(\cref{app:window-sensitivity}) reflects this inactive state.
Investment during such a phase must rely on external financing rather than the
repression dividend.
\end{remark}

\section{Extension E3: The Debt Reduction Paradox}
\label{sec:e3}

The Debt Reduction Paradox is the most counterintuitive result of Part~II and must therefore be stated with full conditioning assumptions. The version below generalizes the original statement to accommodate endogenous fiscal responses of the Bohn~(1998) type.

\begin{proposition}[Debt Reduction Paradox with deficit-relief response]
\label{prop:debt-reduction-paradox}
Suppose that:
\begin{enumerate}[label=(\alph*)]
    \item $\rn-\gn<0$,
    \item the effective deficit satisfies
    \begin{equation}
    \dt = d_0 + \gamma(\btm - b^{\mathrm{ref}}), \qquad \gamma \ge 0,
    \label{eq:deficit-relief}
    \end{equation}
    where $\gamma$ measures the extent to which a lower debt stock mechanically reduces the effective deficit through interest-service relief or related fiscal easing, and
    \item debt reduction does not itself impair the captive system.
\end{enumerate}
Then lowering $\btm$ worsens the annual debt flow $\Db$ whenever the deficit-relief coefficient satisfies
\begin{equation}
\gamma < |\rn - \gn|.
\label{eq:e3-condition}
\end{equation}
When $\gamma \ge |\rn - \gn|$, the deficit-relief response is strong enough to offset the loss of compression, and the paradox does not arise.
\end{proposition}

\begin{proof}
Substituting \cref{eq:deficit-relief} into \cref{eq:core-difference} gives
\[
\Db = (\rn-\gn)\btm + d_0 + \gamma(\btm - b^{\mathrm{ref}}).
\]
Differentiating with respect to $\btm$ yields
\[
\frac{\partial \Db}{\partial \btm} = (\rn-\gn) + \gamma.
\]
Because $\rn-\gn<0$, this derivative is negative if and only if
\[
\gamma < |\rn-\gn|.
\]
A negative derivative means that a higher debt stock reduces $\Db$, equivalently that lowering the debt stock worsens the annual debt flow. This is precisely the paradox.
\end{proof}

\begin{remark}
A canonical Bohn-type fiscal response, in which the primary surplus rises with debt, is a different mechanism. It should not be conflated with \cref{eq:deficit-relief}. In the present proposition, $\gamma$ captures a debt-relief channel associated with lower debt service or related fiscal easing, not a fiscal-discipline response of the standard Bohn type.
\end{remark}

\begin{remark}[Empirical anchor for $\gamma$]
\label{rem:gamma-empirical}
The paradox is empirically operative only when $\gamma < |\rn - \gn|$ \emph{under JFR-rg conditions}. Japan's 2005--2007 primary-balance improvement episode remains informative, but in the current data implementation it is better interpreted as a \emph{contrast case}: the episode coincides with QE exit and policy-rate normalization, and the realized spread remains positive throughout the three-year window ($|r-g| \approx 1.27\%$, $1.64\%$, and $1.72\%$ in 2005--2007, respectively). After correcting the unit convention in the paradox test, the observed implementation flags 2017 and 2021---not 2005--2007---as the years in which the paradox condition is satisfied. The empirical layer should therefore identify candidate episodes from the realized sign pattern on $\rn-\gn$, debt-ratio compression, and the fiscal response, and then test whether the estimated $\gamma$ falls below the contemporaneous threshold. On this reading, 2005--2007 helps distinguish non-JFR fiscal consolidation from genuinely repression-supported debt compression rather than serving as an automatic positive anchor for the paradox itself.
\end{remark}

\section{Extension E4: Multi-Country Repression Equilibrium}
\label{sec:e4}

Extension E4 generalizes the captive-threshold logic into a multi-country environment.

\begin{definition}[Generalized captive threshold]
Let the effective threshold be
\begin{equation}
\phibar = f(\eps_{1,t},\eps_{2,t},\dots,\eps_{N,t},r^{\mathrm{alt}}_t),
\label{eq:phibar-generalized}
\end{equation}
where the foreign repression states and the return on alternative assets jointly determine the effective threshold relevant for sovereign outflow pressure.
\end{definition}

\begin{proposition}[Multi-country repression equilibrium]
\label{prop:e4}
If other major economies also maintain negative real sovereign returns, then the sovereign-bond channel of outflow pressure weakens, lowering the effective captive threshold relative to an otherwise identical world with positive real foreign sovereign returns. However, the threshold has a positive lower bound determined by the availability of alternative assets.
\end{proposition}

The following lemma provides the comparative-statics foundation for \cref{prop:e4} in a minimal two-country setting.

\begin{lemma}[Two-country comparative statics]
\label{lem:two-country}
Consider two countries, Home ($H$) and Foreign ($F$), each with captive shares $\phi_H$ and $\phi_F$. Suppose Home's effective captive threshold is
\[
\bar\phi_H = h\bigl(\eps_F, r^{\mathrm{alt}}\bigr),
\qquad
\frac{\partial h}{\partial \eps_F} < 0,
\quad
\frac{\partial h}{\partial r^{\mathrm{alt}}} > 0.
\]
Then:
\begin{enumerate}[label=(\roman*)]
\item An increase in Foreign repression ($\eps_F$ rises) lowers $\bar\phi_H$, expanding Home's SC1 margin.
\item An increase in the return on non-sovereign alternatives ($r^{\mathrm{alt}}$ rises) raises $\bar\phi_H$, tightening Home's SC1 margin.
\item Even under generalized sovereign repression ($\eps_F > 0$ for all $F$), $\bar\phi_H > 0$ as long as $r^{\mathrm{alt}} > 0$.
\end{enumerate}
\end{lemma}

\begin{proof}
Claims (i) and (ii) follow directly from the sign assumptions on $h$. Claim~(iii) follows from the observation that domestic holders always have the option of exiting sovereign bonds into non-sovereign alternatives; as long as $r^{\mathrm{alt}}>0$, the exit option has positive value, and the threshold below which exit pressure dominates remains strictly positive.
\end{proof}

\begin{remark}
\cref{lem:two-country} provides the minimal theoretical foundation for the 2022--2025 monitoring episode, in which the Federal Reserve's rate-hiking cycle raised $r^{\mathrm{alt}}$ for Japanese institutional investors and exerted downward pressure on the domestic-holder share $\phic$ even though SC1 remained satisfied. The empirical calibration of $h(\cdot)$---including the relevant cross-asset allocation parameters---remains external to the framework but is now anchored in a well-defined comparative-statics structure.
\end{remark}

\section{Extension E5: The Demographic-\texorpdfstring{$\phi$}{phi} Clock}
\label{sec:e5}

Among the dynamic extensions, the most practically consequential is E5. The original model treated SC1 as a static scope condition. The present extension adds a time dimension.

\begin{definition}[Residual time before SC1 failure --- linear case]
\label{def:clock}
Let $\kappa>0$ denote the structural annual decline rate in the captive share. Then the residual time before SC1 failure is defined by
\begin{equation}
T^*=\frac{\phic-\phibar}{\kappa}.
\label{eq:clock}
\end{equation}
\end{definition}

\begin{proposition}[Finite-horizon interpretation of SC1]
\label{prop:finite-window}
If $\kappa>0$, then SC1 is not merely a static regime condition but a finite-horizon constraint. The policy relevance of the JFR-rg regime depends not only on whether $\phic \ge \phibar$ holds today, but also on the time remaining before the inequality is expected to fail.
\end{proposition}

\begin{proof}
Under \cref{eq:clock}, any strictly positive decline rate in $\phic$ maps a positive current distance $(\phic-\phibar)$ into a finite positive horizon. Hence, even if the system satisfies SC1 today, the duration of that satisfaction is finite unless $\kappa=0$.
\end{proof}

\begin{proposition}[Linear clock as an upper bound under uniformly faster-than-linear erosion]
\label{prop:linear-upper-bound}
Suppose that, until the threshold $\phibar$ is reached, the captive share satisfies the differential inequality
\[
\dot\phi_\tau \le -\kappa
\]
for some constant $\kappa>0$. Then the linear clock
\begin{equation}
T^*_{\mathrm{lin}}=\frac{\phic-\phibar}{\kappa}
\label{eq:linear-clock}
\end{equation}
is an upper bound on the residual horizon.
\end{proposition}

\begin{proof}
Integrating $\dot\phi_\tau \le -\kappa$ from $t$ to $t+s$ yields
\[
\phi_{t+s} \le \phi_t - \kappa s.
\]
The threshold is reached no later than the smallest $s$ satisfying
\[
\phi_t - \kappa s = \phibar,
\]
which is exactly \cref{eq:linear-clock}. Hence the linear clock is an upper bound.
\end{proof}

\begin{remark}[Exponential proportional-decay alternative]
\label{rem:exp-alternative}
Under proportional decay,
\[
\dot\phi_\tau = -\kappa_{\exp}\phi_\tau,
\]
the residual horizon is
\begin{equation}
T^*_{\exp} = \frac{1}{\kappa_{\exp}}\ln\!\left(\frac{\phic}{\phibar}\right).
\label{eq:exp-clock}
\end{equation}
This should be treated as an alternative benchmark, not as a universally tighter upper bound. If one normalizes the exponential specification to match the initial absolute decline of the linear specification, i.e.
\[
\kappa_{\exp} = \frac{\kappa_{\mathrm{lin}}}{\phic},
\]
then
\[
T^*_{\exp} = \frac{\phic}{\kappa_{\mathrm{lin}}}\ln\!\left(\frac{\phic}{\phibar}\right)
>
\frac{\phic-\phibar}{\kappa_{\mathrm{lin}}}
=
T^*_{\mathrm{lin}},
\]
because $\ln(1/x)>1-x$ for $x\in(0,1)$. Thus proportional exponential decay does not, in general, generate a tighter upper bound than the linear clock. A genuine upper-bound refinement requires front-loaded or accelerating absolute erosion.
\end{remark}

\begin{remark}[Calibration status]
The logic of the Demographic-$\phi$ Clock is complete. Its concrete numerical calibration remains illustrative. The observed decline in $\phic$ may mix policy-driven and structural demographic factors, $\phibar$ remains only illustratively anchored, and $\kappa$ may itself be time-varying. Accordingly, $T^*$ should be interpreted as a conservative horizon indicator rather than a deterministic countdown.
\end{remark}

\begin{remark}[Baseline $\kappa = 1\%$/yr as a stress anchor, not a central forecast]
\label{rem:kappa-stress}
The main-text baseline uses $\kappa = 0.01$ (1\%/yr), inherited from Part~I as a
conservative stress-style operating point rather than an empirical central estimate.
This value is \emph{not} derived from the full-sample or post-QQE structural estimates
reported in \cref{rem:kappa-empirical} ($\kappa_{\mathrm{full}} = 0.055\%$/yr,
$\kappa_{\mathrm{post\mbox{-}QQE}} = 0.188\%$/yr): it is chosen to stress-test the
timing constraint under a scenario of accelerated erosion, analogous to a regulatory
stress scenario rather than a base-case projection.
Three distinct readings of $\kappa$ should therefore be kept separate:
\begin{enumerate}[label=(\roman*)]
  \item \textbf{Stress baseline} ($\kappa = 1\%$/yr): conservative anchor for
  policy-design bounds, preserving continuity with Part~I.
  \item \textbf{Observed monitoring estimate} ($\kappa_{\mathrm{post\mbox{-}QQE}}
  \approx 0.19\%$/yr): structural trend estimated from post-QQE Flow-of-Funds data;
  used for monitoring updates.
  \item \textbf{Policy-maintenance scenario} ($\kappa \approx 0$, as observed
  2022--2025): reflects near-term slowing of erosion under active BoJ balance-sheet
  operations; not treated as a permanent structural rate.
\end{enumerate}
None of these readings invalidates the others: they answer different questions about
the regime's time horizon under different assumptions about future policy and demographics.
\end{remark}

\begin{remark}[Observed monitoring implementation for $\kappa$]
\label{rem:kappa-empirical}
The full-sample observed implementation using BoJ Flow-of-Funds data for 1997Q4--2025Q4 now turns this empirical program into a concrete monitoring result. The estimated annualized slope is $\kappa_{\mathrm{full}} = 0.055\%$ per year ($p=0.0004$), with a statistically significant structural break around the launch of QQE (Chow $F=24.37$, $p<10^{-8}$). Splitting the sample yields $\kappa_{\mathrm{pre\mbox{-}QQE}} = -0.132\%$ per year ($p=0.0007$) and $\kappa_{\mathrm{post\mbox{-}QQE}} = 0.188\%$ per year ($p<10^{-6}$). The short-window result $\kappa \approx 0$ over 2022--2025 in \cref{app:window-sensitivity} should therefore be read as a local monitoring snapshot rather than as the structural trend itself. The combined evidence is more consistent with materially slowed post-QQE erosion than with a literally paused clock throughout the whole QQE/YCC period. On the current observed monitoring calibration $(\phic,\phibar,\kappa_{\mathrm{post\mbox{-}QQE}})=(0.932,0.85,0.001876)$, the linear clock implies $T^*_{\mathrm{monitor}} \approx 43.6$ years. This monitoring horizon does not replace the inherited Part~I baseline used for the main-text illustrations; it refines it by showing that the post-QQE erosion rate is positive but materially slower than the conservative baseline clock.
\end{remark}

\begin{figure}[t]
\centering
\includegraphics[width=0.92\textwidth]{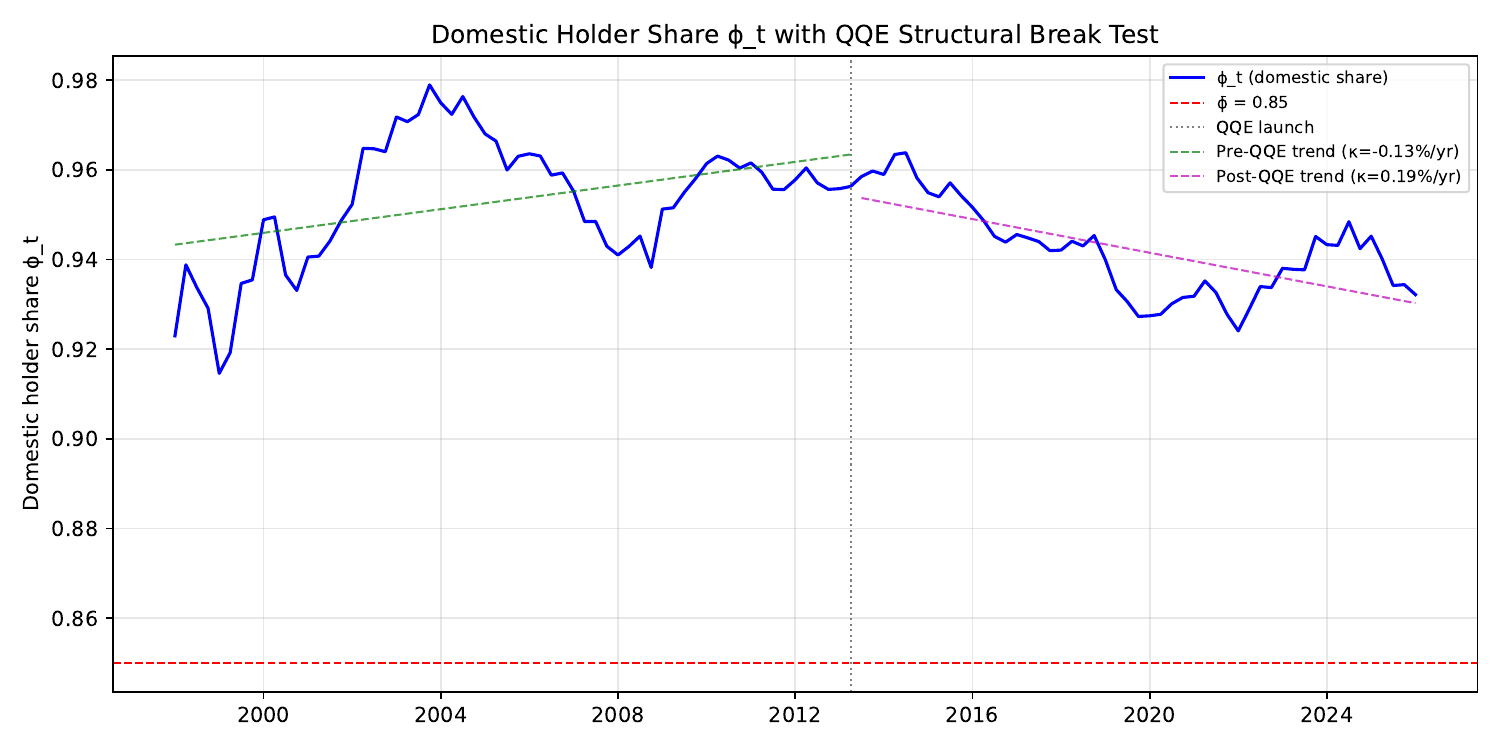}
\caption{Domestic holder share $\varphi_t$ (BoJ Flow-of-Funds, instrument~311, 1997Q4--2025Q4) with estimated pre-QQE and post-QQE linear trends. The Chow structural-break test rejects trend homogeneity at $p < 10^{-8}$. Pre-QQE, the holder share rose at approximately $0.13\%$/year; post-QQE, it declined at approximately $0.19\%$/year. The illustrative threshold $\bar{\varphi} = 0.85$ remains well below the observed series throughout. See \Cref{rem:kappa-empirical} for the full estimation results.}
\label{fig:phi-timeseries}
\end{figure}

\section{Extension E6: Institutional Control Rights and Regime Autonomy}
\label{sec:e6}

The preceding extensions have clarified the domestic mechanics of the JFR-rg regime: how debt compression works (E1--E2), when it becomes paradoxical (E3), how cross-border repression shifts the captive threshold (E4), and how long the institutional window lasts (E5). What remains unformalized is the question of \emph{why} the regime applies to some high-debt economies but not others. The answer, this extension argues, lies in the degree of domestic institutional control over the key state variables $\phic$, $\eps$, and $\De$.

\begin{definition}[Institutional control-rights index]
\label{def:psi}
Let $\psi_t \in [0,1]$ denote the degree of domestic institutional control over the JFR-rg primitives at time~$t$. The index is constructed as a composite of three observable sub-indices:
\[
\psi_t = \frac{1}{3}\bigl(\psi_t^{\mathrm{mon}} + \widehat{\psi}_t^{\mathrm{abs}} + \psi_t^{\mathrm{fx}}\bigr),
\]
where:
\begin{itemize}[leftmargin=2em]
\item $\psi_t^{\mathrm{mon}} \in [0,1]$: monetary-policy autonomy ($1$ = independent central bank with full rate-setting and balance-sheet discretion; $0.5$ = shared supranational authority such as the ECB; $0$ = currency board or external peg);
\item $\widehat{\psi}_t^{\mathrm{abs}} \in [0,1]$: an empirical \emph{proxy} for debt-absorption autonomy; and
\item $\psi_t^{\mathrm{fx}} \in \{0,1\}$: exchange-rate autonomy ($1$ = free float; $0$ = currency union or hard peg).
\end{itemize}
In the baseline empirical implementation, debt-absorption autonomy is proxied by the observed domestic holder share,
\[
\widehat{\psi}_t^{\mathrm{abs}} = \phic.
\]
More generally, one may allow a hybrid proxy
\[
\widehat{\psi}_t^{\mathrm{abs}} = \omega_{\phi}\,\phic + \omega_{\theta}\,\theta_t,
\qquad \omega_{\phi},\omega_{\theta} \ge 0,
\qquad \omega_{\phi}+\omega_{\theta}=1,
\]
where $\theta_t$ is the hard captive core share defined in \Cref{sec:closure}. The baseline uses $(\omega_{\phi},\omega_{\theta})=(1,0)$ in order to preserve the observables-centered implementation layer and to avoid embedding additional parametric structure into the core definition. When $\psi_t = 1$, the sovereign possesses full autonomy over all three JFR-rg levers (the Japan baseline). When $\psi_t \to 0$, control is predominantly external, and the framework collapses toward the mainstream limiting case of \cref{app:limits}.
\end{definition}

\begin{remark}[Stability-condition-derived weights]
\label{rem:psi-weights}
The equal-weight composite $\psi_t = (\psi_t^{\mathrm{mon}} + \widehat{\psi}_t^{\mathrm{abs}} + \psi_t^{\mathrm{fx}})/3$ is an illustrative baseline specification rather than a structural claim about exact cardinal contributions. An alternative, theory-grounded weighting derives from the marginal contribution of each channel to the stability-condition surplus in~\cref{eq:stability}. Specifically, the contribution of each channel to the corridor-supporting institutional space at the March 2026 baseline is: monetary autonomy ($\psi_t^{\mathrm{mon}}$) governs $\eps$, contributing approximately $\eps \cdot \btm$ to the surplus; absorption autonomy ($\widehat{\psi}_t^{\mathrm{abs}}$) governs SC1 viability, which is a prerequisite for the corridor; and exchange-rate autonomy ($\psi_t^{\mathrm{fx}}$) governs $\alpha \De$, contributing $\alpha \cdot \De \cdot \btm$ near the baseline. Since the absorption channel operates partly as an enabling condition rather than purely as a continuous margin, exact percentage-point comparability across the three channels is limited. The equal-weight specification is therefore retained as the baseline taxonomy, with the understanding that under alternative weightings (for example, $w_{\mathrm{mon}} = 0.5$, $w_{\mathrm{abs}} = 0.3$, $w_{\mathrm{fx}} = 0.2$), the ordering Japan $>$ Italy $>$ Greece is preserved (\cref{app:international}). In the present paper, $\psi_t$ should primarily be read as a comparative institutional-ordering device unless otherwise stated. Its cardinal use is illustrative and baseline-dependent, whereas its ordinal implications are the more robust object for cross-country interpretation.
\end{remark}

\begin{remark}[Observables-centered construction]
The sub-indices are deliberately constructed from observable institutional facts rather than from latent structural parameters. $\psi_t^{\mathrm{mon}}$ reflects the legal and operational independence of the central bank; $\widehat{\psi}_t^{\mathrm{abs}}$ is proxied from sovereign-debt holder data; and $\psi_t^{\mathrm{fx}}$ reflects the exchange-rate arrangement. This preserves continuity with the observables-centered methodology of Part~I while avoiding a literal identification between institutional absorption capacity and the observed holder share itself. For Japan, the baseline implementation uses BoJ Flow-of-Funds holder data for instrument 311 (central government securities and FILP bonds), with $\phic$ treated as the baseline proxy and $\theta_t$ retained separately as the hard captive core. An official-style aggregation that combines instruments 310 and 311 is reported, where needed, only as a sensitivity layer because the residual depends on the statistical treatment of short-term government bills. Empirical proxy construction---including the use of BoJ Flow-of-Funds data for Japan, Banca d'Italia/ECB Securities Holdings Statistics for Italy, and Hellenic PDMA reports for Greece---is detailed in \cref{app:international}.
\end{remark}

\begin{proposition}[Corridor width depends on institutional control rights]
\label{prop:corridor-psi}
The width of the Debt Sustainability Corridor $|W_t|$ is weakly increasing in the control-rights index~$\psi_t$. In the limiting case $\psi_t \to 0$, the corridor collapses and the framework reduces to the mainstream case in which debt dynamics are governed solely by the effective interest--growth differential and the fiscal stance.
\end{proposition}

\begin{proof}
From the stability condition~\eqref{eq:stability}, the corridor boundaries are defined by the feasible range of $(\eps, \gns, \De)$ that satisfies the inequality. Each term on the left-hand side is a function of domestic policy levers: $\eps$ depends on $\psi_t^{\mathrm{mon}}$ through monetary operations and yield-curve control; the debt-absorption channel is captured empirically by $\widehat{\psi}_t^{\mathrm{abs}}$ (proxied in the baseline by $\phic$ and, in sensitivity analysis, by combinations of $\phic$ and $\theta_t$); and $\De$ depends on $\psi_t^{\mathrm{fx}}$ through exchange-rate policy autonomy. When $\psi_t \downarrow 0$, the domestic authority loses independent levers over $\eps$ and $\De$, and the captive structure that sustains SC1 erodes. The admissible policy space narrows to zero, and the corridor collapses to the mainstream case where only the exogenous $(r_t^n - g_t^n)$ differential and the fiscal stance $d_t$ matter.
\end{proof}

\begin{remark}[Japan versus Eurozone regime types]
\label{rem:jp-vs-euro}
The $\psi_t$ framework provides a parsimonious explanation for why Japan and Eurozone high-debt economies exhibit qualitatively different debt dynamics despite similar debt-to-GDP ratios:
\begin{itemize}[leftmargin=2em]
\item \textbf{Japan} ($\psi_t \approx 0.97$ in 2025): Full domestic control over $\phic$ (captive system), $\eps$ (BoJ yield-curve control until 2024; ongoing balance-sheet operations), and $\De$ (free-floating yen). The corridor is positive and, until the recent $\eps$-reversal, wide enough to support growth-enhancing investment within the repression dividend.
\item \textbf{Italy} ($\psi_t \approx 0.39$): $\phic$ is partially supported by ECB purchases (PSPP/APP) but remains externally constrained; $\De = 0$ by construction (euro); $\eps$ is largely dictated by ECB policy rates. The effective corridor is narrower and more fragile, consistent with observed sovereign-spread episodes (2011--2012, 2018, 2022).
\item \textbf{Greece} ($\psi_t \approx 0.28$): $\phic$ is low because most debt is held by official external creditors (ESM/EFSF); $\De = 0$; $\eps$ is externally determined. SC1 is not satisfied ($\phic < \phibar$), and the regime falls into the mainstream limiting case. This is consistent with the observed necessity of external programme support rather than self-sustaining JFR-rg stability.
\end{itemize}
See \cref{app:international} for the full cross-country comparison, including time-series evidence and sensitivity analysis confirming the robustness of this classification within $\pm 5$ percentage-point variation in $\phic$ calibration.
\end{remark}

\begin{proposition}[Control rights and transition feasibility]
\label{prop:psi-transition}
A regime transition from repression-dependent stability to mainstream sustainability is feasible only if $\psi_t$ remains sufficiently high during the transition to prevent an immediate corridor collapse. Formally, there exists a threshold $\psi^*$ such that the joint condition
\[
\psi_t \ge \psi^* \quad \text{and} \quad \Delta g_{\min}^{n*} \le \mu \, x_t^{\max,\mathrm{operational}}
\]
is required for a non-disruptive exit.
\end{proposition}

\begin{proof}
If $\psi_t < \psi^*$ during the transition, the corridor width shrinks below the minimum needed to finance the required growth investment $x_t^{\min,\mathrm{operational}}$ while maintaining the safety margin~$m$. Since the transition requires simultaneously raising $g_t^{n*}$ and withdrawing $\eps$, the corridor must remain open long enough for the growth substitution to take hold. If $\psi_t$ falls below $\psi^*$ prematurely---for instance, through loss of monetary autonomy or a sudden rise in foreign holding---the corridor collapses and the economy is forced into the mainstream regime before the growth path is self-sustaining.
\end{proof}

\begin{corollary}
\label{cor:euro-mu}
In Eurozone economies where $\psi_t^{\mathrm{mon}} \le 0.5$ and $\psi_t^{\mathrm{fx}} = 0$, the external control constraint raises the required investment efficiency $\mu$ for a successful transition, because the corridor available for growth investment is narrower at every $\phic$ level. This provides a structural explanation for why fiscal consolidation alone has historically been insufficient for high-debt Eurozone members.
\end{corollary}

\begin{remark}[Link to prior extensions]
E6 generalizes E4 (multi-country repression) by endogenizing the cross-border flow of control rights rather than treating the foreign repression state as purely exogenous. It also tightens the Timing Constraint of E5: the Demographic-$\phi$ Clock now depends not only on demographic erosion of $\phic$ but also on potential erosion of $\psi_t^{\mathrm{mon}}$ or $\psi_t^{\mathrm{fx}}$ through institutional changes (e.g., currency-union entry or supranational fiscal constraints). Finally, E6 sharpens the investment-design block of \cref{sec:investment}: the operational lower bound $x_t^{\min,\mathrm{operational}}$ is $\psi_t$-dependent, since a narrower corridor demands more efficient investment to achieve the same growth target.
\end{remark}

\section{The Timing Constraint: Linking E1 and E5}
\label{sec:timing}

The Virtuous Ratchet (E1) and the Demographic-$\phi$ Clock (E5) are logically independent propositions, but their joint implication is a binding constraint on policy design. The following proposition formalizes this link.

\begin{proposition}[Timing Constraint]
\label{prop:timing}
The Virtuous Ratchet is implementable only within the residual SC1 window. That is, a repression sprint of length $T_{\mathrm{sprint}}$ is feasible only if
\begin{equation}
T_{\mathrm{sprint}} \le T^*.
\label{eq:timing-constraint}
\end{equation}
Moreover, the cumulative debt improvement from the sprint is bounded by
\begin{equation}
\Delta b_{\mathrm{sprint}}^{\mathrm{cumul}}
\le
T^* \cdot \bigl|(\rn-\gn)_{\mathrm{sprint}} - (\rn-\gn)_0\bigr| \cdot b_0,
\label{eq:sprint-bound}
\end{equation}
where $b_0$ is the debt stock at the start of the sprint.
\end{proposition}

\begin{proof}
By E5, SC1 fails at horizon $T^*$. By E1, a repression sprint requires SC1 for execution. If $T_{\mathrm{sprint}} > T^*$, the sprint must be terminated before completion because the captive system can no longer support the deeper repression. The cumulative bound follows from the observation that the per-period improvement is at most $|(\rn-\gn)_{\mathrm{sprint}} - (\rn-\gn)_0| \cdot b_{t-1} \le |(\rn-\gn)_{\mathrm{sprint}} - (\rn-\gn)_0| \cdot b_0$ (since the sprint improves the debt path, so $b_{t-1}\le b_0$ during the sprint), summed over at most $T^*$ periods.
\end{proof}

\begin{remark}
The Timing Constraint transforms the abstract possibility of a Virtuous Ratchet into a concrete policy-design problem: the sprint must be deep enough to generate meaningful improvement, yet short enough to complete before the institutional window closes. This is the precise sense in which JFR-rg stability is best understood as a time-bounded opportunity.
\end{remark}

\begin{remark}[Structural horizon versus market-compression horizon]
\label{rem:two-horizons}
The Timing Constraint as stated in \cref{prop:timing} links the sprint window to the
\emph{structural erosion horizon} $T^*$ defined by the $\phi$-clock of E5.
However, policy urgency is also governed by a second, potentially faster horizon:
the \emph{market-compression horizon} at which an adverse shift in the outside-option
spread $z_t$ or in institutional control rights $\psi_t$ pushes the system from
case~(a) to case~(c) of the complementarity condition in \cref{sec:closure}.

From the counterfactual scenarios of \cref{tab:stress-v2}, the external-stress
threshold (premium emerging when $z_t$ rises above $\approx 2.5\%$ at the current
$\theta_t$) is notably closer to the current position than the core-erosion threshold
(premium emerging when $\theta_t$ falls below $\approx 0.56$).
Since $z_t$ can move suddenly in response to global monetary conditions---as in the
2022--2024 Federal Reserve tightening cycle---while demographic erosion of $\theta_t$
is slow-moving, the effective policy urgency should be read as the
\emph{minimum} of the structural horizon and the market-compression horizon:
\[
  T^{\mathrm{effective}} = \min\bigl(T^*_{\phi},\; T^*_{z}\bigr),
\]
where $T^*_{z}$ is the horizon at which a plausible $z_t$ shock crosses the
external-stress threshold.
This observation does not modify the formal statement of \cref{prop:timing} but
is essential for interpreting the baseline clock in a policy context.
\end{remark}

\section{Investment Design Under JFR-rg}
\label{sec:investment}

Once stronger $\gns$ is recognized as a substitute for stronger repression, the policy problem becomes one of investment design.

\subsection{Investment-to-growth bridge}

Let policy investment raise potential nominal growth according to
\begin{equation}
\Delta \gns = \mu x_t
\label{eq:mu-bridge}
\end{equation}
in the one-parameter aggregate case, or more generally,
\begin{equation}
\Delta \gns
=
\sum_j \mu_j x_{j,t}
+
\sum_{j<k}\gamma_{jk}x_{j,t}x_{k,t}.
\label{eq:mu-bridge-disagg}
\end{equation}

\subsection{Upper bounds on growth investment}

\begin{theorem}[Upper-bound structure]
\label{thm:upper}
The maximum feasible growth investment in any period is constrained by the most conservative of four bounds: the arithmetic bound, the repression-dividend bound, the safe-corridor bound, and the time-constrained cumulative bound.
\end{theorem}

\begin{corollary}[Arithmetic upper bound]
\label{cor:arith-upper}
If additional growth investment $x_t$ is treated as an added fiscal burden and the policymaker permits debt deterioration up to $\bar\delta_t$, then
\begin{equation}
x_t^{\max,\mathrm{arith}}(\bar\delta_t)
=
\bar\delta_t-(\rn-\gn)\btm-\dt.
\label{eq:arith-upper}
\end{equation}
\end{corollary}

\begin{corollary}[Repression-dividend upper bound]
\label{cor:rd-upper}
An internal-funding upper bound is given by
\begin{equation}
x_t^{\max,\mathrm{RD}}
\le
\lambda \eps \btm,
\qquad 0\le \lambda \le 1.
\label{eq:rd-upper}
\end{equation}
\end{corollary}

\begin{corollary}[Safe-corridor upper bound]
\label{cor:safe-upper}
If the policymaker requires a positive safety margin $m>0$ inside the Debt Sustainability Corridor, then
\begin{equation}
x_t^{\max,\mathrm{safe}}
=
\Bigl[
\eps+\gns+\alpha \De-\beta \max(0,\De-\ebar)^2-\pi_t
\Bigr]\btm
-(\dt-\stt)-m.
\label{eq:safe-upper}
\end{equation}
\end{corollary}

\begin{corollary}[Time-constrained cumulative upper bound]
\label{cor:time-upper}
If the finite horizon $T^*$ of \cref{def:clock} is imposed, then the conservative cumulative upper bound is
\begin{equation}
\mathcal{X}^{\max}(T^*)
\le
\sum_{\tau=t}^{t+\lfloor T^*\rfloor-1}
\min
\left\{
x_\tau^{\max,\mathrm{arith}},\;
x_\tau^{\max,\mathrm{RD}},\;
x_\tau^{\max,\mathrm{safe}}
\right\}.
\label{eq:time-upper}
\end{equation}
\end{corollary}

\begin{remark}[Operational upper-bound principle]
Policy design should use the most conservative operational upper bound:
$x_t^{\max,\mathrm{operational}}
=
\min\{
x_t^{\max,\mathrm{arith,cons}},\;
x_t^{\max,\mathrm{RD,cons}},\;
x_t^{\max,\mathrm{safe,cons}}
\}$.
\end{remark}

\subsection{Lower bounds on stabilizing investment}

\begin{theorem}[Lower-bound structure]
\label{thm:lower}
The minimum required growth investment is determined by the most demanding of three requirements: the static minimum, the shock-buffer minimum, and the demographic-compensation minimum.
\end{theorem}

\begin{corollary}[Static minimum requirement]
\label{cor:static-lower}
\begin{equation}
x_t^{\min,\mathrm{static}}
=
\max
\left\{
0,\;
\frac{
\pi_t+\frac{\dt-\stt}{\btm}
-
\left[
\eps+\gns+\alpha \De-\beta\max(0,\De-\ebar)^2
\right]
+m
}{\mu}
\right\}.
\label{eq:static-lower}
\end{equation}
\end{corollary}

\begin{corollary}[Shock-buffer lower bound]
\label{cor:shock-lower}
If a moderate normalization shock worsens $\rn-\gn$ by approximately one percentage point, then the required growth push is
\begin{equation}
x_t^{\min,B}=\frac{1.0}{\mu}.
\label{eq:shock-lower}
\end{equation}
\end{corollary}

\begin{corollary}[Demographic-compensation lower bound]
\label{cor:demo-lower}
If labor-force decline depresses potential nominal growth by approximately $\delta_{\mathrm{demo}} \in [0.5,\, 0.8]$ percentage points per year, then
\begin{equation}
x_t^{\min,\mathrm{demography}}
=
\frac{\delta_{\mathrm{demo}}}{\mu}.
\label{eq:demo-lower}
\end{equation}
\end{corollary}

\begin{definition}[Operational lower bound]
\begin{equation}
x_t^{\min,\mathrm{operational}}
=
\max
\left\{
x_t^{\min,\mathrm{static}},\;
x_t^{\min,B},\;
x_t^{\min,\mathrm{demography}}
\right\}.
\label{eq:operational-lower}
\end{equation}
\end{definition}

\begin{remark}[Political-economy constraint]
The investment bounds above are derived under the assumption that the repression dividend can be directed toward growth-enhancing investment. In practice, political pressures may redirect fiscal space toward transfers, consumption subsidies, or defense spending that does not raise $\gns$. The effective $\lambda$ in \cref{cor:rd-upper} should therefore be interpreted as a political-economy-adjusted share, not a technical optimum. This constraint does not alter the theoretical architecture but may substantially reduce the feasible investment envelope in practice.
\end{remark}

\subsection{Allocation problem}

\begin{definition}[Investment allocation problem]
Given an aggregate investment envelope $\bar X_t$, the JFR-rg-consistent allocation problem is
\begin{equation}
\max_{\{x_{j,t}\ge 0\}}
\Bigl[
\eps+\gns
+\sum_j \mu_j x_{j,t}
+\sum_{j<k}\gamma_{jk}x_{j,t}x_{k,t}
+\alpha \De
-\beta \max(0,\De-\ebar)^2
-\pi_t-\frac{\dt-\stt}{\btm}
\Bigr]
\label{eq:allocation}
\end{equation}
subject to $\sum_j x_{j,t}\le \bar X_t$.
\end{definition}

The preferred objective is maximization of the safety margin relative to the corridor boundary. This choice reflects the regime-conditional logic of JFR-rg: inside the corridor, the first-order concern is not maximizing growth per se, but ensuring that the system does not exit the stability region under plausible shocks. Growth maximization and margin maximization coincide when the system is well inside the corridor; they diverge precisely when the system operates near the boundary---which, as Part~I documented, is Japan's current position.

\section{External Estimation Architecture for \texorpdfstring{$\mu$}{mu}}
\label{sec:mu}

\subsection{Positioning}

\begin{definition}[Status of $\mu$]
$\mu$ is not treated as an internally proven structural constant of the core JFR-rg model. It is treated as an externally estimated calibration parameter that links policy investment to $\Delta \gns$.
\end{definition}

\subsection{Three-layer estimation architecture}

\paragraph{Layer M1: Macro exogenous calibration.}
\begin{equation}
\Delta g_{t+h}^{n*}
=
\mu_h x_t + W_t' \delta_h + \varepsilon_{t+h}.
\label{eq:m1}
\end{equation}

\paragraph{Layer M2: Sector-specific policy-panel estimation.}
\begin{equation}
\Delta g_t^{n*}
=
\sum_j \mu_j x_{j,t}
+
\sum_{j<k}\gamma_{jk}x_{j,t}x_{k,t}
+\xi_t.
\label{eq:m2}
\end{equation}

\paragraph{Layer M3: Bayesian or scenario range calibration.}
\begin{equation}
\mu \in [\mu_L,\mu_M,\mu_H]
\qquad\text{or}\qquad
\mu \sim \mathcal{D}_{\mu}.
\label{eq:m3}
\end{equation}

A conservative working range is $\mu \in [0.02, 0.08]$, interpreted as an illustrative range motivated by the broader public-investment and growth-accounting literature rather than as a structural estimate internal to JFR-rg. Lower values correspond to weak efficiency and slow transmission into productivity, whereas higher values correspond to unusually effective growth-oriented deployment. The range is explicitly provisional and subject to revision as Layer~M1 and M2 estimates become available.

\begin{remark}[Situating the $\mu$ range relative to existing empirical literature]
\label{rem:mu-literature}
Part~II does not estimate $\mu$; it situates operational thresholds relative to an
externally grounded empirical range.
To give this range empirical content, three reference points from the literature are noted.

First, Bom and Ligthart (2014) report a meta-analytic mean output elasticity of public
capital of approximately 0.15 across 68 studies.
Translating an elasticity of 0.15 into a growth-impact parameter requires knowledge of
the public-capital-to-output ratio and the lag structure, but it is broadly consistent
with $\mu$ values in the central part of the working range at reasonable calibration
of investment-to-GDP ratios.

Second, DeLong and Summers (2012) argue that under hysteresis, fiscal multipliers of
1.0 to 1.5 are plausible in depressed economies.  If a multiplier of 1.0 translates
into sustained growth effects (rather than purely cyclical demand effects),
the implied $\mu$ would fall in the upper part of the working range.

Third, the Japanese evidence on public investment productivity is mixed:
infrastructure investments of the 1970s--1990s are associated with higher estimated
$\mu$, while post-2000 fiscal stimulus shows lower returns, consistent with the
conservative end of the range.

Taken together, these reference points suggest that:
\begin{itemize}[leftmargin=2em]
  \item $\mu \approx 0.02$--$0.03$ (conservative): consistent with diminishing-returns
  environments or low-quality deployment (transfers mislabeled as investment).
  \item $\mu \approx 0.05$ (central): broadly consistent with the Bom--Ligthart
  meta-analytic midrange at plausible calibration.
  \item $\mu \approx 0.07$--$0.08$ (optimistic): consistent with high-efficiency,
  complementarity-heavy investment programs under favorable demand conditions.
\end{itemize}
The feasibility assessments in \cref{tab:stress-v2} (``Conditional at $\mu \geq 0.05$'',
``Tight at $\mu \geq 0.07$'') should be read against this external anchoring:
the central case requires efficiency in the middle of the empirically documented range,
and the tight case requires efficiency at the high end.
\end{remark}

\subsection{Boundary between inside and outside JFR-rg}

\begin{definition}[Boundary rule]
\leavevmode
\begin{itemize}
    \item \textbf{Inside JFR-rg:} required $\Delta \gns$, implied reduction in necessary $\eps$, corridor arithmetic, and upper- and lower-bound calculations.
    \item \textbf{Outside JFR-rg:} estimation of $\mu$, sectoral efficiency differences, timing lags, induced private investment, and implementation heterogeneity.
    \item \textbf{JFR-rg-centered sequel work:} insertion of externally estimated $\mu$-ranges or $\mu_j$-vectors into the corridor, threshold, and allocation problems.
\end{itemize}
\end{definition}

\section{Illustrative Calibration}
\label{sec:calibration}

The following table illustrates the quantitative movability of the Part~II extensions using the Part~I baseline calibration (March 2026 FRED data). All entries are computed from the closed-form expressions derived above; no simulation model is required. For continuity with Part~I, the main-text quantitative illustrations in Part~II intentionally inherit the March~2026 operating point used there. Updated observed-value insertions from public sources are reported separately as monitoring or robustness layers and are not used to redefine the baseline of the sequel.

\begin{table}[htbp]
\centering
\caption{Illustrative Calibration of Part~II Extensions (March 2026 Baseline)}
\label{tab:calibration}
\small
\begin{tabular}{>{\raggedright\arraybackslash}p{4.5cm} >{\raggedright\arraybackslash}p{3.0cm} >{\raggedright\arraybackslash}p{6.0cm}}
\toprule
Quantity & Value & Source / Derivation \\
\midrule
$b_0$ (debt/GDP) & 240\% & Baseline calibration inherited from Part~I and anchored to public-debt series \\
$\rn - \gn$ (baseline spread) & $-0.8\%$ & Part~I Scenario A \\
$\eps_t$ (repression bias) & $+0.5\%$ & $\pi_t - \rn = 2.7\% - 2.2\%$ \\
$\dt$ (primary deficit) & 2.0\% GDP & Part~I calibration \\
\midrule
\multicolumn{3}{l}{\textit{E1: Virtuous Ratchet --- 2-year sprint, $\eps=1.0\%$}} \\
Sprint spread $(\rn-\gn)_{\mathrm{sprint}}$ & $-1.3\%$ & $\eps$ raised from 0.5\% to 1.0\% \\
Cumulative improvement & $\approx 2.4$~pp & $2 \times |{-1.3}-({-0.8})| \times 2.40$ \\
\midrule
\multicolumn{3}{l}{\textit{E2: Repression Dividend}} \\
Annual RD & 1.2\% GDP & $0.005 \times 2.40$ \\
\midrule
\multicolumn{3}{l}{\textit{E3: Debt Reduction Paradox}} \\
$\gamma$ threshold & 0.8 pp & $|\rn-\gn| = 0.8\%$ \\
\midrule
\multicolumn{3}{l}{\textit{E5: Demographic-$\phi$ Clock}} \\
$\phic$ (current) & 0.88 & BoJ Flow of Funds \\
$\phibar$ (illustrative) & 0.85 & Hoshi and Ito (2014) \\
$\kappa$ (annual decline) & 0.01 & Observed 2022--2026 trend \\
$T^*_{\mathrm{linear}}$ & 3.0 years & $(0.88-0.85)/0.01$ \\
$T^*_{\exp}$ & 3.47 years & $(1/0.01)\ln(0.88/0.85)$ \\
$T^*$ under $\kappa=0.005$ & 6.0 years & $(0.88-0.85)/0.005$ \\
$T^*$ under $\kappa=0.0075$ & 4.0 years & $(0.88-0.85)/0.0075$ \\
\midrule
\multicolumn{3}{l}{\textit{Investment bounds (at $\mu=0.05$, $\lambda=0.5$)}} \\
$x_t^{\max,\mathrm{RD}}$ & 0.6\% GDP & $0.5 \times 0.005 \times 2.40$ \\
$x_t^{\min,B}$ (shock buffer) & 20\% GDP & $1.0/0.05$ (public $+$ induced private) \\
$x_t^{\min,\mathrm{demo}}$ & 10--16\% GDP & $(0.5\text{--}0.8)/0.05$ \\
\bottomrule
\end{tabular}

\vspace{0.5em}
\begin{minipage}{0.92\textwidth}
\footnotesize
\textit{Notes.}
$b_0=240\%$ is used as the baseline calibration inherited from Part~I and anchored to the public-debt series together with the March 2026 scenario setup; it should not be read as a single directly observed March 2026 print. The working range $\mu \in [0.02,0.08]$ is illustrative and is motivated by the broader public-investment and growth-accounting literature; it is not a structural estimate internal to JFR-rg. The purpose of this table is to demonstrate numerical movability of the closed-form expressions, not to claim final empirical calibration. The reported $T^*_{\exp}$ line uses the direct proportional-decay benchmark with $\kappa_{\exp}=0.01$; it is therefore not the normalized comparison discussed in \cref{rem:exp-alternative}, where the exponential path is rescaled to match the initial absolute decline of the linear specification. Separate observed-value updates may be reported for monitoring or robustness, but such updates do not replace the inherited Part~I baseline used for the main-text illustrations of Part~II.
\end{minipage}
\end{table}

\begin{table}[htbp]
\centering
\caption{Observed Monitoring Update for Japan}
\label{tab:monitoring-update}
\small
\begin{tabular}{>{\raggedright\arraybackslash}p{5.0cm} >{\raggedright\arraybackslash}p{3.2cm} >{\raggedright\arraybackslash}p{5.3cm}}
\toprule
Quantity & Value & Source / Derivation \\
\midrule
\multicolumn{3}{l}{\textit{Latest observed state (latest available)}} \\
$\phic$ & 0.932 (2025) & BoJ Flow of Funds, annualized from 2025Q4 holder data \\
$\theta_t$ & 0.797 (2025) & BoJ $+$ banks $+$ insurance $+$ pensions over instrument 311 total \\
$\psi_t$ (baseline proxy) & 0.977 $\approx 0.98$ & $(1+\phic+1)/3$ with $\widehat{\psi}^{\mathrm{abs}}_t=\phic$ \\
Named-holder coverage & 0.982 & Latest coverage ratio in the 311-only baseline \\
Residual gap share & 1.77\% & Latest residual under the named-holder aggregation \\
Worst-quarter residual gap & 4.30\% & 2001Q2 coverage residual \\
\midrule
\multicolumn{3}{l}{\textit{Structural-break monitoring}} \\
$\kappa_{\mathrm{full}}$ & 0.055\%/yr & Full-sample annualized slope, $p=0.0004$ \\
$\kappa_{\mathrm{pre\mbox{-}QQE}}$ & $-0.132\%$/yr & Pre-QQE slope, $p=0.0007$ \\
$\kappa_{\mathrm{post\mbox{-}QQE}}$ & 0.188\%/yr & Post-QQE slope, $p<10^{-6}$ \\
Chow test & $F=24.37$, $p<10^{-8}$ & Break at the launch of QQE strongly supported \\
Observed monitoring clock $T^*_{\mathrm{lin}}$ & 43.6 years & $(0.932-0.85)/0.001876$ \\
\midrule
\multicolumn{3}{l}{\textit{Fiscal-response monitoring}} \\
$\gamma_{\mathrm{OLS}}$ & $-0.0199$ & Full-sample estimate (SE $=0.0104$, $p=0.0685$) \\
2005--2007 episode & Contrast case & $r-g$ positive throughout; not a JFR-rg anchor \\
Paradox years & 2017, 2021 & Unit-consistent implementation of the paradox test \\
\bottomrule
\end{tabular}

\vspace{0.5em}
\begin{minipage}{0.92\textwidth}
\footnotesize
\textit{Notes.} This table reports the observed monitoring layer produced by the full-sample implementation used for Part~II. It does not replace the inherited March~2026 baseline used for the main-text closed-form illustrations of Part~II. Instead, it shows how the same framework behaves when the latest public data are inserted directly into the measurement layer. The large difference between the inherited baseline clock (3.0 years) and the observed monitoring clock (43.6 years) should therefore be read as a difference between a conservative stress-style operating point and a later observed monitoring layer, not as an internal contradiction of the theory.
\end{minipage}
\end{table}

\subsection{Illustrative logic figures}

\begin{figure}[htbp]
\centering
\IfFileExists{demographic_phi_clock_sensitivity.pdf}{\includegraphics[width=0.78\textwidth]{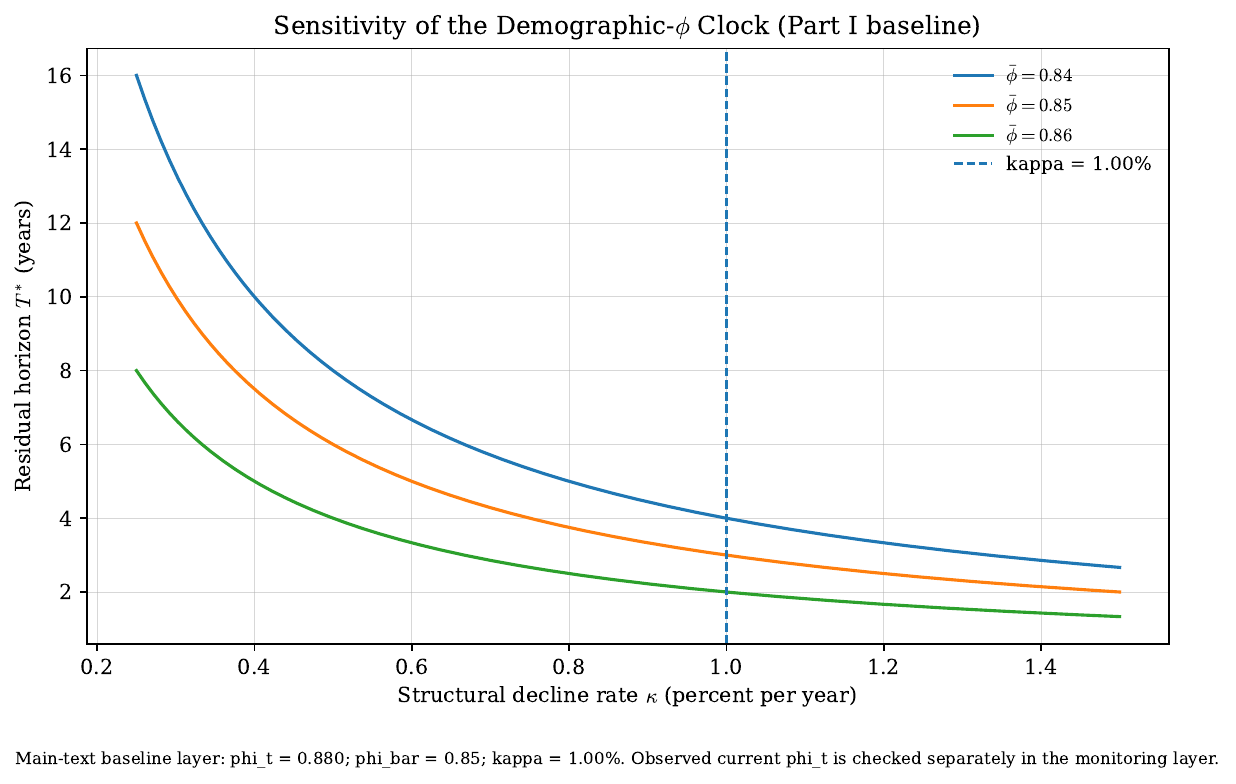}}{\fbox{\parbox{0.75\textwidth}{\centering Placeholder: demographic\_phi\_clock\_sensitivity.pdf not available in this environment.}}}
\caption{Sensitivity of the Demographic-$\phi$ Clock. The figure illustrates how the residual horizon $T^*$ implied by \cref{eq:clock} varies with the structural decline rate $\kappa$ and alternative illustrative threshold values $\phibar$. For continuity with Part~I, the main-text version of the figure inherits the March~2026 baseline anchor used there; updated observed-value insertions belong to a separate monitoring layer and do not redefine the baseline of Part~II. The purpose is analytic transparency rather than precise forecasting.}
\label{fig:phi-clock}
\end{figure}

\begin{figure}[htbp]
\centering
\IfFileExists{investment_bounds_map.pdf}{\includegraphics[width=0.78\textwidth]{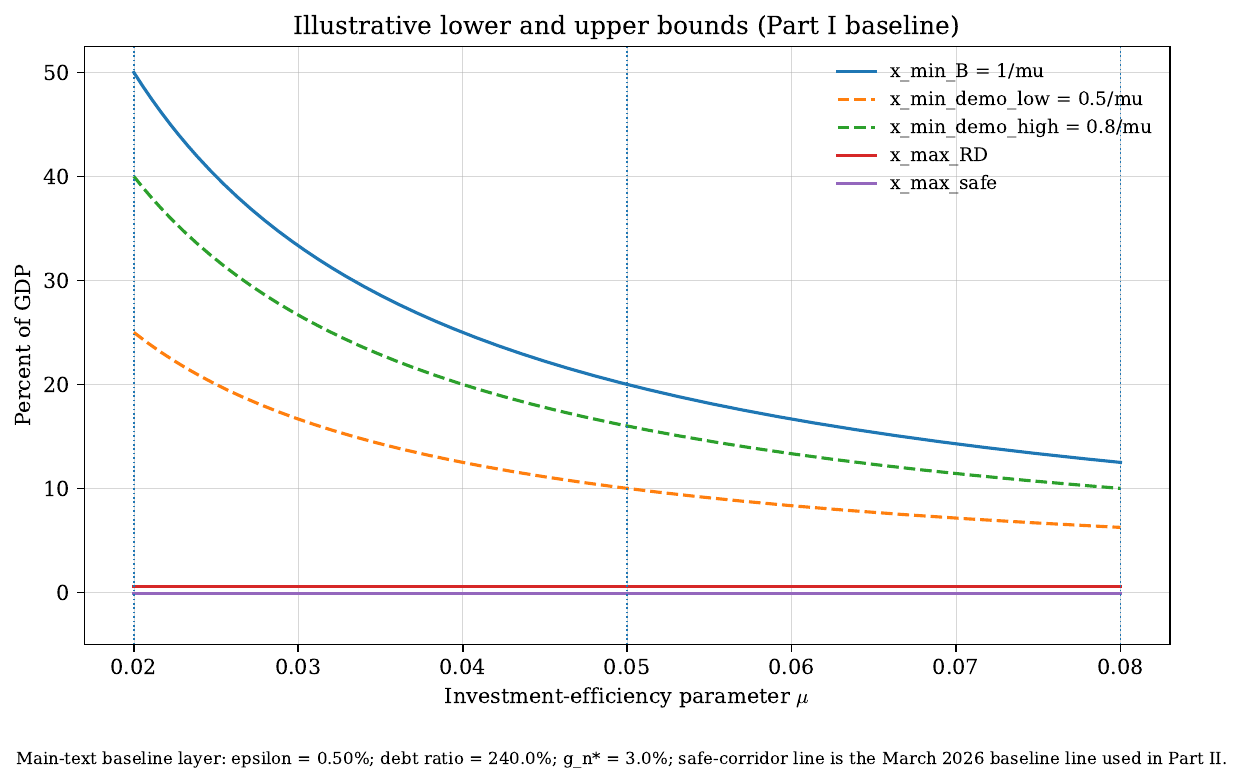}}{\fbox{\parbox{0.75\textwidth}{\centering Placeholder: investment\_bounds\_map.pdf not available in this environment.}}}
\caption{Illustrative lower and upper bounds on stabilizing growth investment. The lower-bound lines use \cref{eq:shock-lower,eq:demo-lower}; the internal-funding upper bound uses the repression-dividend logic of \cref{eq:rd-upper}; and the near-zero safe-corridor line reflects the tight March~2026 baseline inherited from Part~I. Updated observed-value insertions are reported, when needed, as monitoring or robustness figures rather than as replacements for the main-text baseline map. The figure is intended as a map of analytic tension, not as a final calibrated policy prescription.}
\label{fig:bounds-map}
\end{figure}

\paragraph{Observed-value monitoring updates.}
When the latest public observations are inserted into the Part~II formulas---for example, updated Flow-of-Funds values for $\phic$, realized-CPI variants for $\eps_t$, or updated market yields---the resulting charts should be reported as monitoring or robustness figures rather than as replacements for \cref{fig:phi-clock,fig:bounds-map}. This preserves continuity with the Part~I operating point while keeping the framework observables-centered and open to real-time empirical discipline.

\paragraph{Growth-proxy window sensitivity.}
A key finding from the observed-value implementation is that the structural nominal-growth proxy---estimated via log-linear trend on ESRI quarterly GDP data---is highly sensitive to the window length chosen. \Cref{app:window-sensitivity} reports a systematic comparison of 12-quarter and 24-quarter windows against the paper baseline. The 24-quarter window yields $g^{n*}_{\mathrm{struct}} \approx 2.87\%$, closely recovering the Part~I calibration of $3.0\%$. The 12-quarter window, which captures primarily the post-2022 inflationary acceleration, produces $g^{n*}_{\mathrm{struct}} \approx 6.19\%$---a value that reverses the sign of $x_t^{\max,\mathrm{safe}}$ and thus qualitatively alters the investment-bound map. This sensitivity motivates the use of longer estimation windows in any empirical implementation and provides independent support for the plausibility of the $g_0^{n*}=3.0\%$ baseline.

\section{An Empirical Program for Part~II}
\label{sec:empirics}

The empirical ambition of Part~II should be structured rather than maximalist. The objective is not to claim that every extension has already been fully identified in reduced-form data, but to specify a disciplined program through which the completed theoretical architecture can be confronted with observables.

Because the contribution of Part~II lies in logical extension rather than in a novel estimation algorithm, the paper does not rely on code distribution as a condition of transparency. Instead, it adopts an analytic-reproducibility standard: formulas, input values, calculation logic, tables, and figures are stated explicitly so that each core numerical illustration can be reproduced at the level of arithmetic or spreadsheet calculation.

\subsection{Four tasks of the empirical program}

The empirical program of Part~II has four tasks.

First, the directly integrated and conditionally integrated extensions should be illustrated by transparent parameter exercises. In particular, the Virtuous Ratchet, the Debt Reduction Paradox, and the Demographic-$\phi$ Clock should each be presented not merely as formal results but as operational policy maps under the Part~I baseline calibration. Observed-value updates can then be layered on top as separate monitoring exercises, but they should not be allowed to blur the inherited baseline that gives Part~II its continuity as a sequel.

Second, the framework should preserve the observables-centered falsifiability principle of Part~I. Each extension therefore carries at least one observable failure condition. For E1, the relevant test is whether sprint-period debt improvement persists after policy reversion. For E3, the relevant test is whether debt reduction under JFR-rg conditions lowers rather than raises annual debt accumulation when the deficit-relief response is weak. For E5, the relevant test is whether the captive share deteriorates materially faster than the conservative clock benchmark. For the transition block, the relevant test is whether the post-transition premium remains bounded within the range required by \cref{prop:transition}.

Third, the paper should specify an external empirical strategy for $\mu$. The preferred order is conservative range calibration, followed by macro-average estimation, and only later by sector-specific $\mu_j$ and interaction effects $\gamma_{jk}$. This sequencing minimizes scope creep and preserves the boundary between core JFR-rg logic and external policy evaluation.

Fourth, the empirical program should explicitly distinguish logical completion from final calibration. Part~II claims that the dynamic extension of Part~I can be stated in closed form. It does not claim that every external parameter has already been estimated with final precision. The empirical task is therefore implementation and calibration, not rescue of an otherwise incomplete theory.

\subsection{Extension-by-extension empirical design}

\paragraph{E1 (Virtuous Ratchet).}
The testable implication is that a temporary repression sprint generates debt improvement whose persistence exceeds what ordinary mean reversion would predict. A natural empirical strategy is a Local Projection design~\cite{Jorda2005} in which future debt-flow outcomes are projected on sprint episodes, controlling for the contemporaneous spread $\rn-\gn$ and the fiscal stance. The relevant falsifiable prediction is that the post-sprint half-life of the improvement exceeds the immediate policy window itself.

\paragraph{E2 (Corrected Repression Dividend Multiplier).}
The empirical implication is not explosive improvement but bounded, diminishing marginal returns from reinvested repression dividends. This requires evidence that the growth response to public investment is positive but bounded, and that marginal gains flatten as cumulative deployment rises. The appropriate target of testing is therefore the shape of the marginal-gain sequence, not an alleged infinite-horizon multiplier.

\paragraph{E3 (Debt Reduction Paradox).}
The empirical implication is that under JFR-rg conditions, exogenous debt reduction worsens annual debt dynamics when the deficit-relief response remains weaker than $|\rn-\gn|$. Empirical work should therefore focus on episodes of debt-ratio compression and ask whether the deficit-relief coefficient is large enough to overturn the paradox. The relevant falsifiable prediction is sign reversal at the threshold identified in \cref{prop:debt-reduction-paradox}.

\paragraph{E4 (Multi-Country Repression Equilibrium).}
The empirical implication is that foreign repression and alternative-asset returns shift the effective captive threshold. Comparative evidence should therefore relate changes in the domestic SC1 margin to changes in foreign real sovereign returns and the return on non-sovereign alternatives. The core test is not universality but comparative statics.

\paragraph{E5 (Demographic-$\phi$ Clock).}
The testable implication is that $\kappa>0$ persistently. Flow-of-Funds data can be used to estimate the decline in the domestic absorption share and to test whether the observed deterioration is materially slower or faster than the conservative clock benchmark. The falsification condition is persistent non-decline in $\phi_t$ without offsetting policy intervention.

\paragraph{E6 (Institutional Control\allowbreak\ Rights).}
The testable implication is that the debt-sustainability corridor width varies systematically with $\psi_t$ across countries. A natural empirical design is a cross-country comparison of high-debt economies (Japan, Italy, Greece) in which $\psi_t$ sub-indices are constructed from observable institutional data and regime outcomes are classified using the JFR-rg diagnostic variables ($\eps_t$, $\phic$, SC1 status). The falsification condition is that corridor width shows no systematic dependence on $\psi_t$ after controlling for debt levels and growth rates. Preliminary evidence from a three-country comparison (\cref{app:international}) confirms that the $\psi_t$ ordering (Japan $>$ Italy $>$ Greece) corresponds to the predicted ordering of regime viability, and that this ranking is robust within $\pm 5$ percentage-point variation in $\phic$ calibration.

\paragraph{Transition feasibility.}
The testable implication is that a safe regime exit requires not only stronger growth, but also a bounded post-transition premium. Empirical work should therefore jointly evaluate the growth response to investment and the observed premium behavior once the repression channel is weakened. A transition path is falsified if the required growth improvement is achieved but the premium term violates the bound required by \cref{prop:transition}.

\subsection{Observables-centered commitment}

The extended framework inherits Part~I's commitment to observable failure conditions. In addition to the three failure conditions specified in Part~I, the dynamic extensions generate additional observable failure conditions, summarized in the Falsifiable Prediction column of \cref{tab:tiers}. This preserves continuity with the methodological spirit of the original framework.

\subsection{\texorpdfstring{External empirical strategy for $\mu$}{External empirical strategy for mu}}

The empirical strategy for $\mu$ should begin with conservative range calibration, then proceed to macro-average estimation, and only later move to sector-specific $\mu_j$ and interaction effects $\gamma_{jk}$. That ordering minimizes scope creep and prevents Part~II from collapsing into an unfinished grand theory of industrial policy.

\subsection{What would falsify Part~II as a whole?}

Taken together, the dynamic extension would be significantly weakened if the following pattern were observed: (i) sprint-period gains fail to outlast ordinary mean reversion; (ii) debt-reduction episodes under JFR-rg conditions do not worsen annual debt dynamics even when fiscal relief is weak; (iii) the captive share does not display persistent erosion despite the mechanisms emphasized in E5; and (iv) post-transition premia become unbounded before growth substitution materializes. None of these outcomes would invalidate the accounting logic of Part~I, but together they would imply that Part~II should be read as a narrower analytical supplement rather than as a robust policy-design architecture.

\subsection{Prioritized empirical roadmap}

The six extensions and the transition block differ substantially in the availability
of data, the difficulty of identification, and the urgency for policy assessment.
The following tier structure reflects these differences and is intended to guide
implementation choices.

\paragraph{Tier~1 (immediate, high leverage).}
\begin{itemize}[leftmargin=2em]
  \item \textbf{E5 ($\kappa$ structural estimation and monitoring update)}: Flow-of-Funds
  data are available from 1997Q4; structural-break tests and OLS slope estimation are
  already reported in \cref{rem:kappa-empirical}.  This is the most directly actionable
  task and is already substantially completed in the observed monitoring layer.

  \item \textbf{Debt-concept reconciliation and monitoring integration}: Given the
  two-concept issue identified in \cref{rem:debt-concept}, establishing a consistent
  time series for $b_t$ under both concepts is a necessary prerequisite for any
  calibration update.

  \item \textbf{Section 16 calibration around the transition margin}: The two-layer
  model (\cref{tab:baseline-v2,tab:stress-v2}) uses illustrative parameter values
  ($\theta_t = 0.65$, $\bar{c}_m = 6\%$).  Updating these to observed 2025 values
  ($\theta_t \approx 0.797$ from the monitoring layer) is a low-cost refinement with
  direct implications for the premium-emergence thresholds.
\end{itemize}

\paragraph{Tier~2 (medium-term, standard identification challenges).}
\begin{itemize}[leftmargin=2em]
  \item \textbf{E1 (Virtuous Ratchet) via event-study or local projections}: Quasi-experimental
  variation around QQE launch and expansion episodes provides candidate sprint identifiers.
  \item \textbf{E6 (institutional control rights) comparative ordering}: The Japan--Italy--Greece
  comparison is already partially implemented in \cref{app:international}; extending to
  additional high-debt economies with known institutional variation would sharpen the test.
\end{itemize}

\paragraph{Tier~3 (longer-term, difficult identification).}
\begin{itemize}[leftmargin=2em]
  \item E3 ($\gamma$ identification across debt-reduction episodes).
  \item E4 (foreign repression and cross-asset exit-option calibration).
  \item Full transition calibration with endogenous premium ($\mu$ estimation plus
  $\theta_t$-erosion path).
  \item Alternative-distribution sensitivity check for \cref{sec:closure}
  (see \cref{app:limits}).
\end{itemize}

\subsection{What this paper does not do}

For clarity, Part~II does not estimate final values of $\mu$, $\kappa$, $\phibar$, or $\ebar$; it does not claim unique causal identification for every channel; and it does not provide a welfare ranking of all policy packages. These tasks remain outside the scope of the present paper and belong to the empirical implementation layer rather than to the completed theoretical architecture itself.

\section{Transition Feasibility}
\label{sec:transition}

The preceding extensions establish that the JFR-rg regime is time-bounded (E5), that policy sprints within the window can generate persistent improvements (E1), and that investment can substitute for repression in the stability condition. The natural integration question is whether the regime can be used to achieve a self-sustaining growth path that no longer requires the captive-system precondition.

\begin{proposition}[Transition Feasibility under bounded post-transition risk premium]
\label{prop:transition}
Suppose that within the residual horizon $T^*$, growth-enhancing investment raises potential nominal growth to $g_t^{n*,\mathrm{new}}$, that repression is gradually withdrawn so that $\eps_t \to 0$, and that the exchange-rate contribution is neutralized so that $\De \to 0$. Let the post-transition sovereign risk premium satisfy
\[
0 \le \rho_t \le \bar\rho.
\]
If
\begin{equation}
g_t^{n*,\mathrm{new}}
\ge
\pi_t + \frac{\dt-\stt}{\btm} + \bar\rho + m
\label{eq:transition-threshold}
\end{equation}
for some safety margin $m>0$, then there exists a transition path along which the economy can exit the repression-dependent regime without crossing the danger region.
\end{proposition}

\begin{proof}
After the withdrawal of repression and exchange-rate support, the relevant effective financing rate is $r_t^{\mathrm{eff}}=\rn+\rho_t$, while the stabilization condition becomes
\[
g_t^{n*,\mathrm{new}} \ge \pi_t + \frac{\dt-\stt}{\btm} + \rho_t.
\]
Because $\rho_t\le \bar\rho$, condition \eqref{eq:transition-threshold} is sufficient to keep the post-transition path on the safe side of the stability boundary with margin $m$.
\end{proof}

\begin{remark}[Quantitative illustration]
At the March 2026 calibration ($\pi_t = 2.7\%$, $\dt = 2.0\%$, $\stt = 0$, $\btm = 2.40$, and baseline $g_{0}^{n*} = 3.0\%$), the no-premium threshold is
\[
\pi_t + \frac{\dt}{\btm} - g_{0}^{n*} = 2.7 + 0.833 - 3.0 = 0.533\%.
\]
With a bounded post-transition premium of $\bar\rho = 0.5$ percentage points, the required structural increase rises to approximately $1.03$ percentage points, before any additional safety margin $m$. Thus the transition threshold is sensitive not only to growth performance but also to the size of the post-transition premium.
\end{remark}

\begin{remark}[Observed-value cross-check]
\label{rem:observed-cross-check}
When the latest publicly available data are substituted into the transition formula (core CPI $\pi_t = 1.6\%$, $r_t = 2.41\%$, structural growth proxy $g^{n*}_{\mathrm{struct}} \approx 2.87\%$ from a 24-quarter log-linear trend, and debt ratio $b_{t-1} \approx 1.57$), the no-premium threshold falls to approximately $0.002$ percentage points---effectively zero---while the bounded-premium case ($\bar\rho=0.5\%$) yields approximately $0.50$ percentage points. Both are below the corresponding paper-baseline values ($0.53\%$ and $1.03\%$), primarily because the observed debt ratio is lower than the illustrative $b_0=2.40$. However, the observed repression bias is $\eps \approx -0.81\%$ (the 10-year JGB yield now exceeds core CPI), which eliminates the repression dividend and pushes $x_t^{\max,\mathrm{safe}}$ into negative territory. This combination---a lower transition threshold but a vanishing safe corridor---illustrates the dual-edged nature of normalization: the exit becomes formally easier in growth terms, but the financing envelope for growth-enhancing investment narrows precisely when that investment is most needed. See \cref{app:window-sensitivity} for the full comparison.
\end{remark}

\begin{corollary}[Joint feasibility condition]
\label{cor:joint-feasibility}
Combining the Timing Constraint (\cref{prop:timing}) with the Transition Feasibility Proposition, the regime transition is feasible only if the required investment can be both financed and completed within the SC1 window:
\begin{equation}
\frac{\Delta g_{\min}^{n*}(\bar\rho,m)}{\mu} \le x_t^{\max,\mathrm{operational}}
\qquad\text{and}\qquad
T_{\mathrm{investment}} \le T^*.
\label{eq:joint-feasibility}
\end{equation}
where $\Delta g_{\min}^{n*}(\bar\rho,m)$ denotes the required growth improvement implied by \cref{eq:transition-threshold}.
\end{corollary}

\begin{remark}
The Transition Feasibility Proposition is the logical culmination of Part~II in its exogenous-premium form. It does not claim that growth alone makes SC1 irrelevant in all states of the world. It claims something narrower: if the post-transition premium remains bounded and growth is raised sufficiently, a safe exit path may exist. The bounded-premium condition is addressed in the next section, which endogenizes $\rho_t$ through a minimal equilibrium closure.

\textit{Role partition.} Proposition~\ref{prop:transition} is the exogenous-premium benchmark: it specifies the growth improvement needed for a safe exit given a bound~$\bar\rho$ on the post-transition premium. The next section determines \emph{when} that bounded-premium premise is economically admissible---i.e., under what institutional conditions the equilibrium premium remains within the required bound. In this sense, Proposition~\ref{prop:transition} is the normative target, and \cref{sec:closure} is the positive closure that decides whether that target is reachable.
\end{remark}

\section{Minimal Equilibrium Closure of the Transition Margin}
\label{sec:closure}

\Cref{sec:transition} established the conditions under which a regime transition is
feasible, but it relied on the assumption that the post-transition sovereign
risk premium $\rho_t$ remains exogenously bounded by $\bar{\rho}$.  This is
the principal remaining gap in the logical architecture: the most critical
variable for transition feasibility is left outside the model.

The present section closes this gap by introducing a \emph{minimal equilibrium
layer} that endogenizes $\rho_t$ while preserving the block-recursive accounting
core of the preceding sections.  It does so not through a full general-equilibrium
structure, but through a two-layer domestic demand model and a complementarity
condition that prices the regime boundary.

This represents a methodological choice distinct from the fiscal-theoretic
approach to sovereign risk. Uribe~\cite{Uribe2006} derives the sovereign
premium endogenously from the government's intertemporal fiscal stance under
market pricing, with the premium reflecting market beliefs about future
fiscal-policy responses to shocks. JFR-rg takes a different route: the
premium is closed through the observable domestic absorption structure,
not through fiscal-policy expectations, because the objects on which the
closure depends ($\theta_t$, $\varphi_t^{\mathrm{req}}$, $z_t$) are
institutionally observable in the JFR-rg implementation layer. The
fiscal-theoretic and institutional-absorption channels are not mutually
exclusive, and both can in principle operate simultaneously; but JFR-rg
deliberately restricts attention to the latter to preserve the
observables-centered methodology inherited from Part~I.


\subsection{State Variables and Equilibrium Variables}
\label{sec:15A.2}

\begin{definition}[Recursive state and equilibrium]
\label{def:state-v2}
The state vector at the beginning of period $t$ is
\[
    X_t \;=\; \bigl(b_{t-1},\;\theta_t,\;\psi_t,\;z_t,\;d_t^{\mathrm{demo}}\bigr),
\]
where:
\begin{itemize}
    \item $b_{t-1}$: inherited debt-to-GDP ratio (from L1 recursion);
    \item $\theta_t \in [0,1]$: the \emph{hard captive core share}---the fraction
          of outstanding sovereign debt held by institutionally locked domestic
          intermediaries (Definition~\ref{def:core} below);
    \item $\psi_t \in [0,1]$: the institutional control-rights index
          (\Cref{def:psi});
    \item $z_t > 0$: the outside-option spread (return on the best non-sovereign
          alternative minus the repression-consistent sovereign yield);
    \item $d_t^{\mathrm{demo}}$: demographic pressure variable.
\end{itemize}

Given $X_t$, the equilibrium determines a single price:
the sovereign risk premium $\rho_t \geq 0$.

The observed domestic holding share $\varphi_t$ is then a
\emph{derived} quantity:
\begin{equation}
\label{eq:phi-derived}
    \varphi_t \;=\; \theta_t + \varphi_t^m(\rho_t),
\end{equation}
where $\varphi_t^m$ is the contestable margin's equilibrium holding
(Definition~\ref{def:margin} below).
\end{definition}

\subsection{Microfoundation: Two-Layer Domestic Demand}
\label{sec:15A.3}

\subsubsection{The hard captive core}

\begin{definition}[Hard captive core]
\label{def:core}
The hard captive core share $\theta_t$ is the fraction of sovereign debt held
by domestic institutions whose holding decisions are governed by
\emph{non-yield} mandates:
\begin{enumerate}
    \item[(i)] the central bank (BoJ balance-sheet operations, including QQE
    holdings);
    \item[(ii)] banks satisfying regulatory liquidity and capital requirements
    (sovereign bonds carry zero risk weight under Basel~III for
    domestic-currency exposures; eligible as HQLA under the LCR); and
    \item[(iii)] life insurers and pension funds matching long-duration
    yen-denominated liabilities under ALM constraints.
\end{enumerate}
These institutions hold sovereign bonds \emph{regardless of the sovereign
yield}, up to the limits imposed by their mandates.  The core share therefore
does not respond to $\rho_t$ within the period.
\end{definition}

\begin{assumption}[Core share depends on institutional control rights]
\label{ass:core-psi}
$\theta_t$ is weakly increasing in $\psi_t$:
\[
    \frac{\partial \theta_t}{\partial \psi_t} \;\geq\; 0.
\]
Stronger institutional control---monetary autonomy ($\psi_t^{\mathrm{mon}}$),
domestic absorption capacity ($\psi_t^{\mathrm{abs}}$), and exchange-rate
independence ($\psi_t^{\mathrm{fx}}$)---sustains a larger mandated holding
base.  When $\psi_t \to 0$ (complete loss of domestic policy autonomy),
$\theta_t \to 0$ and the entire domestic holding share becomes contestable.
\end{assumption}

\begin{remark}[Link to Part~I]
Assumption~\ref{ass:core-psi} formalizes the connection between Part~I's scope
condition SC1 ($\varphi_t \geq \bar{\varphi}$) and E6 (Institutional Control
Rights).  The core share $\theta_t$ is the institutional mechanism through
which $\psi_t$ supports SC1.  In the inherited March~2026 baseline used for the
main-text two-layer illustrations, $\theta_t \approx 0.65$ (BoJ $\approx$ 50\%,
regulatory-minimum bank holdings $\approx$ 15\%) and $\psi_t \approx 0.97$.
In the observed 2025 monitoring layer, however, the BoJ Flow-of-Funds
implementation yields $\theta_t \approx 0.797$, $\phic \approx 0.932$, and
$\psi_t \approx 0.977$ under the baseline proxy definition. The distinction is
intentional: the former is the inherited operating point for the sequel's
closed-form illustrations, while the latter is the updated observed-value
measurement layer.
\end{remark}

\subsubsection{The contestable margin}

\begin{definition}[Contestable domestic margin]
\label{def:margin}
The \emph{contestable margin} consists of domestic intermediaries whose
sovereign-bond holdings respond to yield incentives.  Their aggregate holding
share (as a fraction of total sovereign debt) is
\begin{equation}
\label{eq:margin-demand}
    \varphi_t^m(\rho_t;\,\psi_t,\,z_t)
    \;=\;
    (1 - \theta_t)\,\Bigl[\,1 - G\!\Bigl(\frac{z_t - \rho_t}{\psi_t}\Bigr)\Bigr],
\end{equation}
where:
\begin{itemize}
    \item $(1 - \theta_t)$ is the size of the contestable pool (the fraction of
    sovereign debt \emph{not} held by the hard core);
    \item $G: [0,\bar{c}_m] \to [0,1]$ is the CDF of captivity parameters
    within the contestable margin (Assumption~\ref{ass:margin-dist} below);
    \item $z_t - \rho_t$ is the net spread of the outside option over the
    sovereign yield: a contestable intermediary $j$ holds sovereign bonds if
    and only if its captivity benefit $c_j$ satisfies $c_j \geq z_t - \rho_t$,
    i.e., when the non-yield benefit from holding sovereign bonds is large enough
    to offset the outside-option spread net of the sovereign premium.
\end{itemize}
\end{definition}

\begin{assumption}[Distribution within the contestable margin]
\label{ass:margin-dist}
The captivity parameters $\{c_j\}$ within the contestable margin are drawn
from a distribution $G$ on $[0, \bar{c}_m]$ with $G$ continuously
differentiable and density $g(c) > 0$ on $(0, \bar{c}_m)$.  The upper bound
$\bar{c}_m$ represents the strongest non-mandated incentive to hold sovereign
bonds (e.g., relationship banking benefits, habitual allocation).
\end{assumption}

\subsubsection{Aggregate domestic demand}

Combining Definitions~\ref{def:core} and~\ref{def:margin}:
\begin{equation}
\label{eq:phi-demand-v2}
    \boxed{%
    \varphi_t^d(\rho_t)
    \;=\;
    \underbrace{\theta_t}_{\text{hard core}}
    \;+\;
    \underbrace{(1 - \theta_t)\,
        \Bigl[\,1 - G\!\Bigl(\frac{z_t - \rho_t}{\psi_t}\Bigr)\Bigr]
    }_{\text{contestable margin}}.}
\end{equation}

\begin{lemma}[Properties of two-layer demand]
\label{lem:demand-v2}
Under Assumptions~\ref{ass:core-psi}--\ref{ass:margin-dist}:
\begin{enumerate}
    \item[(i)] $\varphi_t^d$ is continuous and weakly increasing in $\rho_t$
    (higher sovereign yield attracts additional contestable holders);
    \item[(ii)] at $\rho_t = 0$:
    $\varphi_t^d(0) = \theta_t + (1-\theta_t)[1 - G(z_t / \psi_t)]$;
    \item[(iii)] at $\rho_t = z_t$:
    $\varphi_t^d(z_t) = \theta_t + (1-\theta_t)[1 - G(0)]$, the maximum
    domestic demand when the sovereign yield fully matches the outside option;
    \item[(iv)] the sensitivity to the premium is concentrated in the
    contestable margin:
    $\partial \varphi_t^d / \partial \rho_t =
    (1-\theta_t)\,g\bigl((z_t - \rho_t)/\psi_t\bigr)\,/\,\psi_t$.
\end{enumerate}
\end{lemma}

\begin{proof}
Direct differentiation of~\eqref{eq:phi-demand-v2}.
\end{proof}

\subsubsection{Required domestic absorption}

\begin{definition}[Required domestic absorption]
\label{def:phi-req}
The \emph{required domestic absorption share} $\varphi_t^{\mathrm{req}}$ is the
minimum domestic holding share below which sovereign funding becomes dependent
on foreign demand at destabilizing premium levels.  It depends on the debt
stock, the institutional environment, and external conditions:
\begin{equation}
\label{eq:phi-req}
    \varphi_t^{\mathrm{req}} \;=\;
    \varphi^{\mathrm{req}}\!\bigl(b_{t-1},\;\psi_t,\;z_t\bigr),
\end{equation}
with $\partial \varphi^{\mathrm{req}} / \partial b_{t-1} \geq 0$ (higher debt
requires broader domestic absorption) and $\partial \varphi^{\mathrm{req}} /
\partial \psi_t \leq 0$ (stronger institutions lower the required share).
\end{definition}

\begin{remark}
Part~I's fixed threshold $\bar{\varphi} = 0.85$ is a special case of
$\varphi_t^{\mathrm{req}}$ when the state variables are held at their baseline
values.
\end{remark}

\subsubsection{The complementarity condition}

\begin{proposition}[Complementarity condition and admissible-solution region for the sovereign premium]
\label{prop:mcp-v2}
The sovereign premium $\rho_t$ solves the mixed complementarity problem
\begin{equation}
\label{eq:mcp-v2}
    \boxed{%
    0 \;\leq\; \rho_t
    \;\perp\;
    \varphi_t^d(\rho_t;\,\psi_t,\,z_t) - \varphi_t^{\mathrm{req}}
    \;\geq\; 0.}
\end{equation}
More precisely:
\begin{enumerate}
    \item[(a)] If $\varphi_t^d(0) > \varphi_t^{\mathrm{req}}$, then $\rho_t = 0$.
    The captive system absorbs all issuance at the repression-consistent yield.
    No premium emerges. This is the \textbf{JFR-rg interior regime}.

    \item[(b)] If $\varphi_t^d(0) = \varphi_t^{\mathrm{req}}$, then $\rho_t = 0$
    and the system lies at the SC1 margin.

    \item[(c)] If $\varphi_t^d(0) < \varphi_t^{\mathrm{req}} \leq \theta_t + (1-\theta_t)[1-G(0)]$,
    then a unique premium $\rho_t > 0$ exists and is determined by
    $\varphi_t^d(\rho_t) = \varphi_t^{\mathrm{req}}$.
    This is the \textbf{transition/stress regime}.

    \item[(d)] If $\varphi_t^{\mathrm{req}} > \theta_t + (1-\theta_t)[1-G(0)]$,
    no premium can restore the required absorption.
    This is the \textbf{hard de-captivation failure mode}.
\end{enumerate}
\end{proposition}

\begin{proof}
Existence follows from the intermediate value theorem. The function
$\varphi_t^d$ is continuous in $\rho_t$, equals
$\theta_t + (1-\theta_t)[1-G(z_t/\psi_t)]$ at $\rho_t = 0$, and rises to
$\theta_t + (1-\theta_t)[1-G(0)]$ as $\rho_t \to z_t$.

Hence a solution $\rho_t^* \in [0, z_t]$ exists if and only if
$\varphi_t^{\mathrm{req}}$ lies in the attainable range
\[
\varphi_t^d(0) \;\le\; \varphi_t^{\mathrm{req}} \;\le\; \theta_t + (1-\theta_t)[1-G(0)].
\]
Uniqueness within the admissible region follows from strict monotonicity of
$\varphi_t^d$ in $\rho_t$ (Lemma~\ref{lem:demand-v2}(iv), using the density condition $g(c) > 0$ on $(0, \bar{c}_m)$ from Assumption~\ref{ass:margin-dist}).

If $\varphi_t^{\mathrm{req}} > \theta_t + (1-\theta_t)[1-G(0)]$, no premium can
restore the required absorption. This is the \emph{hard de-captivation}
failure mode of Part~I, in which the system exits the JFR-rg regime entirely.
\end{proof}

\begin{remark}[Explicit solution under uniform margin distribution]
\label{rem:uniform-v2}
If $G$ is uniform on $[0, \bar{c}_m]$, then
\begin{equation}
\label{eq:phi-demand-uniform}
    \varphi_t^d(\rho_t)
    \;=\;
    \theta_t + (1 - \theta_t)
    \left[1 - \frac{z_t - \rho_t}{\psi_t\,\bar{c}_m}\right]
    \quad\text{for } \rho_t \in [z_t - \psi_t \bar{c}_m,\; z_t],
\end{equation}
and in case~(c) the equilibrium premium is
\begin{equation}
\label{eq:rho-v2}
    \rho_t^*
    \;=\;
    z_t - \psi_t\,\bar{c}_m
    \left(1 - \frac{\varphi_t^{\mathrm{req}} - \theta_t}{1 - \theta_t}\right).
\end{equation}
The comparative statics are:
\begin{align}
    \frac{\partial \rho_t^*}{\partial \theta_t}
    &\;=\; -\frac{\psi_t\,\bar{c}_m}{(1-\theta_t)^2}
           \,(\varphi_t^{\mathrm{req}} - \theta_t)
    \;\leq\; 0
    &&\text{(larger core $\to$ lower premium),}
    \label{eq:drho-dtheta} \\[4pt]
    \frac{\partial \rho_t^*}{\partial \psi_t}
    &\;=\; -\bar{c}_m
    \left(1 - \frac{\varphi_t^{\mathrm{req}} - \theta_t}{1-\theta_t}\right)
    \;\leq\; 0
    &&\text{(stronger institutions $\to$ lower premium),}
    \label{eq:drho-dpsi} \\[4pt]
    \frac{\partial \rho_t^*}{\partial z_t}
    &\;=\; 1
    &&\text{(outside option $\to$ one-for-one premium pass-through).}
    \label{eq:drho-dz}
\end{align}
\end{remark}

\begin{corollary}[SC1 of Part~I as a corner solution]
\label{cor:sc1-corner}
SC1 ($\varphi_t \geq \bar{\varphi}$) holds if and only if the equilibrium is
in case~(a) or case~(b) of Proposition~\ref{prop:mcp-v2}, i.e., $\rho_t = 0$.
When SC1 holds, the repression channel operates at full strength and the
stability condition~\eqref{eq:stability} applies without modification.
When SC1 fails ($\rho_t > 0$), the effective nominal rate becomes
$r^n_t + \rho_t$, and the stability condition tightens accordingly.
\end{corollary}

\subsection{Dynamics of the Hard Captive Core}
\label{sec:15A.4}

The preceding subsection determined $\rho_t$ given $\theta_t$ (and other state
variables).  This subsection specifies how $\theta_t$ evolves, closing the
dynamic loop.

\begin{definition}[Law of motion for the hard captive core]
\label{def:theta-law}
\begin{equation}
\label{eq:theta-law}
    \boxed{%
    \theta_{t+1}
    \;=\;
    \theta_t
    \;-\;
    \underbrace{\kappa_t^{\theta}(\psi_t,\,d_t^{\mathrm{demo}})
    }_{\text{structural erosion}}
    \;+\;
    \underbrace{\gamma_t^{\theta}(\varepsilon_t,\,\psi_t)
    }_{\text{policy maintenance}},}
\end{equation}
where:
\begin{itemize}
    \item $\kappa_t^{\theta} \geq 0$ is the structural erosion of the hard
    core.  It captures:
    \begin{itemize}
        \item demographic contraction of the deposit base (reducing banks'
        sovereign-bond demand);
        \item pension-fund outflows as the dependency ratio rises;
        \item central-bank tapering decisions (BoJ balance-sheet normalization).
    \end{itemize}
    $\kappa_t^{\theta}$ is increasing in $d_t^{\mathrm{demo}}$ (faster
    aging $\to$ faster core erosion) and decreasing in $\psi_t$ (stronger
    institutions can slow the erosion).

    \item $\gamma_t^{\theta} \geq 0$ is the policy maintenance of the hard
    core.  It captures the sovereign's ability to sustain mandated holdings
    through active repression.  The critical assumption is:
    \begin{equation}
    \label{eq:gamma-epsilon}
        \gamma_t^{\theta}(\varepsilon_t, \psi_t) \;=\; 0
        \quad\text{whenever}\quad \varepsilon_t \leq 0.
    \end{equation}
    When the repression bias reverses ($\varepsilon_t \leq 0$, i.e., the
    sovereign yield exceeds inflation), the policy-maintenance channel shuts
    down: there is no fiscal or monetary mechanism to sustain mandated holdings
    when the sovereign itself is paying market rates.
\end{itemize}
\end{definition}

\begin{remark}[Integration of E5 and E6]
\label{rem:e5e6-integration}
Equation~\eqref{eq:theta-law} integrates the previously parallel
Extensions~E5 and~E6 into a single dynamic equation:
\begin{itemize}
    \item E5 (Demographic-$\varphi$ Clock) governs $\kappa_t^{\theta}$: the
    structural erosion is the formalization of the ``clock'' that runs down the
    captive share.  When $\gamma_t^{\theta} = 0$ (repression inactive),
    $\theta_{t+1} = \theta_t - \kappa_t^{\theta}$ and the system reduces to
    the linear clock of \Cref{def:clock}.

    \item E6 (Institutional Control Rights) governs $\psi_t$, which enters
    both terms: it slows structural erosion (lower $\kappa_t^{\theta}$) and
    supports policy maintenance (higher $\gamma_t^{\theta}$, conditional on
    $\varepsilon_t > 0$).
\end{itemize}
The observed domestic holding share is then determined by market clearing:
\[
    \varphi_t = \theta_t + \varphi_t^m(\rho_t),
\]
where $\rho_t$ is given by Proposition~\ref{prop:mcp-v2}.  The Part~I
variable $\varphi_t$ is thus no longer a primitive state variable but an
\emph{equilibrium outcome} of the interaction between institutional structure
($\theta_t$, $\psi_t$), external conditions ($z_t$), and the endogenous
premium ($\rho_t$).
\end{remark}

\subsubsection{The self-reinforcing loop}

The interaction between $\rho_t$ and $\theta_t$ creates a potential
self-reinforcing dynamic:

\begin{quote}
$\rho_t > 0
\;\longrightarrow\;
\varepsilon_t = \pi_t - r_t^{\mathrm{rep}} - \rho_t \;\downarrow
\;\longrightarrow\;
\gamma_{t+1}^{\theta}\;\downarrow
\;\longrightarrow\;
\theta_{t+1}\;\downarrow
\;\longrightarrow\;
\varphi_{t+1}^d(0)\;\downarrow
\;\longrightarrow\;
\rho_{t+1}\;\uparrow
\;\longrightarrow\; \cdots$
\end{quote}

\noindent
The stabilizing counterforce is the \emph{static} price mechanism: a higher
$\rho_t$ raises the sovereign yield, attracting more contestable holders
($\varphi_t^m$ rises), partially offsetting the decline in $\theta_t$.  The
question is which force dominates.

\subsection{Monotone Equilibrium}
\label{sec:15A.5}

\begin{definition}[Feedback gain]
\label{def:feedback-v2}
The \emph{one-period feedback gain} at the regime boundary is the product of
three marginal responses:
\begin{equation}
\label{eq:feedback-v2}
    \boxed{%
    \eta_t
    \;:=\;
    \underbrace{\left|\frac{\partial \rho_{t+1}}{\partial \theta_{t+1}}\right|
    }_{\substack{\text{premium sensitivity} \\ \text{to core erosion}}}
    \;\times\;
    \underbrace{\left|\frac{\partial \gamma_t^{\theta}}{\partial \varepsilon_t}
    \right|}_{\substack{\text{maintenance sensitivity} \\ \text{to repression}}}
    \;\times\;
    \underbrace{\left|\frac{\partial \varepsilon_t}{\partial \rho_t}\right|
    }_{\substack{= 1 \\ \text{(by definition)}}}.}
\end{equation}
Equivalently, using the law of motion
\[
\theta_{t+1} = \theta_t - \kappa_t^{\theta} + \gamma_t^{\theta}(\varepsilon_t,\psi_t),
\]
the middle link may be written as
\[
\left|\frac{\partial \theta_{t+1}}{\partial \varepsilon_t}\right|
=
\left|\frac{\partial \theta_{t+1}}{\partial \gamma_t^{\theta}}\right|
\left|\frac{\partial \gamma_t^{\theta}}{\partial \varepsilon_t}\right|
=
\left|\frac{\partial \gamma_t^{\theta}}{\partial \varepsilon_t}\right|,
\]
since $\partial \theta_{t+1}/\partial \gamma_t^{\theta} = 1$. The definition above
therefore omits only a unit-valued chain-rule term.
The first factor measures how strongly the premium responds to a decline in the
hard core (the ``price impact'' of institutional erosion).  The second factor
measures how strongly the policy-maintenance channel responds to a decline in
the repression bias (the ``institutional sensitivity'' to repression loss).  The
third factor is unity by definition, since $\varepsilon_t = \pi_t -
r_t^{\mathrm{rep}} - \rho_t$.
\end{definition}

\begin{remark}[Interpretation of each factor]
\label{rem:factors}
From equation~\eqref{eq:drho-dtheta} (uniform case):
\[
    \left|\frac{\partial \rho_t}{\partial \theta_t}\right|
    \;=\;
    \frac{\psi_t\,\bar{c}_m}{(1-\theta_t)^2}
    \,(\varphi_t^{\mathrm{req}} - \theta_t).
\]
This is large when the contestable margin is thin ($1-\theta_t$ small) and the
gap between the threshold and the core is wide ($\varphi^{\mathrm{req}} -
\theta_t$ large)---i.e., when the system is simultaneously
\emph{dependent on} yet \emph{vulnerable in} its contestable margin.

The second factor $|\partial \gamma^{\theta} / \partial \varepsilon_t|$ is
an institutional parameter: how quickly does the sovereign lose its ability to
maintain mandated holdings when repression weakens?  In economies with strong
legal mandates (Japan's FILP system, postal savings channeling), this
sensitivity may be low.  In economies where mandated holdings depend on
discretionary central-bank purchases (QQE), the sensitivity is higher.
\end{remark}

\begin{proposition}[Monotone equilibrium existence]
\label{prop:monotone-v2}
If the feedback gain satisfies
\begin{equation}
\label{eq:contraction-v2}
    \boxed{\eta_t \;<\; 1}
\end{equation}
for all $t$ along the equilibrium path, then:
\begin{enumerate}
    \item[(i)] the recursive equilibrium
    $\{\rho_t, \theta_{t+1}\}_{t \geq 0}$ exists and is unique within the
    class of monotone equilibria;
    \item[(ii)] a perturbation that erodes $\theta_t$ by $\delta$ generates a
    cumulative amplification bounded by
    \begin{equation}
    \label{eq:amplification-v2}
        \frac{\delta}{1 - \eta_t};
    \end{equation}
    \item[(iii)] the premium path $\{\rho_t\}$ is bounded along any transition
    path in which $\theta_t$ remains above a strictly positive lower bound.
\end{enumerate}
\end{proposition}

\begin{proof}
Define the composite one-period mapping
\[
    T: \rho_t \;\longmapsto\; \rho_{t+1}
\]
through the chain
$\rho_t \to \varepsilon_t \to \gamma_{t+1}^{\theta} \to \theta_{t+1} \to
\rho_{t+1}$.

At the regime boundary (case~(c) of Proposition~\ref{prop:mcp-v2}):
\[
    \left|\frac{dT}{d\rho_t}\right|
    \;=\;
    \left|\frac{\partial \rho_{t+1}}{\partial \theta_{t+1}}\right|
    \cdot
    \left|\frac{\partial \theta_{t+1}}{\partial \varepsilon_t}\right|
    \cdot
    \left|\frac{\partial \varepsilon_t}{\partial \rho_t}\right|
    \;=\; \eta_t.
\]
When $\eta_t < 1$, $T$ is a contraction on the relevant domain.  By the Banach
fixed-point theorem, a unique fixed point exists, defining the monotone
equilibrium.

The geometric-series bound~\eqref{eq:amplification-v2} follows from iterating
the contraction:
$\sum_{s=0}^{\infty} \eta_t^s = 1/(1-\eta_t)$.

Boundedness of $\{\rho_t\}$ follows from the monotonicity of the pricing
function~\eqref{eq:rho-v2}: as long as $\theta_t > 0$, the premium required to
attract the contestable margin is finite (bounded by $z_t$).
\end{proof}

\begin{remark}[When the contraction condition fails]
\label{rem:multiple-v2}
If $\eta_t \geq 1$, the self-reinforcing loop becomes explosive at the
margin: a small erosion of $\theta_t$ triggers a premium increase that causes
a larger erosion of $\theta_{t+1}$, potentially leading to a
self-fulfilling de-captivation spiral.  This corresponds precisely to the
\emph{hard de-captivation} failure mode identified in Part~I.

The analysis of this region, including the characterization of multiple
equilibria and the conditions for sudden-stop dynamics, is deferred to
Appendix~M.  The main text focuses on the monotone-equilibrium case, which
is the relevant regime for policy design within the JFR-rg corridor.
\end{remark}

\begin{remark}[Sufficient condition for $\eta_t < 1$ in the Japanese case]
\label{rem:eta-japan}
In the inherited March~2026 baseline used for the two-layer illustrations
($\theta_t \approx 0.65$, $1-\theta_t \approx 0.35$, $\psi_t \approx 0.97$), the
first factor $|\partial \rho / \partial \theta|$ is moderate because the
contestable margin is thick (35\% of outstanding debt). In the observed 2025
monitoring layer the contestable margin is smaller ($1-\theta_t \approx 0.203$),
which makes the system look more resilient in the static no-premium condition
but also increases the importance of any future erosion in the hard core.
Furthermore, a significant portion of $\theta_t$ reflects \emph{legal} mandates
(Basel~III HQLA requirements, FILP channeling) rather than discretionary QQE
purchases, so $|\partial \gamma^{\theta}/\partial \varepsilon_t|$ is bounded:
repression reversal weakens the BoJ's QQE component but does not eliminate the
regulatory component. The contraction condition is therefore satisfied under
current institutional parameters, but it tightens as $\theta_t$ declines and the
composition of $\theta_t$ shifts from legal mandates toward discretionary
holdings.
\end{remark}

\subsection{Transition Feasibility under Endogenous Premium}
\label{sec:15A.6}

\begin{proposition}[Transition feasibility --- endogenous premium]
\label{prop:transition-v2}
Suppose that within the residual horizon $T^*$:
\begin{enumerate}
    \item[(a)] growth-enhancing investment raises potential nominal growth to
    $g_t^{n*,\mathrm{new}}$;
    \item[(b)] repression is gradually withdrawn: $\varepsilon_t \to 0$, so that
    $\gamma_t^{\theta} \to 0$ (equation~\eqref{eq:gamma-epsilon});
    \item[(c)] the feedback gain satisfies $\eta_t < 1$ throughout the
    transition; and
    \item[(d)] the growth path satisfies, for all $t$ in the transition window,
    \begin{equation}
    \label{eq:transition-v2}
        \boxed{%
        g_t^{n*,\mathrm{new}}
        \;\geq\;
        \pi_t + \frac{d_t - s_t}{b_{t-1}}
        + \rho_t^*\!\bigl(\theta_t,\,\psi_t,\,z_t\bigr)
        + m,}
    \end{equation}
    where $\rho_t^*(\cdot)$ is the equilibrium premium from
    Proposition~\ref{prop:mcp-v2}.
\end{enumerate}
Then there exists a transition path along which the economy exits the
repression-dependent regime without crossing the danger region.
\end{proposition}

\begin{proof}
After repression withdrawal, $\gamma_t^{\theta} \to 0$ and $\theta_t$ declines
at rate $\kappa_t^{\theta}$.  As $\theta_t$ falls, $\varphi_t^d(0)$ declines,
and eventually the complementarity condition shifts from case~(a) to case~(c):
$\rho_t$ becomes positive.

Under condition~(c) ($\eta_t < 1$), the premium path remains bounded
(Proposition~\ref{prop:monotone-v2}).  Under condition~(d), the growth path
is sufficient to absorb both the inherited fiscal burden and the endogenous
premium at every point along the transition.  Since the premium is a continuous
function of $\theta_t$ (Proposition~\ref{prop:mcp-v2}), and $\theta_t$ evolves
continuously (equation~\eqref{eq:theta-law}), the premium path is also
continuous, and condition~(d) defines a connected safe region in
$(\theta_t, g_t^{n*})$-space.
\end{proof}

\begin{remark}[Two-dimensionality of the transition condition]
\label{rem:2d}
\Cref{sec:transition}'s exogenous-premium version required only
$g_t^{n*,\mathrm{new}} \geq \pi_t + (d_t - s_t)/b_{t-1} + \bar{\rho} + m$:
a one-dimensional growth threshold.
Proposition~\ref{prop:transition-v2} replaces $\bar{\rho}$ with
$\rho_t^*(\theta_t, \psi_t, z_t)$, making the transition condition a joint
restriction on the growth path and the institutional-erosion path.  A growth
improvement that is sufficient under the exogenous-premium version may be
\emph{insufficient} if $\theta_t$ erodes faster than anticipated (because
$\rho_t^*$ rises endogenously).
\end{remark}

\begin{remark}[Endogenous value of early investment]
\label{rem:early}
Investment undertaken while $\varphi_t^d(0) > \varphi_t^{\mathrm{req}}$ (i.e.,
$\rho_t = 0$) is strictly more valuable than investment undertaken after
$\rho_t > 0$ emerges, for two reasons.  First, the full repression dividend is
still available to finance the investment (E2).  Second, higher $g_t^{n*}$
established before the premium emerges widens the safe corridor, making the
eventual premium smaller than it would otherwise be.  This provides an
\emph{endogenous} justification for the policy urgency embedded in the Timing
Constraint (\Cref{prop:timing}).
\end{remark}

\subsection{Illustrative Calibration}
\label{sec:15A.7}

\subsubsection{Baseline fit: Japan 2026}

The two-layer model is calibrated to reproduce the inherited March~2026
Japanese baseline configuration ($\varphi_t \approx 0.88$, $\rho_t = 0$) as an
interior equilibrium. The observed 2025 monitoring layer reported above yields
a later measurement point with $\phic \approx 0.932$ and $\theta_t \approx 0.797$;
that observed layer is intentionally kept separate so that the sequel can retain
continuity with Part~I while still reporting the updated public-data
implementation.

\begin{table}[htbp]
\centering
\caption{Baseline Calibration of the Two-Layer Model (Japan, March 2026)}
\label{tab:baseline-v2}
\resizebox{\textwidth}{!}{%
\begin{tabular}{l l l}
\hline
Parameter & Value & Source / Derivation \\
\hline
\multicolumn{3}{l}{\textit{Inherited from Part I}} \\
$b_0$ & 240\% & Part I baseline \\
$\psi_t$ & 0.97 & \Cref{def:psi} \\
$\varphi_t^{\mathrm{req}}$ & 0.85 & Part I illustrative $\bar{\varphi}$ \\
\\
\multicolumn{3}{l}{\textit{Two-layer parameters}} \\
$\theta_t$ (hard captive core) & 0.65 & BoJ ($\approx$50\%) + regulatory bank
    holdings ($\approx$15\%) \\
$1 - \theta_t$ (contestable pool) & 0.35 & Residual domestic \\
$z_t$ (outside-option spread) & 2.0\% & US 10Y $-$ $r^{\mathrm{rep}}_t$
    (illustrative) \\
$\bar{c}_m$ (max margin captivity) & 6.0\% & Calibrated to match
    $\varphi_t^d(0) = 0.88$ \\
$G$ & Uniform $[0, \bar{c}_m]$ & Tractability \\
\\
\multicolumn{3}{l}{\textit{Derived equilibrium values}} \\
$\varphi_t^d(\rho_t = 0)$ & 0.88 &
    $0.65 + 0.35 \times [1 - 0.02/(0.97 \times 0.06)]$ \\
SC1 status & Satisfied & $0.88 > 0.85$ \\
$\rho_t$ & 0 & Case (a): interior regime \\
Slack ($\varphi_t^d(0) - \varphi_t^{\mathrm{req}}$) & 3 pp &
    Safety margin before premium emerges \\
\hline
\end{tabular}%
}
\end{table}

\begin{remark}[Why $\rho_t = 0$ is a success, not a failure]
\label{rem:success}
The baseline calibration shows $\rho_t = 0$.  This is the \emph{correct}
model prediction: Japan's institutional captivity is currently strong enough
that no sovereign premium is needed to clear the bond market.  The model's
value lies not in pricing a premium that does not yet exist, but in identifying
the conditions under which it \emph{will} emerge---and how quickly it can
amplify through the self-reinforcing loop once it does.
\end{remark}

\subsubsection{Counterfactual erosion scenarios}

Table~\ref{tab:stress-v2} varies $\theta_t$ and $z_t$ to identify the
conditions under which $\rho_t > 0$ emerges and the required growth
improvement for a safe transition.

\begin{table}[htbp]
\centering
\caption{Counterfactual Erosion Scenarios under the Two-Layer Model}
\label{tab:stress-v2}
\begin{tabular}{l c c c c c c}
\hline
Scenario & $\theta_t$ & $z_t$ & $\varphi_t^d(0)$ & $\rho_t^*$ &
    Required $\Delta g^{n*}$ & Feasible? \\
\hline
Baseline (2026) & 0.65 & 2.0\% & 0.88 & 0\% & 0.53\% & Conditional \\
Core erosion I & 0.60 & 2.0\% & 0.84 & 0.14\% & 0.67\% & Conditional \\
Core erosion II & 0.55 & 2.0\% & 0.81 & 0.44\% & 0.97\% & Tight \\
External stress & 0.65 & 3.0\% & 0.82 & 0.39\% & 0.92\% & Tight \\
Combined & 0.55 & 3.0\% & 0.72 & 1.17\% & 1.70\% & Unlikely \\
Severe & 0.45 & 3.5\% & 0.60 & 2.18\% & 2.71\% & Infeasible \\
\hline
\end{tabular}

\medskip
\noindent\textit{Notes.}  $\varphi_t^d(0)$ computed from
equation~\eqref{eq:phi-demand-uniform} at $\rho_t = 0$.
$\rho_t^*$ computed from equation~\eqref{eq:rho-v2} when
$\varphi_t^d(0) < \varphi_t^{\mathrm{req}} = 0.85$; zero otherwise.
Required $\Delta g^{n*}$ is the growth improvement needed to satisfy
condition~\eqref{eq:transition-v2} with $m = 0$, computed as
$\rho_t^* + 0.53\%$ (where 0.53\% is the no-premium baseline from
\Cref{sec:transition}).  ``Conditional'' = feasible at $\mu \geq 0.05$;
``Tight'' = requires $\mu \geq 0.07$;
``Unlikely'' = requires $\mu$ above the illustrative range;
``Infeasible'' = exceeds $\mu \times x_t^{\max}$ under all scenarios.
$\psi_t = 0.97$ and $\bar{c}_m = 6.0\%$ held fixed across scenarios;
$\varphi_t^{\mathrm{req}} = 0.85$ throughout.
\end{table}

\begin{remark}[The critical boundary]
\label{rem:boundary}
The scenario analysis identifies two critical thresholds for premium emergence:
\begin{enumerate}
    \item \textbf{Core-erosion threshold:} at $z_t = 2.0\%$, the premium
    becomes positive when $\theta_t$ falls below approximately $0.56$---a
    decline of roughly 9 percentage points from the current level;
    \item \textbf{External-stress threshold:} at $\theta_t = 0.65$, the premium
    becomes positive when $z_t$ rises above approximately $2.5\%$---an
    increase of 50 basis points from the baseline.
\end{enumerate}
The external-stress threshold is notably closer to the current position than
the core-erosion threshold.  This implies that the most immediate risk to the
JFR-rg regime is not demographic erosion of $\theta_t$ (which is slow) but a
shift in the global interest-rate environment that raises $z_t$ (which can be
sudden).  This finding reinforces the link to E4 (Multi-Country Repression
Equilibrium): the effective captive threshold depends on the foreign
repression state.
\end{remark}

\section{Inference for Regime Boundaries and Transition Margins}
\label{sec:regime-inference}

\Cref{sec:transition} and \Cref{sec:closure} imply two distinct but connected inference problems. The first concerns the emergence of a positive sovereign premium once the zero-premium domestic absorption condition fails. The second concerns the growth threshold required for a safe transition out of repression-dependent stability. These should not be conflated. The first is governed primarily by the two-layer domestic demand structure and the complementarity boundary of \Cref{sec:closure}; the second is governed primarily by the transition-feasibility condition of \Cref{sec:transition} and is therefore more directly exposed to debt-concept and fiscal-burden measurement choices.

For that reason, the inferential layer is developed in two parts. \Cref{sec:pe-inference} formulates conservative inference for the premium-emergence boundary. \Cref{sec:tf-inference} formulates conservative inference for the transition-feasibility margin. The common design principle is observables-centered outer inference: the paper does not force point classification when the maintained implementation layer admits multiple empirically legitimate readings.

\subsection{Premium-emergence boundary inference}
\label{sec:pe-inference}

\Cref{sec:closure} defines the interior regime as the case in which domestic demand at zero premium weakly exceeds the required domestic absorption share. This suggests the following premium-emergence boundary score.

\begin{definition}[Premium-emergence boundary score]
\label{def:bpe}
Let
\[
B_t^{PE}
:=
\varphi_t^d(0;\theta_t,\psi_t,z_t,G)-\varphi_t^{\mathrm{req}},
\]
where
\[
\varphi_t^d(0;\theta_t,\psi_t,z_t,G)
=
\theta_t+(1-\theta_t)\left[1-G\!\left(\frac{z_t}{\psi_t}\right)\right].
\]
Then:
\begin{itemize}
\item $B_t^{PE}>0$ indicates a zero-premium interior regime;
\item $B_t^{PE}=0$ indicates the premium-emergence boundary;
\item $B_t^{PE}<0$ indicates that zero-premium absorption is insufficient and a positive premium is required.
\end{itemize}
\end{definition}

Under the uniform-margin baseline of \Cref{sec:closure},
\[
\varphi_t^d(0)
=
\theta_t+(1-\theta_t)\left[1-\frac{z_t}{\psi_t \bar c_m}\right].
\]
Hence
\[
\frac{\partial B_t^{PE}}{\partial \theta_t}
=
\frac{z_t}{\psi_t \bar c_m},
\qquad
\frac{\partial B_t^{PE}}{\partial \psi_t}
=
(1-\theta_t)\frac{z_t}{\psi_t^2 \bar c_m},
\qquad
\frac{\partial B_t^{PE}}{\partial z_t}
=
-\frac{1-\theta_t}{\psi_t \bar c_m}.
\]
At the March 2026 baseline calibration of Table~\ref{tab:baseline-v2},
\[
(\theta_t,\psi_t,z_t,\bar c_m)=(0.65,0.97,2.0\%,6.0\%),
\]
which implies
\[
\frac{\partial B_t^{PE}}{\partial \theta_t}\approx 0.344,
\qquad
\frac{\partial B_t^{PE}}{\partial \psi_t}\approx 0.124,
\qquad
\frac{\partial B_t^{PE}}{\partial z_t}\approx -6.01.
\]
Thus a one-percentage-point uncertainty in $\theta_t$ changes the boundary score by about 0.34 percentage points, while a 50-basis-point increase in $z_t$ changes the score by about 3 percentage points. Since the baseline slack in Table~\ref{tab:baseline-v2} is 3 percentage points, uncertainty in $z_t$ and $\theta_t$ is quantitatively relevant to boundary classification.

The premium-emergence score is exposed mainly to four sources of ambiguity: the construction of the hard captive core $\theta_t$, the observational proxy for $\psi_t$, the outside-option spread $z_t$, and the contestable-margin specification $G$. These are summarized through a tiered admissible set.

\begin{definition}[Tiered admissible set for premium-emergence inference]
\label{def:mpe}
Let
\[
\mathcal M_{t,PE}^{(1)} \subseteq \mathcal M_{t,PE}^{(2)} \subseteq \mathcal M_{t,PE}^{(3)}
\]
be nested families of admissible observational/specification choices, where:
\begin{itemize}
\item \textbf{Tier 1} uses the baseline hard-core aggregation, the baseline equal-weight construction of $\psi_t$, the baseline $z_t$, and the uniform-margin specification;
\item \textbf{Tier 2} adds admissible alternative observational constructions of $\theta_t$ and $\psi_t$;
\item \textbf{Tier 3} adds admissible alternative contestable-margin specifications $G$ satisfying the support, continuity, and monotonicity conditions of \Cref{sec:closure}.
\end{itemize}
\end{definition}

For each $m\in \mathcal M_{t,PE}^{(k)}$, let $B_t^{PE}(m)$ denote the implied score. Define
\[
\underline B_{t,PE}^{(k)}
:=
\inf_{m\in \mathcal M_{t,PE}^{(k)}} B_t^{PE}(m),
\qquad
\overline B_{t,PE}^{(k)}
:=
\sup_{m\in \mathcal M_{t,PE}^{(k)}} B_t^{PE}(m).
\]
Then:
\begin{itemize}
\item the economy is \emph{robustly interior at tier $k$} if $\underline B_{t,PE}^{(k)} > 0$;
\item the economy is \emph{robustly boundary-near at tier $k$} if $0\in [\underline B_{t,PE}^{(k)},\overline B_{t,PE}^{(k)}]$;
\item the economy is \emph{robustly premium-emergent at tier $k$} if $\overline B_{t,PE}^{(k)} < 0$.
\end{itemize}

A useful baseline illustration is immediate from Table~\ref{tab:stress-v2}. Relative to the March 2026 baseline, moving from $z_t=2.0\%$ to $z_t=3.0\%$ while holding $(\theta_t,\psi_t,\bar c_m)$ fixed reduces $\varphi_t^d(0)$ from 0.88 to 0.82, thereby shifting the score from
\[
B_t^{PE}=0.88-0.85=0.03
\]
to
\[
B_t^{PE}=0.82-0.85=-0.03.
\]
This shows that the premium-emergence boundary is empirically close in the sense that a moderate outside-option widening can reverse the sign of the score even without any further erosion of $\psi_t$ or $\bar c_m$.

\begin{table}[htbp]
\centering
\caption{Illustrative Tier Widening for the Premium-Emergence Score}
\label{tab:pe-tier-widening}
\small
\begin{tabular}{>{\raggedright\arraybackslash}p{7.1cm} >{\centering\arraybackslash}p{2.2cm} >{\raggedright\arraybackslash}p{3.5cm}}
\toprule
Illustrative admissible reading & $B_t^{PE}$ & Conservative reading \\
\midrule
Tier 1 baseline only (Table~\ref{tab:baseline-v2}) & $+3$ pp & Robustly interior \\
Tier 2 lower-end hard-core reading using Core erosion I (Table~\ref{tab:stress-v2}) & $-1$ pp & Boundary-near \\
Tier 2 wider outside-option reading using External stress (Table~\ref{tab:stress-v2}) & $-3$ pp & Boundary-near \\
\bottomrule
\end{tabular}

\vspace{0.5em}
\begin{minipage}{0.92\textwidth}
\footnotesize
\textit{Notes.} These are illustrative widening examples built from the existing scenario map of Table~\ref{tab:stress-v2}; they are not themselves estimated confidence sets. Their purpose is to show that once the admissible reading is widened beyond the Tier 1 baseline, the lower envelope can cross zero quickly enough to justify conservative set-valued classification.
\end{minipage}
\end{table}

The inferential problem is not a standard threshold problem. The relevant boundary is defined by a complementarity system, not by an observed scalar cutoff; the premium-emergence regime may be only sparsely observed; and the principal ambiguities are structured, not incidental. For these reasons, the paper does not force point classification.
The proof strategy follows the same conservative logic. Appendix~\ref{app:pe-proof} first characterizes the tier-defined target score set through its outer envelope and then establishes asymptotic coverage for the detrended envelope process. The main-text theorem therefore summarizes a proof that is fully aligned with the observables-centered implementation layer rather than an add-on detached from it. Throughout Section~17 and Appendix~A, the mixing requirement is imposed on the detrended inferential remainder process; it is a technical condition for the outer-coverage layer rather than a theorem derived here from the full structural law of motion in Section~\ref{sec:closure}.

\begin{theorem}[Outer coverage for the premium-emergence score]
\label{thm:pe-outer}
Fix a tier $k\in\{1,2,3\}$ and a local inference window $\mathcal T_n(t_0)$ of width $h_n$ with $h_n \to \infty$ and $h_n/n \to 0$. Suppose that:
\begin{enumerate}
\item the \Cref{sec:closure} domestic demand schedule is continuous and weakly increasing in the premium, and the score map $B_t^{PE}(m)$ is continuous in the observables for every admissible $m\in\mathcal M_{t,PE}^{(k)}$;
\item the tier-$k$ admissible family $\mathcal M_{t,PE}^{(k)}$ is fixed ex ante, nested by construction, and compact under the maintained parameterization, so that the raw envelope width
\[
\overline B_{t,PE}^{(k)}-\underline B_{t,PE}^{(k)}
\]
is uniformly bounded on $\mathcal T_n(t_0)$;
\item the raw envelope processes admit local decompositions
\[
\underline B_{t,PE}^{(k)}=\underline\mu_{t,PE}^{(k)}+\underline u_{t,PE}^{(k)},
\qquad
\overline B_{t,PE}^{(k)}=\overline\mu_{t,PE}^{(k)}+\overline u_{t,PE}^{(k)},
\]
where the drift terms are piecewise-Lipschitz on $\mathcal T_n(t_0)$ and the local-linear detrending errors satisfy
\[
\sup_{t\in\mathcal T_n(t_0)}\bigl|\widehat{\underline\mu}_{t,PE}^{(k)}-\underline\mu_{t,PE}^{(k)}\bigr|=o_p(1),
\qquad
\sup_{t\in\mathcal T_n(t_0)}\bigl|\widehat{\overline\mu}_{t,PE}^{(k)}-\overline\mu_{t,PE}^{(k)}\bigr|=o_p(1);
\]
\item the detrended envelope remainders have uniformly bounded $(2+\delta)$ moments for some $\delta>0$ and satisfy an $\alpha$-mixing condition strong enough for dependence-robust subsampling or self-normalized block inference, for example
\[
\sum_{j\ge 1} j^{1/2}\alpha(j)^{\delta/(2+\delta)}<\infty;
\]
\item the block length satisfies $\ell_n\to\infty$ and $\ell_n/h_n\to 0$;
\item the interior side of the complementarity condition is observed with positive limiting frequency on $\mathcal T_n(t_0)$, so that the zero-premium demand schedule is empirically bounded from the interior side.
\end{enumerate}
Then a detrended subsampling or self-normalized block procedure applied to the envelope process yields an asymptotically valid outer confidence band for the tier-$k$ score interval
\[
[\underline B_{t,PE}^{(k)},\overline B_{t,PE}^{(k)}].
\]
Consequently, false declarations of robust interior safety or robust premium emergence vanish asymptotically away from the boundary, while boundary-near states remain conservatively set-valued. A full proof is given in Appendix~\ref{app:pe-proof}.
\end{theorem}

\begin{proof}[Proof sketch]
The complementarity structure implies that the sign of $B_t^{PE}$ is the relevant classification object. The tiered family determines an outer envelope over observational and specification ambiguity. Local detrending is applied to the envelope itself rather than to each admissible specification separately, which avoids order-dependence between detrending and the outer operator. After detrending, dependence-robust subsampling or self-normalized block arguments deliver valid outer coverage for the score interval.
\end{proof}

\subsection{Transition-feasibility margin inference}
\label{sec:tf-inference}

\Cref{sec:transition} gives a growth-improvement threshold for safe exit from repression-dependent stability. Unlike the premium-emergence boundary, this object is directly tied to the fiscal-burden term and is therefore more naturally exposed to debt-concept choice.

\begin{definition}[Transition-feasibility margin score]
\label{def:btf}
Let
\[
B_t^{TF}
:=
g_t^{n*,\mathrm{new}}
-
\left[
\pi_t+\frac{d_t-s_t}{b_{t-1}}+\bar\rho+m
\right].
\]
Then:
\begin{itemize}
\item $B_t^{TF}>0$ indicates that the transition condition is satisfied with margin;
\item $B_t^{TF}=0$ indicates a knife-edge transition;
\item $B_t^{TF}<0$ indicates that the transition remains infeasible at the proposed growth path.
\end{itemize}
\end{definition}

The key inferential distinction here is debt concept. Under the March 2026 no-premium baseline of \Cref{sec:transition}, with
\[
(\pi_t,d_t,b_{t-1},g_t^{n*,\mathrm{new}})=(2.7\%,2.0\%,2.40,3.0\%),
\]
the no-premium threshold is
\[
2.7+\frac{2.0}{2.40}-3.0 = 0.533\%.
\]
If the monitoring debt concept is used instead, with $b_{t-1}\approx 1.574$, the analogous quantity becomes
\[
2.7+\frac{2.0}{1.574}-3.0 \approx 0.971\%.
\]
Thus debt-concept choice alone raises the no-premium transition threshold by about 44 basis points. This does not invalidate the baseline result of \Cref{sec:transition}; it shows that transition-feasibility inference must distinguish between baseline and monitoring layers rather than silently mixing them.

\begin{table}[htbp]
\centering
\caption{Illustrative Tier Widening for the Transition-Feasibility Threshold}
\label{tab:tf-tier-widening}
\small
\begin{tabular}{>{\raggedright\arraybackslash}p{4.2cm} >{\raggedright\arraybackslash}p{4.7cm} >{\centering\arraybackslash}p{2.2cm} >{\centering\arraybackslash}p{2.1cm}}
\toprule
Tier & Debt concept reading & No-premium threshold & Widening \\
\midrule
Tier 1 & Baseline general-government concept & $0.533\%$ & --- \\
Tier 2 & Monitoring-layer JGB$+$FILP concept & $0.971\%$ & $+0.438$ pp \\
\bottomrule
\end{tabular}

\vspace{0.5em}
\begin{minipage}{0.92\textwidth}
\footnotesize
\textit{Notes.} The widening is mechanical: lowering $b_{t-1}$ raises the fiscal-burden term $(d_t-s_t)/b_{t-1}$ and therefore tightens the no-premium transition condition. In the JFR-rg architecture this is not merely a statistical discrepancy; it reflects the distinction between the broad accounting concept used for corridor logic and the narrower instrument-level monitoring concept used for observed-value checks.
\end{minipage}
\end{table}

For that reason, the admissible set for transition-feasibility inference should be centered on debt-concept and fiscal-burden choices.

\begin{definition}[Tiered admissible set for transition-feasibility inference]
\label{def:mtf}
Let
\[
\mathcal M_{t,TF}^{(1)} \subseteq \mathcal M_{t,TF}^{(2)}
\]
be nested admissible families, where:
\begin{itemize}
\item \textbf{Tier 1} uses the baseline debt concept and baseline fiscal-burden inputs of \Cref{sec:transition};
\item \textbf{Tier 2} adds the observed monitoring debt concept and admissible monitoring-layer fiscal-burden constructions.
\end{itemize}
\end{definition}

For each $m\in \mathcal M_{t,TF}^{(k)}$, let $B_t^{TF}(m)$ denote the implied score and define
\[
\underline B_{t,TF}^{(k)}
:=
\inf_{m\in \mathcal M_{t,TF}^{(k)}} B_t^{TF}(m),
\qquad
\overline B_{t,TF}^{(k)}
:=
\sup_{m\in \mathcal M_{t,TF}^{(k)}} B_t^{TF}(m).
\]
Tier widening is informative here because it quantifies how much of the transition margin is robust to debt-concept and fiscal-burden measurement choices.
Appendix~\ref{app:tf-proof} develops the corresponding proof in parallel form: it first fixes the tier-defined transition-feasibility target set and then proves outer coverage for the detrended envelope generated by debt-concept and fiscal-burden ambiguity. This keeps the transition-feasibility layer analytically distinct from the premium-emergence layer while preserving a common inferential architecture.

\begin{theorem}[Outer coverage for the transition-feasibility score]
\label{thm:tf-outer}
Fix a tier $k\in\{1,2\}$ and a local inference window $\mathcal T_n(t_0)$ of width $h_n$ with $h_n \to \infty$ and $h_n/n \to 0$. Suppose that:
\begin{enumerate}
\item the tier-$k$ admissible family $\mathcal M_{t,TF}^{(k)}$ is fixed ex ante, nested by construction, and compact under the maintained debt-concept and fiscal-burden parameterization;
\item conditional on the proposed growth path, the score map $B_t^{TF}(m)$ is affine in the fiscal-burden term for every admissible $m\in\mathcal M_{t,TF}^{(k)}$, and the resulting envelope width
\[
\overline B_{t,TF}^{(k)}-\underline B_{t,TF}^{(k)}
\]
is uniformly bounded on $\mathcal T_n(t_0)$;
\item the envelope processes admit local decompositions
\[
\underline B_{t,TF}^{(k)}=\underline\mu_{t,TF}^{(k)}+\underline u_{t,TF}^{(k)},
\qquad
\overline B_{t,TF}^{(k)}=\overline\mu_{t,TF}^{(k)}+\overline u_{t,TF}^{(k)},
\]
where the drift terms are piecewise-Lipschitz on $\mathcal T_n(t_0)$ and the local-linear detrending errors satisfy
\[
\sup_{t\in\mathcal T_n(t_0)}\bigl|\widehat{\underline\mu}_{t,TF}^{(k)}-\underline\mu_{t,TF}^{(k)}\bigr|=o_p(1),
\qquad
\sup_{t\in\mathcal T_n(t_0)}\bigl|\widehat{\overline\mu}_{t,TF}^{(k)}-\overline\mu_{t,TF}^{(k)}\bigr|=o_p(1);
\]
\item the detrended envelope remainders have uniformly bounded $(2+\delta)$ moments for some $\delta>0$ and satisfy an $\alpha$-mixing condition strong enough for dependence-robust subsampling or self-normalized block inference, for example
\[
\sum_{j\ge 1} j^{1/2}\alpha(j)^{\delta/(2+\delta)}<\infty;
\]
\item the block length satisfies $\ell_n\to\infty$ and $\ell_n/h_n\to 0$;
\item the growth-path proposal $g_t^{n*,\mathrm{new}}$ is observed or calibrated independently of the debt-concept choice within the tier-$k$ admissible family.
\end{enumerate}
Then the detrended envelope procedure yields an asymptotically valid outer confidence band for the tier-$k$ transition-feasibility interval
\[
[\underline B_{t,TF}^{(k)},\overline B_{t,TF}^{(k)}].
\]
Accordingly, robust transition feasibility is declared only when the lower envelope remains strictly positive, while debt-concept-sensitive cases are classified conservatively. A full proof is given in Appendix~\ref{app:tf-proof}.
\end{theorem}

\begin{proof}[Proof sketch]
The score is affine in the fiscal-burden term conditional on the growth proposal. The tiered admissible family induces an outer envelope over debt-concept and measurement-layer ambiguity. Applying the detrending-and-envelope logic to the resulting interval process yields valid outer coverage and hence conservative classification.
\end{proof}

\subsection{Relation to stochastic robustness}
\label{sec:stochastic-inference-link}

\Cref{sec:robustness} showed that bounded stochastic perturbations preserve the regime logic of JFR-rg in expectation and preserve the complementarity structure conditional on the state. That result should not be confused with pathwise safety classification.

For the premium-emergence boundary, a local approximation
\[
B_t^{PE}=\bar B_t^{PE}+\sigma_{B,t}^{PE}\varepsilon_t,
\qquad E[\varepsilon_t]=0,
\]
implies
\[
\Pr(B_t^{PE}<0)
=
\Pr\!\left(
\varepsilon_t<-\bar B_t^{PE}/\sigma_{B,t}^{PE}
\right).
\]
At the March 2026 baseline, $\bar B_t^{PE}=3$ pp. Table~\ref{tab:stress-v2} shows that increasing $z_t$ by 1 percentage point, from 2.0\% to 3.0\%, shifts the boundary score by roughly 6 pp, from $+3$ pp to $-3$ pp. Hence even moderate volatility in outside-option conditions can move the realized score by an amount of the same order as the baseline slack. The inferential problem is therefore not eliminated by expected-path preservation.

\begin{corollary}[Expected-path preservation does not imply pathwise safe classification]
\label{cor:pathwise}
Under \Cref{prop:stochastic-robust}, preservation of the JFR-rg regime logic in expectation does not imply that the realized premium-emergence score or transition-feasibility score remains uniformly positive with high probability. The relevant empirical object is the confidence band for the realized score interval, not the sign of the expected score alone.
\end{corollary}

\begin{figure}[t]
\centering
\includegraphics[width=0.95\textwidth]{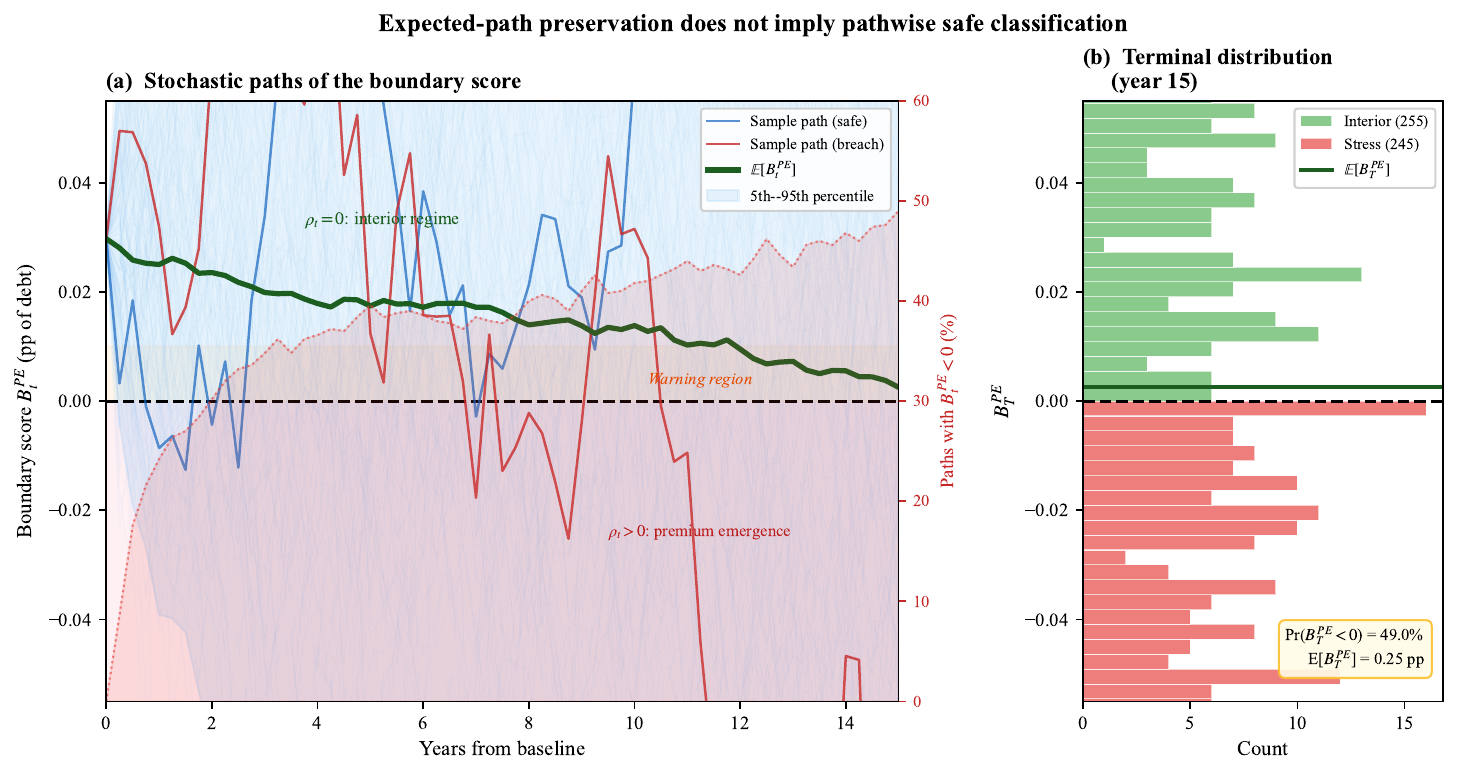}
\caption{Expected-path preservation versus pathwise boundary risk. Panel (a) shows simulated paths of the premium-emergence boundary score $B_t^{PE}$ under stochastic variation in the outside-option spread around the March 2026 baseline. The expected path remains positive, but individual realizations frequently breach the zero boundary. Panel (b) shows the terminal distribution after 15 years. The figure illustrates the logic of \Cref{cor:pathwise}: preservation of regime logic in expectation does not imply uniformly safe realized-path classification.}
\label{fig:pathwise-vs-expected}
\end{figure}

\subsection{Simulation evidence}
\label{sec:simulation-design}

This subsection reports two Monte Carlo exercises. The main experiment targets the
premium-emergence boundary score of Section~17.1, which is the sharper nonlinear object
induced by the complementarity closure of Section~16. A supplementary experiment targets
the transition-feasibility margin of Section~17.2, whose ambiguity is driven primarily by
debt-concept choice and fiscal-burden measurement rather than by nonlinear boundary
emergence. In both cases the numerical implementation should be read as a simplified
operational version of the proposed outer procedure rather than as a literal full-envelope
implementation of the corresponding theorem. Accordingly, the coverage criterion reported
below is regime-label coverage rather than full score-set coverage.

A further implication is that 100\% reported coverage in the simplified Monte Carlo
exercises should not be read as evidence of pointwise inferential efficiency. It reflects
the intentionally conservative design of a regime-label outer procedure under one-sided
regime dominance and restricted implementation-level ambiguity. The gain is the
suppression of false interior or false-feasible declarations; the cost is a wider warning
or marginal region than would arise under more aggressive point-classification rules. The
simulation results should therefore be read as evidence of conservative classification
discipline, not as a claim that the procedure is efficiency-optimal in the usual
statistical sense.

\paragraph{Premium-emergence boundary simulation.}
The premium-emergence experiment uses the Section~16 two-layer structure with persistent
state dynamics in $(\theta_t,\psi_t,z_t)$, sparse positive-premium episodes, and structured
measurement ambiguity. The goal is not to maximize point decisiveness, but to suppress
false safety under one-sided regime dominance. Table~\ref{tab:mc-comparison} reports the
main comparison across four evaluation horizons. The proposed tiered classifier eliminates
false safety and false alarm in the reported experiment, while returning boundary-near
set-valued classifications in approximately 29\% of cases at Tier~2 and about 59--70\% at
Tier~3. The naive plug-in benchmark instead forces point classification and generates
persistent false-safety frequencies around 0.8--0.9\%.

\begin{table}[htbp]
\centering
\caption{Monte Carlo Comparison of Regime-Classification Methods}
\label{tab:mc-comparison}
\small
\begin{tabular}{l l r r r r}
\toprule
Eval.\ horizon & Method & False safety & False alarm & Coverage & Warning \\
 & & (\%) & (\%) & (\%) & (\%) \\
\midrule
3.8 yr & Proposed (Tier 2) & 0.0 & 0.0 & 100.0 & 29.0 \\
 & Proposed (Tier 3) & 0.0 & 0.0 & 100.0 & 69.5 \\
 & B1: Naive plug-in & 0.9 & 0.9 & 98.2 & 0.0 \\
 & B2: Single threshold & 0.0 & 0.1 & 99.9 & 13.6 \\
 & B3: Fixed-spec test & 0.1 & 0.2 & 99.6 & 10.5 \\
\midrule
7.5 yr & Proposed (Tier 2) & 0.0 & 0.0 & 100.0 & 29.6 \\
 & Proposed (Tier 3) & 0.0 & 0.0 & 100.0 & 61.8 \\
 & B1: Naive plug-in & 0.8 & 0.9 & 98.2 & 0.0 \\
 & B2: Single threshold & 0.1 & 0.1 & 99.8 & 13.8 \\
 & B3: Fixed-spec test & 0.1 & 0.2 & 99.7 & 11.6 \\
\midrule
11.2 yr & Proposed (Tier 2) & 0.0 & 0.0 & 100.0 & 28.3 \\
 & Proposed (Tier 3) & 0.0 & 0.0 & 100.0 & 59.0 \\
 & B1: Naive plug-in & 0.8 & 0.8 & 98.5 & 0.0 \\
 & B2: Single threshold & 0.0 & 0.0 & 100.0 & 12.2 \\
 & B3: Fixed-spec test & 0.1 & 0.1 & 99.8 & 11.8 \\
\midrule
15.0 yr & Proposed (Tier 2) & 0.0 & 0.0 & 100.0 & 29.4 \\
 & Proposed (Tier 3) & 0.0 & 0.0 & 100.0 & 59.2 \\
 & B1: Naive plug-in & 0.8 & 0.8 & 98.4 & 0.0 \\
 & B2: Single threshold & 0.1 & 0.1 & 99.9 & 12.2 \\
 & B3: Fixed-spec test & 0.0 & 0.1 & 100.0 & 12.9 \\
\bottomrule
\end{tabular}

\vspace{0.5em}
\begin{minipage}{0.92\textwidth}
\footnotesize
\textit{Notes.} $N=2{,}000$ Monte Carlo replications; $T=60$ quarters; $\alpha=0.10$. DGP uses the \Cref{sec:closure} two-layer model with parameters from Table~\ref{tab:baseline-v2}. Measurement noise $\sigma_{\theta}^{\mathrm{obs}}=2\%$. Block length $\ell_n=6$ (default). ``False safety'' = classified as interior when true state is stress. ``False alarm'' = classified as stress when true state is interior. ``Coverage'' = true state is contained in the reported classification set. ``Warning'' = procedure returns set-valued (boundary-near) classification.
\end{minipage}
\end{table}

Figure~\ref{fig:mc-comparison} visualizes the terminal-horizon trade-off. The proposed
procedures sharply reduce false precision, but they do so by enlarging the warning region.
This is a design feature rather than a defect: the purpose of the classifier is to avoid
false interior declarations when the regime boundary is latent and the sample is dominated
by one side of the split.

\begin{figure}[t]
\centering
\includegraphics[width=0.88\textwidth]{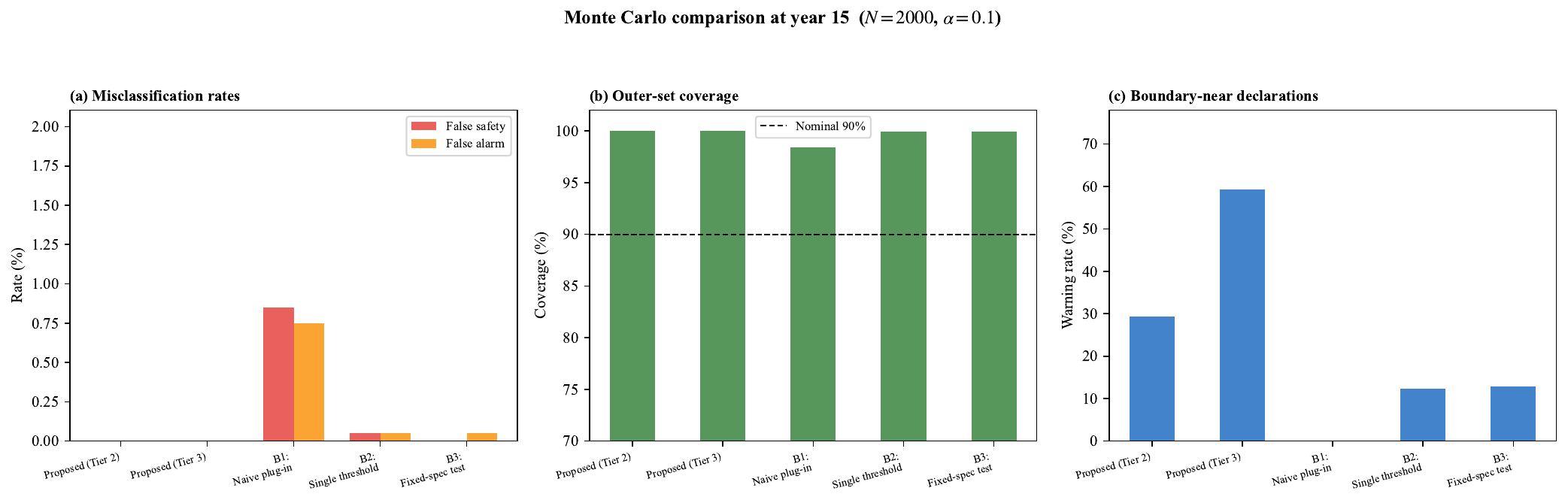}
\caption{Monte Carlo comparison for the premium-emergence boundary classifier at the
terminal evaluation horizon. The proposed Tier~2 and Tier~3 procedures eliminate false
safety in the reported experiment, but do so by returning a larger warning region than the
simpler benchmark rules. The figure should be read as evidence for conservative regime-label
classification rather than as a claim of pointwise optimality.}
\label{fig:mc-comparison}
\end{figure}

Table~\ref{tab:mc-blocklength} reports block-length sensitivity for the proposed Tier~2
classifier. Across $\ell_n\in\{4,6,8\}$, the classification outcome remains stable: false
safety and false alarm remain at zero, coverage remains at 100\%, and the warning rate
moves only from 28.6\% to 29.9\%. This supports the small-sample recommendation in the
Appendix that conservative block-length grids be used as a sensitivity device rather than
as a source of pointwise tuning.

\begin{table}[htbp]
\centering
\caption{Block-Length Sensitivity for the Proposed Classifier (Tier 2, Terminal Evaluation)}
\label{tab:mc-blocklength}
\small
\begin{tabular}{c r r r r}
\toprule
Block length $\ell_n$ & False safety (\%) & False alarm (\%) & Coverage (\%) & Warning (\%) \\
\midrule
4 & 0.0 & 0.0 & 100.0 & 28.6 \\
6 & 0.0 & 0.0 & 100.0 & 29.4 \\
8 & 0.0 & 0.0 & 100.0 & 29.9 \\
\bottomrule
\end{tabular}

\vspace{0.5em}
\begin{minipage}{0.92\textwidth}
\footnotesize
\textit{Notes.} All rows use the proposed Tier-2 outer-inference procedure at the terminal evaluation point (year 15). Larger block lengths produce wider confidence bands, increasing the warning rate and coverage at the cost of reduced decisiveness. The classification outcome is stable across the grid $\ell_n \in \{4,6,8\}$, confirming that the regime diagnosis is not driven by the block-length choice.
\end{minipage}
\end{table}

\paragraph{Transition-feasibility simulation.}
A supplementary Monte Carlo is also reported for the transition-feasibility margin. In
contrast to the premium-emergence boundary, the TF-side object is affine in the
fiscal-burden term and is therefore primarily informative about debt-concept ambiguity
rather than nonlinear boundary emergence. Table~\ref{tab:mc-tf-comparison} shows that the
proposed tiered outer procedure eliminates false feasibility and false infeasibility in the
reported experiment, at the cost of a modest increase in marginal classifications from
22.5\% to 25.8\% when moving from Tier~1 to Tier~2 at the benchmark premium bound
$\bar\rho=0.5\%$. The corresponding regime-label coverage remains at 100\% in the reported
experiment.

\begin{table}[htbp]
\centering
\caption{Transition-Feasibility Monte Carlo: Classification Performance at Terminal Horizon}
\label{tab:mc-tf-comparison}
\small
\begin{tabular}{l l r r r r}
\toprule
Premium bound & Method & False feas.\ & False infeas.\ & Coverage & Marginal \\
 & & (\%) & (\%) & (\%) & (\%) \\
\midrule
no premium & Proposed (Tier 1) & 0.0 & 0.0 & 100.0 & 28.2 \\
 & Proposed (Tier 2) & 0.0 & 0.0 & 100.0 & 36.2 \\
 & Naive plug-in (baseline) & 0.3 & 0.0 & 99.7 & 0.0 \\
 & Naive plug-in (monitoring) & 0.0 & 0.9 & 99.1 & 0.0 \\
 & Fixed-spec test (baseline) & 0.0 & 0.0 & 100.0 & 11.4 \\
\midrule
$\bar\rho=0.5\%$ & Proposed (Tier 1) & 0.0 & 0.0 & 100.0 & 22.5 \\
 & Proposed (Tier 2) & 0.0 & 0.0 & 100.0 & 25.8 \\
 & Naive plug-in (baseline) & 0.1 & 0.0 & 99.9 & 0.0 \\
 & Naive plug-in (monitoring) & 0.0 & 0.4 & 99.6 & 0.0 \\
 & Fixed-spec test (baseline) & 0.0 & 0.0 & 100.0 & 9.4 \\
\midrule
$\bar\rho=1.0\%$ & Proposed (Tier 1) & 0.0 & 0.0 & 100.0 & 13.2 \\
 & Proposed (Tier 2) & 0.0 & 0.0 & 100.0 & 14.9 \\
 & Naive plug-in (baseline) & 0.1 & 0.0 & 99.9 & 0.0 \\
 & Naive plug-in (monitoring) & 0.0 & 0.3 & 99.7 & 0.0 \\
 & Fixed-spec test (baseline) & 0.0 & 0.0 & 100.0 & 5.0 \\
\bottomrule
\end{tabular}

\vspace{0.5em}
\begin{minipage}{0.92\textwidth}
\footnotesize
\textit{Notes.} $N=2{,}000$ replications; $T=60$ quarters; $\alpha=0.10$. The true debt concept is drawn uniformly from $[b_{\mathrm{monitoring}}, b_{\mathrm{baseline}}] = [1.574, 2.40]$ in each replication. ``False feas.''\ = declared feasible when truly infeasible; ``False infeas.''\ = declared infeasible when truly feasible. ``Marginal'' = set-valued classification returned. The Tier-2 procedure eliminates false feasibility by explicitly incorporating debt-concept ambiguity into the outer envelope.
\end{minipage}
\end{table}

The widening generated by debt-concept ambiguity is economically meaningful but not
excessive. The simulated Tier~2 envelope width has mean 43.7 basis points and median
43.7 basis points, closely matching the analytic 44-basis-point widening implied by the
March 2026 baseline-versus-monitoring comparison in Section~17.2. Figure~\ref{fig:tf-envelope-width}
shows the distribution of this widening across Monte Carlo replications.

\begin{figure}[t]
\centering
\includegraphics[width=0.75\textwidth]{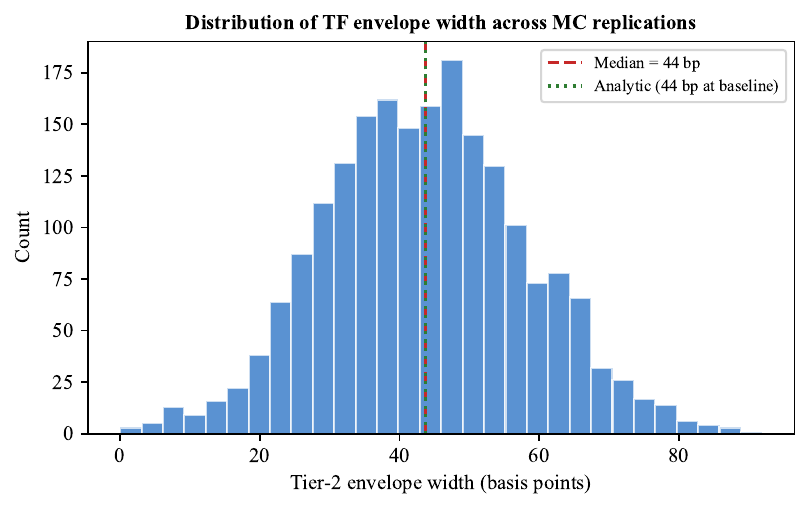}
\caption{Distribution of Tier~2 transition-feasibility envelope width across Monte Carlo
replications. The mean and median widening are both approximately 43.7 basis points,
closely matching the analytic 44-basis-point widening implied by the March 2026 baseline
and monitoring debt concepts.}
\label{fig:tf-envelope-width}
\end{figure}

Table~\ref{tab:mc-tf-blocklength} shows that the TF-side classification is likewise stable
across $\ell_n\in\{4,6,8\}$, with false feasibility and false infeasibility both equal to
zero in the reported experiment and only minor variation in the marginal classification
rate. As with the premium-emergence exercise, this numerical evidence should be read as a
conservative operational illustration of the inferential layer, not as a literal full-score-set
implementation of the theorem.

\begin{table}[htbp]
\centering
\caption{TF Block-Length Sensitivity (Tier 2, $\bar\rho=0.5\%$, Terminal Evaluation)}
\label{tab:mc-tf-blocklength}
\small
\begin{tabular}{c r r r r}
\toprule
$\ell_n$ & False feas.\ (\%) & False infeas.\ (\%) & Coverage (\%) & Marginal (\%) \\
\midrule
4 & 0.0 & 0.0 & 100.0 & 24.9 \\
6 & 0.0 & 0.0 & 100.0 & 25.8 \\
8 & 0.0 & 0.0 & 100.0 & 25.7 \\
\bottomrule
\end{tabular}
\end{table}

Taken together, the two simulation exercises support a common conclusion. The proposed
outer procedures are not designed to maximize point decisiveness. Their purpose is to
reduce false precision when either the premium-emergence boundary or the transition margin
is observed through structured ambiguity. The price of that conservatism is a nontrivial
set-valued region, but the reward is the near-elimination of false interior or false-feasible
declarations in the reported experiments.

\subsection{Scope and simulation caveats}
\label{sec:inference-scope}

This section does not claim unique point identification of the latent premium-emergence boundary or of the transition-feasibility margin. Its claim is narrower. Given the observables-centered discipline of JFR-rg and the implementation-layer ambiguities already recognized elsewhere in the paper, the two relevant empirical tasks are naturally set-valued: one concerns the outer location of the premium-emergence boundary, and the other concerns the outer location of the transition-feasibility margin. The inferential procedures proposed here are designed to match that structure conservatively rather than to replace the economic logic of \Cref{sec:transition,sec:closure}.

Three caveats apply to the simulation evidence reported above. First, both Monte Carlo exercises implement simplified operational versions of the proposed outer-inference procedure; they are not literal implementations of the theorem-level envelope-and-subsampling protocol described in \Cref{thm:pe-outer,thm:tf-outer}. Second, the coverage criterion used is regime-label coverage---whether the true regime is contained in the reported classification set---rather than the score-set coverage $\Prob(\mathcal{B}_t^{(k)} \subseteq \widehat{\mathcal{B}}_{t,1-\alpha}^{(k)})$ stated in the theorems; the latter requires the full envelope implementation. Third, the structured ambiguity in the measurement sets is restricted to a minimal nontrivial version ($\theta_t$-noise and $G$-specification for PE; debt concept for TF); the full admissible families described in \Cref{def:mpe,def:mtf} are broader.

\section{What Standard Models Leave Exogenous in the Japanese Case: 2012--2024}
\label{sec:mainstream-comparison}

The preceding sections have developed the JFR-rg architecture in abstract
form.  This section grounds the framework by examining a concrete
twelve-year episode---Japan from the launch of Quantitative and Qualitative
Easing (QQE) in April 2013 through the monetary-policy normalization that
began in 2024---and asking two questions: what did standard debt-sustainability
models predict for this period, and what does JFR-rg explain that those
models do not?

The purpose is not to declare standard models wrong.  It is to identify the
specific empirical regularities that fall outside their explanatory reach
and to show that these regularities are precisely the phenomena that the
JFR-rg regime variables---$\eps_t$, $\phic$, $\psi_t$, and now
$\theta_t$ and $\rho_t$---are designed to organize.

\subsection{The empirical puzzle}

Between 2012 and 2024, Japan's general-government gross debt rose from
approximately 220\% to over 240\% of GDP.  Throughout this period, the
10-year JGB yield remained below 1\% until late 2022, and the sovereign
spread over the overnight policy rate was effectively controlled by
yield-curve control (YCC, September 2016--October 2023).  No fiscal crisis
materialized.  No sovereign downgrade triggered a destabilizing sell-off.
The domestic absorption share $\phic$ remained above 85\%, and at no point
did the sovereign risk premium $\rho_t$ become visibly positive in JGB
pricing.

This outcome is difficult to reconcile with three prominent strands of the
mainstream literature.

\subsection{Three mainstream predictions and their outcomes}

\paragraph{Prediction 1: Fiscal crisis within a decade.}
Hoshi and Ito (2014) projected that Japan's sovereign debt could not
continue to increase without a crisis, estimating that the domestic
private-sector capacity to absorb further issuance would be exhausted
within approximately ten years.  Their analysis was grounded in a careful
assessment of household financial assets relative to outstanding
government liabilities.  The crisis did not materialize within the
projected window.

The JFR-rg diagnosis is that Hoshi and Ito correctly identified the
\emph{stock} constraint on domestic absorption but did not model the
\emph{flow} mechanism through which QQE altered the constraint.  In
JFR-rg terms, the BoJ's asset purchases directly sustained $\theta_t$
(the hard captive core) by transferring sovereign bonds from the
contestable margin to the institutional core, thereby postponing the
date at which $\varphi_t^d(0)$ would fall below $\varphi_t^{\mathrm{req}}$.
The Hoshi-Ito prediction implicitly assumed $\gamma^{\theta}_t = 0$
(no active policy maintenance); QQE set $\gamma^{\theta}_t > 0$ for a
decade, extending the SC1 window well beyond the projected horizon.

\paragraph{Prediction 2: $r < g$ as a sufficient condition.}
Blanchard (2019) argued that when the safe interest rate is below
the growth rate, the fiscal cost of debt is lower than conventionally
assumed, and debt roll-overs may be feasible even at elevated debt
levels.  This framework provides a partial explanation for Japan's
stability: throughout 2013--2024, $r^n_t - g^n_t$ was indeed negative,
ranging from approximately $-0.5\%$ to $-3.0\%$ depending on the
maturity and growth measure used.

However, the Blanchard framework does not explain \emph{why} $r^n_t$ remained
so low in Japan specifically.  In a standard open-economy setting, a
sovereign with debt at 240\% of GDP should face a positive risk premium
that partially offsets the $r < g$ advantage.  The fact that Japan did not
face such a premium is the phenomenon that requires explanation, not a
primitive that can be taken as given.

JFR-rg provides the missing mechanism: the combination of a captive
domestic absorption structure ($\varphi_t \geq \bar{\varphi}$), active
repression ($\eps_t > 0$ under QQE), and full institutional control
($\psi_t \approx 0.97$ in the inherited baseline, about $0.98$ in the 2025 observed proxy layer) together suppressed $\rho_t$ to zero.  The
$r < g$ condition was not exogenous good fortune; it was an endogenous
consequence of the JFR-rg regime operating under SC1.  \Cref{sec:closure}
formalizes this: $\rho_t = 0$ is the equilibrium outcome when
$\varphi_t^d(0) > \varphi_t^{\mathrm{req}}$, which is precisely the
condition that Japan's institutional structure sustained throughout the
QQE period.

\paragraph{Prediction 3: $r < g$ is not sufficient for sustainability.}
Two complementary refinements of the low-rate argument establish that
$r < g$ alone does not guarantee fiscal sustainability.  Mehrotra and
Sergeyev~(2021) show that when the primary surplus is bounded, a
state-dependent threshold level of public debt still determines
sustainability even under $r < g$.  Reis~\cite{Reis2021} develops a
distinct wedge: with $g < m$ (where $m$ denotes the marginal product
of capital), the present-value budget constraint binds through bubble
premia on safe debt rather than through $r - g$ directly.

Both refinements are useful but remain silent on the institutional
mechanism that determines $r$.  In particular, neither distinguishes
between an economy where $r$ is low because of global safe-asset scarcity
(as in the United States or Germany) and one where $r$ is low because of
domestic financial repression (as in Japan).  The two cases have
identical implications for the $r - g$ arithmetic \emph{today} but
radically different implications for what happens when the institutional
configuration changes---as Japan discovered when the BoJ began normalizing
in 2024 and $\eps_t$ reversed sign.

JFR-rg makes this distinction operational through the decomposition
of the effective interest rate into a repression-consistent component
($r^{\mathrm{rep}}_t$) and a sovereign premium ($\rho_t$).  As long as
SC1 holds, $\rho_t = 0$ and the observed $r$ equals $r^{\mathrm{rep}}_t$.
When SC1 weakens, $\rho_t$ emerges endogenously (\Cref{sec:closure}), and the
effective $r$ rises discretely---a transition that the $r < g$ framework
treats as an exogenous parameter shift but that JFR-rg predicts as a
regime consequence.

\subsection{The JFR-rg reading of 2012--2024}

Mapping the twelve-year episode onto JFR-rg variables produces a coherent
three-phase narrative.

\paragraph{Phase 1: Regime activation (2013--2016).}
The launch of QQE in April 2013, followed by its expansion in October 2014
and the introduction of negative interest rates in January 2016, activated
the full JFR-rg configuration. The repression channel turned strongly
positive as BoJ purchases pinned yields below inflation ($\eps_t > 0$),
domestic absorption strengthened as the BoJ absorbed net issuance through the
secondary market (supporting both $\phic$ and $\theta_t$), and institutional
control remained effectively complete ($\psi_t \approx 0.97$ in the inherited
baseline, about $0.98$ in the 2025 observed proxy layer). The stability
condition~\eqref{eq:stability} was therefore satisfied with substantial
margin.

In \Cref{sec:closure} terms, Phase~1 corresponds to case~(a) of the
complementarity condition: $\varphi_t^d(0) \gg \varphi_t^{\mathrm{req}}$,
$\rho_t = 0$, and the safe equilibrium is uniquely selected with a large
buffer.

\paragraph{Phase 2: Maintained repression with slowed erosion (2016--2022).}
The introduction of yield-curve control in September 2016 converted the
initial activation phase into an explicitly administered maintenance regime:
the BoJ committed to keeping the 10-year JGB yield near zero, later widening
the tolerance band to $\pm 0.5\%$. During this phase, $\eps_t$ remained
positive and $\theta_t$ stayed high at roughly 0.65, reflecting the large BoJ
presence in the JGB market. The most careful reading of E5 is therefore not
that the Demographic-$\phi$ Clock was literally paused throughout the whole
period, but that policy maintenance materially slowed structural erosion.
Short-window monitoring can display $\kappa \approx 0$ in
\cref{app:window-sensitivity}, while the full-sample structural-break
monitoring estimate in the main text shows that the post-QQE holder-share trend
remains negative once the entire 1997--2025 path is used. In that sense,
$\gamma^{\theta}_t$ partially offset $\kappa^{\theta}_t$ during the QQE/YCC
regime without fully eliminating structural erosion.

In \Cref{sec:closure} terms, Phase~2 still belongs to case~(a):
$\varphi_t^d(0)$ remained comfortably above
$\varphi_t^{\mathrm{req}}$, $\rho_t = 0$ continued to hold, and the main
change was not the emergence of a premium but the gradual narrowing of the
buffer through time.

\paragraph{Phase 3: Partial regime transition with a still-slack interior (2022--2024).}
The global monetary-tightening cycle, led by the Federal Reserve's rate hikes
beginning in March 2022, raised $z_t$ for Japanese institutional investors.
At the same time, the BoJ moved from strict YCC maintenance to progressive
normalization: the tolerance band was widened in December 2022, widened again
in July 2023, effectively loosened in October 2023, and formally exited in
March 2024 alongside the end of negative interest rates.

In JFR-rg terms, this phase combines a decline in $\gamma^{\theta}_t$
(policy maintenance weakening) with a rise in $z_t$ (the outside option
improving). The result was an $\eps$-reversal that, by mid-2025, was visible in the observed data:
$\eps_t \approx -0.81\%$, so the repression dividend disappeared. Yet
$\theta_t$ did not collapse. The hard captive core---regulatory bank
holdings, pension-fund ALM requirements, and other mandate-like holders---
remained structurally intact even without QQE support. In
\Cref{sec:closure} terms, $\varphi_t^d(0)$ declined but remained above
$\varphi_t^{\mathrm{req}}$, so $\rho_t = 0$ still held, consistent with the
absence of a visible sovereign premium in JGB pricing as of mid-2025.

This configuration---$\eps_t < 0$ but $\rho_t = 0$---is precisely the
``still-slack'' interior regime identified in the calibration and stress
scenarios of \Cref{sec:closure}. The repression dividend has vanished, but
the captive structure remains strong enough that no premium emerges.
The relevant risk is therefore no longer immediate regime collapse but transition
across the interior boundary: further erosion of $\theta_t$ or a further rise
in $z_t$ could move the system from case~(a) to case~(c), triggering a
positive $\rho_t$ for the first time in the post-QQE era.

\begin{remark}[Retrospective discriminating implication]
\label{rem:retrospective}
The YCC exit of March 2024 and the subsequent $\eps$-reversal provide a useful
discriminating implication in a retrospective sense.
A Hoshi--Ito-style model, which treats the capacity for domestic absorption as
a \emph{stock} constraint, would predict that removing the yield-suppression
mechanism should trigger an immediate premium as the captive-absorption shortfall
emerges.
JFR-rg predicts otherwise: the hard captive core $\theta_t$---comprising regulatory
bank holdings, pension-fund ALM requirements, and FILP-channeled savings---is
\emph{structurally distinct} from the QQE-supported portion of domestic demand,
and SC1 can remain satisfied even after the QQE component is withdrawn, as long
as $\theta_t > \varphi_t^{\mathrm{req}}$.
The observed persistence of $\rho_t = 0$ through mid-2025 is consistent with this
prediction and inconsistent with a pure stock-absorption model.
This is offered as a \emph{retrospective discriminating implication}---a phenomenon
that the JFR-rg architecture organizes more sharply than the alternatives---rather
than as an ex ante out-of-sample forecast.
\end{remark}

\subsection{The explanatory gap}

The comparison can be summarized as follows.

\begin{table}[htbp]
\centering
\caption{Mainstream vs.\ JFR-rg: Explanatory Reach for Japan 2012--2024}
\label{tab:mainstream-comparison}
\small
\begin{tabular}{>{\raggedright\arraybackslash}p{4.5cm}
                >{\centering\arraybackslash}p{2.0cm}
                >{\centering\arraybackslash}p{2.0cm}
                >{\centering\arraybackslash}p{2.0cm}
                >{\centering\arraybackslash}p{2.0cm}}
\toprule
Empirical regularity &
    Hoshi--Ito &
    Blanchard &
    Mehrotra--Sergeyev &
    JFR-rg \\
\midrule
No crisis at 240\% debt/GDP &
    $\times$ &
    $\checkmark$ &
    $\checkmark$ &
    $\checkmark$ \\[4pt]
Why $r$ remained so low &
    --- &
    Assumed &
    Assumed &
    Explained (SC1 $+$ $\eps > 0$) \\[4pt]
Role of BoJ purchases in stability &
    Partial &
    --- &
    --- &
    $\theta_t$ maintenance \\[4pt]
Why short windows can look flat even when full-sample erosion remains positive &
    --- &
    --- &
    --- &
    Policy maintenance slows erosion; local flatness need not imply zero structural trend \\[4pt]
$\eps$-reversal and its consequences &
    --- &
    --- &
    --- &
    Predicted (failure mode) \\[4pt]
Conditions for future $\rho_t > 0$ &
    --- &
    --- &
    --- &
    \Cref{sec:closure} MCP \\[4pt]
Why Italy/Greece differ at similar debt &
    --- &
    --- &
    --- &
    $\psi_t$ ordering (E6) \\
\bottomrule
\end{tabular}

\vspace{0.5em}
\begin{minipage}{0.92\textwidth}
\footnotesize
\textit{Notes.}
$\checkmark$ = explained within the framework;
$\times$ = predicted outcome inconsistent with observation;
--- = not addressed by the framework;
``Assumed'' = the framework takes the phenomenon as an input rather than
explaining it;
``Partial'' = the framework addresses a related mechanism but does not
provide a complete account.
Different frameworks externalize different objects for valid analytical reasons;
the purpose of this table is to show the additional organizing value of JFR-rg
in the Japanese case, not to imply a defect in frameworks designed for other
questions.
\end{minipage}
\end{table}

The table identifies the specific explanatory contribution of JFR-rg:
it accounts for the institutional \emph{mechanism} through which low
sovereign rates were sustained, why that mechanism has begun to weaken,
and under what conditions it will fail.  Standard models either take
the low rate as given (Blanchard, Mehrotra--Sergeyev) or predict a crisis
that did not occur (Hoshi--Ito).  JFR-rg does not claim superiority in
all dimensions---it does not, for instance, provide a welfare analysis or
a political-economy model of the BoJ's decision-making---but it fills a
specific gap that the mainstream literature has left open.

\begin{remark}[Complementarity, not replacement]
\label{rem:complementarity}
The comparison above is not intended to dismiss the mainstream frameworks.
The Blanchard $r < g$ insight correctly identifies the arithmetic
conditions under which high debt is sustainable; the Mehrotra--Sergeyev
refinement correctly bounds the fiscal space even under $r < g$ (as does
the related $g < m$ wedge of Reis~\cite{Reis2021}); and the Hoshi--Ito
analysis correctly identifies the long-run stock constraint.
JFR-rg complements these contributions by providing the
\emph{institutional layer} that determines whether the $r < g$ condition
holds, how long it will hold, and what happens when it ceases to hold.
In the limiting case $\psi_t \to 0$ (loss of institutional control),
JFR-rg collapses to the mainstream framework
(\cref{app:limits,prop:corridor-psi}), confirming that the standard
analysis is nested as a special case rather than contradicted.
\end{remark}

\section{Conclusion}
\label{sec:conclusion}

Part~II does not supersede Part~I~\cite{Wakimoto2026PartI}. It develops the logical extension of that framework within the same observables-centered, regime-conditional architecture. The original framework clarified how a high-debt, low-growth economy could remain temporarily stable under a specific combination of financial repression, captive domestic absorption, and bounded exchange-rate depreciation. The present paper shows that, once such a regime is admitted, its principal dynamic implications can be stated more fully and more systematically than in Part~I alone.

The claim of logical completion should be read in this architectural sense. It is not a claim that all empirical parameters have been estimated with final precision, that welfare or political-economy questions have been exhausted, or that a full general-equilibrium microfoundation has been derived. Rather, the paper establishes that the principal dynamic implications internal to the JFR-rg framework can now be stated in closed form, and that the most natural excluded generalizations---bounded stochastic perturbations and endogenous fiscal responses---do not overturn the regime logic on which those implications depend.

Three results organize this conclusion. First, the six extensions E1--E6 map the principal dynamic consequences of Part~I into explicit propositions concerning path dependence, bounded reinvestment gains, debt-reduction asymmetries, multi-country threshold interactions, finite institutional horizons, and institutional control rights. Second, the robustness results show that leaving stochastic and fiscal-response generalizations outside the main architecture does not render the framework dynamically incomplete. Third, the Minimal Equilibrium Closure provides an endogenous premium relation sufficient to close the transition problem within the model's block-recursive structure. This closure should be understood as minimal and sufficient, not as a substitute for a full household-portfolio or welfare-theoretic general equilibrium.

Because the completed architecture defines two latent regime boundaries---the premium-emergence boundary of \cref{sec:closure} and the transition-feasibility margin of \cref{sec:transition}---the paper also formulates the corresponding inferential problem. The proposed outer-inference procedures are conservative by design: they construct set-valued regime classifications rather than forcing point precision when the available observables remain compatible with multiple nearby regime readings. Monte Carlo evidence confirms that the tiered measurement-set approach eliminates false-safety and false-feasibility declarations in the reported experiments, at the cost of a nontrivial warning region. This trade-off is intentional. The purpose of the inferential layer is not to maximize decisiveness but to prevent false interior or false-feasible conclusions when the economy is near the complementarity boundary or when debt-concept ambiguity is quantitatively relevant.

The endogenous-premium framework yields several implications that the exogenous-premium version could not state as sharply. The transition problem becomes two-dimensional, requiring the joint evaluation of the growth path and the institutional-erosion path. Early investment acquires endogenous urgency because investment undertaken while the captive system remains intact is more valuable than the same investment undertaken after premium emergence. And the most immediate transition risk may arise not only from slow-moving domestic erosion, but also from external yield conditions that raise the outside-option spread and accelerate premium formation.

The comparison with the mainstream literature also becomes clearer in light of this extension. JFR-rg is not offered as a wholesale replacement for standard debt-sustainability analysis. Blanchard~(2019), Hoshi--Ito~(2014), and Mehrotra--Sergeyev~(2021) each identify important dimensions of debt dynamics. What JFR-rg adds is the institutional layer: the conditions under which low sovereign rates are sustained, the mechanisms through which that support can weaken, and the boundary at which the analysis collapses back toward the mainstream limiting case. Its contribution is therefore complementary and conditional rather than universal.

Observed-value monitoring as of mid-2025 reveals that part of the regime transition anticipated by the framework may already be underway: the 10-year JGB yield has risen above core CPI, reversing the sign of $\eps_t$ and eliminating the repression dividend. In the two-layer model of \cref{sec:closure}, this $\eps$-reversal shuts down the policy-maintenance channel ($\gamma_t^{\theta} \to 0$), but the hard captive core remains structurally intact, so $\rho_t = 0$ continues to hold---consistent with the absence of a visible sovereign premium in JGB pricing. The system is therefore in the ``still-slack'' interior regime. At the same time, the observed monitoring layer reports $\kappa_{\mathrm{post\mbox{-}QQE}} = 0.188\%$/yr and a linear clock of about $43.6$ years, so the E5 horizon looks materially less urgent than the conservative baseline clock taken over from Part~I. The timing problem nevertheless remains live because, in Section~\ref{sec:closure}, transition risk is no longer governed by the $\phi$-clock alone: the relevant buffer is the joint margin in $\theta_t$, $z_t$, and the premium-emergence condition.

For that reason, the empirical roadmap developed in Part~II matters as much as the formal architecture itself. Each extension is paired with observable implications, failure modes, and a stated interpretation of rejection. If future evidence fails to support a given extension, the result is not that the entire framework collapses wholesale, but that the relevant claim should be read more narrowly. That is the sense in which Part~II aims to strengthen JFR-rg: not by insulating it from falsification, but by sharpening its scope, clarifying its architecture, and making its strongest claims easier to evaluate against public evidence.

\clearpage

\section*{Acknowledgments}
\addcontentsline{toc}{section}{Acknowledgments}

The author acknowledges the use of large language models as assistive
tools during the preparation of this manuscript. More specifically:

{\sloppy\raggedright
\begin{itemize}
  \item \textbf{Claude Opus 4.6 \& 4.7} (Anthropic) --- used frequently for
    iterative drafting of theoretical sections, refinement of
    mathematical exposition, \LaTeX{} phrasing, and structural revision
    of the manuscript.

  \item \textbf{GPT-5.4 Thinking} (OpenAI) --- used frequently for
    editorial review of exposition, consistency checks, logical
    stress-testing, reference and wording cross-checks, and
    pre-submission feedback.

  \item \textbf{Gemini 3.1 Pro \& 3 Flash} (Google) --- used for rapid alternative
    formulation checks, exploratory verification, and supplementary
    drafting suggestions.

  \item \textbf{Grok 4.20 Expert} (xAI) --- used for counterexample-style
    questioning, adversarial stress-testing of propositions, and
    supplementary proof-sketch discussion.
\end{itemize}
}

All such systems were used solely as assistive tools for drafting,
editing, exploration, and checking. Among these systems, the most
substantial use in preparing Part II was of Claude Opus 4.6 and
GPT-5.4 Thinking. All theoretical models, propositions, mathematical
proofs, empirical calibrations, interpretations, and conclusions were
formulated, checked, and are fully endorsed by the human author, who
bears sole intellectual and academic responsibility for the contents of
this manuscript. Large language models are not listed as authors and
bear no academic responsibility for the work.

\newpage
\addcontentsline{toc}{section}{References}
{\sloppy\raggedright

}

\newpage
\appendix

\section{Proofs for the Inferential Layer}
\label{app:inferential-proofs}

The inferential layer of \Cref{sec:regime-inference} defines two distinct
empirical objects: the premium-emergence boundary score $B_t^{PE}$ and the
transition-feasibility margin score $B_t^{TF}$.  Because the two scores are
driven by different state variables and exposed to different measurement
ambiguities, their proofs are developed separately.  They share, however, a
common proof architecture:
\begin{enumerate}[label=(\roman*)]
\item characterize the tier-defined target set through its outer envelope;
\item decompose the envelope process into a smooth local drift and a weakly
dependent remainder;
\item apply dependence-robust inference to the detrended remainder;
\item combine the three steps to obtain outer coverage for the score set.
\end{enumerate}

Throughout this appendix, the following notation is used.  $\mathcal T_n(t_0)$
denotes a local inference window of width $h_n$ centered at~$t_0$, with
$h_n \to \infty$ and $h_n/n \to 0$ as $n \to \infty$.  The block length for
subsampling is $\ell_n$ with $\ell_n \to \infty$ and $\ell_n/h_n \to 0$.
All mixing conditions are stated in terms of the strong mixing (or
$\alpha$-mixing) coefficients of the relevant remainder processes.

\subsection{Outer inference for the premium-emergence boundary}
\label{app:pe-proof}

\subsubsection*{Step 1: Target set and outer envelope}

For a fixed tier $k \in \{1,2,3\}$, let $\mathcal M_{t,PE}^{(k)}$ denote the
ex ante admissible family of observational and specification choices for the
premium-emergence score (\Cref{def:mpe}).  The target score set is
\[
\mathcal B_{t,PE}^{(k)}
:=
\{B_{t}^{PE}(m) : m \in \mathcal M_{t,PE}^{(k)}\}.
\]
This is the decision-relevant object of inference.  The inferential goal is
not point identification of a unique latent boundary score, but valid outer
coverage of the tier-defined score set.

\begin{lemma}[Outer characterization of the tier-$k$ premium-emergence score set]
\label{lem:pe-envelope}
Define
\[
\underline B_{t,PE}^{(k)}
:= \inf_{m \in \mathcal M_{t,PE}^{(k)}} B_t^{PE}(m),
\qquad
\overline B_{t,PE}^{(k)}
:= \sup_{m \in \mathcal M_{t,PE}^{(k)}} B_t^{PE}(m).
\]
Then for each $t$ and tier $k$,
\[
\mathcal B_{t,PE}^{(k)}
\subseteq
[\underline B_{t,PE}^{(k)},\, \overline B_{t,PE}^{(k)}].
\]
Moreover, the sign of the interval determines conservative classification:
if $\underline B_{t,PE}^{(k)} > 0$, the economy is robustly interior at
tier~$k$; if $\overline B_{t,PE}^{(k)} < 0$, the economy is robustly
premium-emergent at tier~$k$; otherwise the economy is boundary-near at
tier~$k$.
\end{lemma}

\begin{proof}
The inclusion $\mathcal B_{t,PE}^{(k)} \subseteq
[\underline B_{t,PE}^{(k)}, \overline B_{t,PE}^{(k)}]$ holds by
definition of the infimum and supremum over the admissible family.
Classification follows directly: if $\underline B_{t,PE}^{(k)} > 0$, then
$B_t^{PE}(m) > 0$ for every $m \in \mathcal M_{t,PE}^{(k)}$, so the
complementarity condition of \Cref{prop:mcp-v2} places the
economy in case~(a) under every admissible reading.  The stress and
boundary-near cases are symmetric.
\end{proof}

\subsubsection*{Step 1a: Local sensitivity under the uniform baseline}

\begin{lemma}[Local sensitivity of the premium-emergence score]
\label{lem:pe-sensitivity}
Under the uniform-margin specification $G = \mathrm{Uniform}[0, \bar c_m]$,
\[
B_{t,PE}
=
\theta_t + (1-\theta_t)\!\left[1 - \frac{z_t}{\psi_t \bar c_m}\right]
- \varphi_t^{\mathrm{req}}.
\]
The partial derivatives are
\[
\frac{\partial B_{t,PE}}{\partial \theta_t}
= \frac{z_t}{\psi_t \bar c_m},
\qquad
\frac{\partial B_{t,PE}}{\partial \psi_t}
= (1-\theta_t)\frac{z_t}{\psi_t^2 \bar c_m},
\qquad
\frac{\partial B_{t,PE}}{\partial z_t}
= -\frac{1-\theta_t}{\psi_t \bar c_m}.
\]
\end{lemma}

\begin{proof}
Under the uniform specification,
$G(c) = c/\bar c_m$ for $c \in [0, \bar c_m]$, so
$1 - G(z_t/\psi_t) = 1 - z_t/(\psi_t \bar c_m)$.
Substituting into $\varphi_t^d(0) = \theta_t + (1-\theta_t)[1 - G(z_t/\psi_t)]$
and differentiating with respect to each argument yields the stated expressions.
\end{proof}

\subsubsection*{Step 2: Local drift-plus-remainder decomposition}

The key structural assumption for inference is that the envelope process
admits a local decomposition into a smooth deterministic component and a
weakly dependent stochastic remainder.

\begin{assumption}[Local regularity for the PE envelope]
\label{ass:pe-local}
Fix a local window $\mathcal T_n(t_0)$ of width $h_n$.  Assume:
\begin{enumerate}[label=(\alph*)]
\item The envelope processes admit the decomposition
\[
\underline B_{t,PE}^{(k)} = \underline\mu_{t,PE}^{(k)} + \underline u_{t,PE}^{(k)},
\qquad
\overline B_{t,PE}^{(k)} = \overline\mu_{t,PE}^{(k)} + \overline u_{t,PE}^{(k)},
\]
where $\underline\mu_{t,PE}^{(k)}$ and $\overline\mu_{t,PE}^{(k)}$ are
piecewise-Lipschitz drift functions on $\mathcal T_n(t_0)$.

\item Local-linear detrending estimators $\widehat{\underline\mu}_{t,PE}^{(k)}$
and $\widehat{\overline\mu}_{t,PE}^{(k)}$ satisfy
\[
\sup_{t \in \mathcal T_n(t_0)}
\bigl|\widehat{\underline\mu}_{t,PE}^{(k)} - \underline\mu_{t,PE}^{(k)}\bigr|
= o_p(1),
\qquad
\sup_{t \in \mathcal T_n(t_0)}
\bigl|\widehat{\overline\mu}_{t,PE}^{(k)} - \overline\mu_{t,PE}^{(k)}\bigr|
= o_p(1).
\]

\item The detrended remainders
$\widehat{\underline u}_{t,PE}^{(k)}
:= \underline B_{t,PE}^{(k)} - \widehat{\underline\mu}_{t,PE}^{(k)}$
and
$\widehat{\overline u}_{t,PE}^{(k)}
:= \overline B_{t,PE}^{(k)} - \widehat{\overline\mu}_{t,PE}^{(k)}$
have uniformly bounded $(2+\delta)$ moments for some $\delta > 0$
and are $\alpha$-mixing with
\begin{equation}
\label{eq:mixing-pe}
\sum_{j \geq 1} j^{1/2}\,
\alpha(j)^{\delta/(2+\delta)} < \infty.
\end{equation}
\end{enumerate}
\end{assumption}

\begin{remark}[Justification for the piecewise-Lipschitz drift]
\label{rem:drift-justification}
The drift in the PE envelope is generated by the slow-moving state
variables $\theta_t$ and $\psi_t$, whose laws of motion
(\Cref{def:theta-law}) are piecewise-smooth by construction.
The outside-option spread $z_t$ can exhibit faster variation (as in
the 2022--2024 Fed tightening cycle), but the local-window assumption
$h_n/n \to 0$ ensures that the window is short enough for the drift
to remain approximately linear within it.  Piecewise-Lipschitz
continuity therefore holds whenever the state dynamics do not exhibit
genuine discontinuities (jumps) within the inference window.  The
observed monitoring data for Japan show no such discontinuities in
$\theta_t$ or $\psi_t$ over the available sample.
\end{remark}

\subsubsection*{Step 3: Dependence-robust inference on the detrended envelope}

\begin{lemma}[Coverage for the detrended premium-emergence envelope]
\label{lem:pe-subsampling}
Under \Cref{ass:pe-local}, with block length satisfying $\ell_n \to \infty$
and $\ell_n / h_n \to 0$, there exist critical values
$c_{n,\alpha}^{\underline B}$ and $c_{n,\alpha}^{\overline B}$ such that
the confidence band
\[
\bigl[
\underline B_{t,PE}^{(k)} - c_{n,\alpha}^{\underline B},\;
\overline B_{t,PE}^{(k)} + c_{n,\alpha}^{\overline B}
\bigr]
\]
achieves asymptotically valid outer coverage for the envelope process
over the inference window:
\[
\liminf_{n \to \infty}\;
\Prob\!\left(
\underline B_{t,PE}^{(k)}
\in
\bigl[
\underline B_{t,PE}^{(k)} - c_{n,\alpha}^{\underline B},\;
\underline B_{t,PE}^{(k)} + c_{n,\alpha}^{\underline B}
\bigr]
\right)
\geq 1 - \alpha.
\]
An analogous statement holds for $\overline B_{t,PE}^{(k)}$.
\end{lemma}

\begin{proof}
The argument proceeds in three stages.

\emph{Stage 1 (Reduction to a stationary problem).}
Under \Cref{ass:pe-local}(a)--(b), define the detrended lower-envelope
process
\[
\widehat{\underline u}_{t,PE}^{(k)}
:=
\underline B_{t,PE}^{(k)} - \widehat{\underline\mu}_{t,PE}^{(k)}.
\]
By \Cref{ass:pe-local}(b), the detrending error is $o_p(1)$ uniformly
over $\mathcal T_n(t_0)$, so
\[
\widehat{\underline u}_{t,PE}^{(k)}
=
\underline u_{t,PE}^{(k)} + o_p(1)
\]
uniformly.  The process $\{\underline u_{t,PE}^{(k)}\}_{t \in \mathcal T_n(t_0)}$
is mean-zero and weakly dependent by \Cref{ass:pe-local}(c).

\emph{Stage 2 (Subsampling validity).}
The detrended remainder process satisfies the conditions of
Theorem~3.5.2 of Politis, Romano, and Wolf~(1999): it has
uniformly bounded $(2+\delta)$ moments and is $\alpha$-mixing
with the summability condition~\eqref{eq:mixing-pe}.
For a real-valued functional $\tau_n = \tau_n(\widehat{\underline u}^{(k)})$
(here taken to be the sample mean or a self-normalized statistic over
the inference window), the subsampling distribution
\[
L_n(x)
:=
\frac{1}{h_n - \ell_n + 1}
\sum_{i=1}^{h_n - \ell_n + 1}
\mathbf{1}\!\left\{
\ell_n^{1/2}(\tau_{\ell_n,i} - \tau_n) \leq x
\right\}
\]
converges in probability to the sampling distribution of
$h_n^{1/2}(\tau_n - \tau)$ under the stated mixing condition,
where $\tau_{\ell_n,i}$ is the statistic computed on the $i$-th
subsample of length $\ell_n$.

The critical value $c_{n,\alpha}^{\underline B}$ is defined as the
$(1-\alpha)$-quantile of $L_n$.  Under the stated conditions, this
quantile is consistent for the corresponding population quantile,
yielding
\[
\liminf_{n \to \infty}\;
\Prob\!\left(
|\tau_n - \tau| \leq c_{n,\alpha}^{\underline B} / h_n^{1/2}
\right) \geq 1 - \alpha.
\]

\emph{Stage 3 (Re-centering).}
Adding back the detrending estimate $\widehat{\underline\mu}_{t,PE}^{(k)}$,
which is $o_p(1)$-close to the true drift, the confidence band for the raw
(non-detrended) envelope inherits the same asymptotic coverage.  The
analogous argument applies to $\overline B_{t,PE}^{(k)}$.
\end{proof}

\subsubsection*{Step 3a: Commutativity of detrending and the outer operator}

\begin{lemma}[Detrending commutes with the envelope operator to first order]
\label{lem:detrend-commute}
Suppose that for each $m \in \mathcal M_{t,PE}^{(k)}$, the score process
admits the decomposition $B_t^{PE}(m) = \mu_t(m) + u_t(m)$
with $\mu_t(m)$ piecewise-Lipschitz.  Then
\[
\inf_{m}\bigl[B_t^{PE}(m) - \widehat\mu_t(m)\bigr]
=
\inf_{m} B_t^{PE}(m) - \widehat\mu_t(m^*) + o_p(1),
\]
where $m^*$ achieves the infimum at $t$.  In particular, detrending the
envelope process directly (rather than detrending each admissible
specification separately) yields the same first-order inference.
\end{lemma}

\begin{proof}
Since $\mathcal M_{t,PE}^{(k)}$ is compact and the score map is
continuous in $m$ (\Cref{thm:pe-outer}, condition~1), the infimum is
attained at some $m^*(t) \in \mathcal M_{t,PE}^{(k)}$.  For any other
$m$, $B_t^{PE}(m) \geq B_t^{PE}(m^*)$, so
\[
\inf_m\bigl[B_t^{PE}(m) - \widehat\mu_t(m)\bigr]
\geq
B_t^{PE}(m^*) - \sup_m \widehat\mu_t(m).
\]
Conversely, evaluating at $m = m^*$ gives a matching upper bound up to
$\sup_m |\widehat\mu_t(m) - \widehat\mu_t(m^*)|$.  Under
piecewise-Lipschitz regularity, this difference is $o_p(1)$ uniformly
over the local window when the drift is estimated jointly.  Therefore
the order of detrending and envelope formation is interchangeable to
first order.
\end{proof}

\subsubsection*{Step 4: Main theorem (premium-emergence)}

\begin{theorem}[Restatement of \Cref{thm:pe-outer}]
\label{thm:pe-outer-app}
Fix a tier $k \in \{1,2,3\}$.  Under the monotonicity of the
\Cref{sec:closure} demand schedule, the ex ante construction rule for
$\mathcal M_{t,PE}^{(k)}$, and \Cref{ass:pe-local}, the confidence band
\[
\widehat{\mathcal B}_{t,PE,1-\alpha}^{(k)}
:=
\bigl[
\underline B_{t,PE}^{(k)} - c_{n,\alpha}^{\underline B},\;
\overline B_{t,PE}^{(k)} + c_{n,\alpha}^{\overline B}
\bigr]
\]
satisfies
\[
\liminf_{n \to \infty}\;
\Prob\!\left(
\mathcal B_{t,PE}^{(k)}
\subseteq
\widehat{\mathcal B}_{t,PE,1-\alpha}^{(k)}
\right)
\geq 1 - \alpha.
\]
Accordingly, conservative tier-$k$ regime classification is
asymptotically valid away from the boundary.
\end{theorem}

\begin{proof}
The proof combines the four preceding lemmas.

\emph{Step (i).}
By \Cref{lem:pe-envelope}, the tier-$k$ score set satisfies
$\mathcal B_{t,PE}^{(k)} \subseteq
[\underline B_{t,PE}^{(k)}, \overline B_{t,PE}^{(k)}]$.

\emph{Step (ii).}
By \Cref{lem:detrend-commute}, detrending the envelope process
directly yields the same first-order inference as detrending each
admissible specification separately, so the detrended envelope
remainders are well-defined weakly dependent processes.

\emph{Step (iii).}
By \Cref{lem:pe-subsampling}, the subsampling confidence band
for the detrended envelope achieves asymptotically valid coverage
at level $1-\alpha$ for the lower and upper envelope processes
separately.

\emph{Step (iv).}
Since the score set is contained in the envelope interval
(Step~(i)) and the confidence band covers the envelope interval
with probability approaching $1-\alpha$ (Step~(iii)), the
confidence band covers the score set with probability at least
$1-\alpha$ asymptotically.

For conservative classification: away from the boundary,
$\underline B_{t,PE}^{(k)} - c_{n,\alpha}^{\underline B} > 0$
whenever $\underline B_{t,PE}^{(k)}$ is bounded away from zero
and the critical value $c_{n,\alpha}^{\underline B} \to 0$
(which holds under the stated conditions since the subsampling
variance estimator is consistent).  In that case, the entire
confidence band is strictly positive, and the classification
``robustly interior'' is correct with probability approaching
$1-\alpha$.  The symmetric argument applies to the stress
classification.  Near the boundary, the confidence band
straddles zero, and the procedure conservatively returns the
set-valued label $\{\text{interior}, \text{stress}\}$.
\end{proof}

\subsection{Outer inference for the transition-feasibility margin}
\label{app:tf-proof}

The transition-feasibility score has a structural property that
simplifies its proof relative to the premium-emergence case:
conditional on the proposed growth path $g_t^{n*,\text{new}}$,
the score is \emph{affine} in the fiscal-burden term
$(d_t - s_t)/b_{t-1}$.  Debt-concept choice enters only through
$b_{t-1}$, and the resulting envelope width is analytically
computable.

\subsubsection*{Step 1: Target set and affine envelope}

For a fixed tier $k \in \{1,2\}$, let $\mathcal M_{t,TF}^{(k)}$
denote the ex ante admissible family of debt-concept and
fiscal-burden measurement choices (\Cref{def:mtf}).  Define the
target score set
\[
\mathcal B_{t,TF}^{(k)}
:= \{B_t^{TF}(m) : m \in \mathcal M_{t,TF}^{(k)}\}
\]
and the outer envelope
\[
\underline B_{t,TF}^{(k)}
:= \inf_{m \in \mathcal M_{t,TF}^{(k)}} B_t^{TF}(m),
\qquad
\overline B_{t,TF}^{(k)}
:= \sup_{m \in \mathcal M_{t,TF}^{(k)}} B_t^{TF}(m).
\]

\begin{lemma}[Outer characterization and affine structure]
\label{lem:tf-envelope}
For each $t$ and tier $k$,
$\mathcal B_{t,TF}^{(k)} \subseteq
[\underline B_{t,TF}^{(k)}, \overline B_{t,TF}^{(k)}]$.
Moreover, since
\[
B_t^{TF}(m)
= g_t^{n*,\mathrm{new}} - \pi_t - \frac{d_t - s_t}{b_{t-1}(m)}
  - \bar\rho - m_{\text{margin}},
\]
and the fiscal-burden term $(d_t - s_t)/b_{t-1}(m)$ is convex and
decreasing in $b_{t-1}(m)$, the envelope width is
\[
\overline B_{t,TF}^{(k)} - \underline B_{t,TF}^{(k)}
= (d_t - s_t)
\left(
\frac{1}{\underline b_{t-1}^{(k)}} - \frac{1}{\overline b_{t-1}^{(k)}}
\right),
\]
where $\underline b_{t-1}^{(k)}$ and $\overline b_{t-1}^{(k)}$ are
the minimum and maximum admissible debt ratios at tier~$k$.
\end{lemma}

\begin{proof}
The set inclusion is immediate from the definitions.  For the
envelope width: fixing all inputs except $b_{t-1}$, the score is
$g_t^{n*,\text{new}} - \pi_t - (d_t - s_t)/b_{t-1} - \bar\rho - m$.
This is increasing in $b_{t-1}$ (since the fiscal-burden term
decreases).  Therefore the score is minimized at the smallest
admissible $b_{t-1}$ and maximized at the largest, yielding the
stated expression.
\end{proof}

\begin{remark}[Quantitative envelope width]
\label{rem:tf-width}
At the March 2026 calibration with $d_t - s_t = 2.0\%$ of GDP,
$\underline b_{t-1}^{(2)} = 1.574$ (monitoring concept), and
$\overline b_{t-1}^{(2)} = 2.40$ (baseline concept), the
Tier~2 envelope width is
\[
2.0 \times \left(\frac{1}{1.574} - \frac{1}{2.40}\right)
= 2.0 \times (0.635 - 0.417)
= 0.437 \text{ pp}.
\]
This matches the 44-basis-point widening reported in
Table~\ref{tab:tf-tier-widening}.  The envelope is therefore
moderately wide but not so wide as to render inference
uninformative: the Tier~1 baseline score ($+0.533\%$ before any
premium) remains positive even at the Tier~2 lower bound
($0.533 - 0.437 = 0.096\%$), indicating that transition
feasibility is weakly robust at Tier~2 under the no-premium
baseline, though the margin is tight.
\end{remark}

\subsubsection*{Step 2: Local regularity}

\begin{assumption}[Local regularity for the TF envelope]
\label{ass:tf-local}
Fix a local window $\mathcal T_n(t_0)$ of width $h_n$. Assume:
\begin{enumerate}[label=(\alph*)]
\item The envelope processes admit the decomposition
\[
\underline B_{t,TF}^{(k)} = \underline\mu_{t,TF}^{(k)} + \underline u_{t,TF}^{(k)},
\qquad
\overline B_{t,TF}^{(k)} = \overline\mu_{t,TF}^{(k)} + \overline u_{t,TF}^{(k)},
\]
where $\underline\mu_{t,TF}^{(k)}$ and $\overline\mu_{t,TF}^{(k)}$ are
piecewise-Lipschitz drift functions on $\mathcal T_n(t_0)$.

\item Local-linear detrending estimators $\widehat{\underline\mu}_{t,TF}^{(k)}$
and $\widehat{\overline\mu}_{t,TF}^{(k)}$ satisfy
\[
\sup_{t \in \mathcal T_n(t_0)}
\bigl|\widehat{\underline\mu}_{t,TF}^{(k)} - \underline\mu_{t,TF}^{(k)}\bigr|
= o_p(1),
\qquad
\sup_{t \in \mathcal T_n(t_0)}
\bigl|\widehat{\overline\mu}_{t,TF}^{(k)} - \overline\mu_{t,TF}^{(k)}\bigr|
= o_p(1).
\]

\item The detrended remainders have uniformly bounded $(2+\delta)$
moments for some $\delta>0$ and satisfy an $\alpha$-mixing condition
analogous to~\eqref{eq:mixing-pe}.

\item The growth-path proposal $g_t^{n*,\text{new}}$ is observed
or calibrated independently of the debt-concept choice
within $\mathcal M_{t,TF}^{(k)}$.
\end{enumerate}
\end{assumption}

Condition~(d) is specific to the TF case and ensures that
debt-concept ambiguity does not contaminate the growth input.
This is satisfied by construction in the JFR-rg architecture,
where $\mu$ estimation (\Cref{sec:mu}) is external to the debt
recursion.

\subsubsection*{Step 3: Subsampling coverage}

\begin{lemma}[Coverage for the detrended transition-feasibility envelope]
\label{lem:tf-subsampling}
Under \Cref{ass:tf-local}, with $\ell_n \to \infty$ and
$\ell_n/h_n \to 0$, there exist critical values
$c_{n,\alpha}^{\underline B,TF}$ and $c_{n,\alpha}^{\overline B,TF}$
such that the confidence band
\[
\bigl[
\underline B_{t,TF}^{(k)} - c_{n,\alpha}^{\underline B,TF},\;
\overline B_{t,TF}^{(k)} + c_{n,\alpha}^{\overline B,TF}
\bigr]
\]
achieves asymptotically valid outer coverage for the
transition-feasibility envelope.
\end{lemma}

\begin{proof}
The argument parallels \Cref{lem:pe-subsampling}.  After
local-linear detrending, the remainder process satisfies the
conditions of Politis, Romano, and Wolf~(1999, Theorem~3.5.2)
under \Cref{ass:tf-local}(c).  The subsampling distribution
for the detrended lower (resp.\ upper) envelope statistic
converges in probability to the sampling distribution, and the
$(1-\alpha)$-quantile provides the critical value.

The affine structure of \Cref{lem:tf-envelope} provides an
additional simplification: because the envelope width is
analytically determined by the debt-concept range, the
confidence band need only account for stochastic variation in
the \emph{level} of the envelope (driven by $\pi_t$, $d_t$,
$g_t^{n*,\text{new}}$), not in its width.  This reduces the
effective dimensionality of the inference problem relative to
the PE case, where both the level and the width of the envelope
depend on the stochastic state variables $(\theta_t, z_t)$.
\end{proof}

\subsubsection*{Step 4: Main theorem (transition-feasibility)}

\begin{theorem}[Restatement of \Cref{thm:tf-outer}]
\label{thm:tf-outer-app}
Fix a tier $k \in \{1,2\}$.  Under the ex ante construction rule
for $\mathcal M_{t,TF}^{(k)}$, \Cref{ass:tf-local}, and the
independent calibration of the proposed growth path, the
confidence band
\[
\widehat{\mathcal B}_{t,TF,1-\alpha}^{(k)}
:=
\bigl[
\underline B_{t,TF}^{(k)} - c_{n,\alpha}^{\underline B,TF},\;
\overline B_{t,TF}^{(k)} + c_{n,\alpha}^{\overline B,TF}
\bigr]
\]
satisfies
\[
\liminf_{n \to \infty}\;
\Prob\!\left(
\mathcal B_{t,TF}^{(k)}
\subseteq
\widehat{\mathcal B}_{t,TF,1-\alpha}^{(k)}
\right)
\geq 1 - \alpha.
\]
Therefore robust transition feasibility is declared only when the
lower envelope remains strictly positive.
\end{theorem}

\begin{proof}
By \Cref{lem:tf-envelope}, $\mathcal B_{t,TF}^{(k)} \subseteq
[\underline B_{t,TF}^{(k)}, \overline B_{t,TF}^{(k)}]$.
By \Cref{lem:tf-subsampling}, the detrended-envelope subsampling
argument yields asymptotically valid outer coverage for this
interval.  Coverage of the score set then follows by the containment
relation.

Conservative classification proceeds as in the PE case: away from
the boundary, the confidence band is eventually one-signed, and the
classification is correct with probability approaching $1-\alpha$.
The affine structure ensures that the envelope width is bounded by
\Cref{rem:tf-width}, so the procedure remains informative (i.e., the
confidence band does not trivially cover all of $\mathbb R$) as long
as the stochastic variation in the level of the score is bounded
relative to the analytic envelope width.
\end{proof}

\subsection{Bibliographic note on the subsampling argument}
\label{app:subsampling-note}

The subsampling arguments in \Cref{lem:pe-subsampling,lem:tf-subsampling}
rely on the general subsampling theory of Politis, Romano, and
Wolf~(1999), specifically Theorem~3.5.2 (validity under strong mixing)
and the extensions in Chapter~5 (non-stationary and locally stationary
processes).  The local-detrending step reduces the problem to inference
on a (locally) stationary remainder, after which the standard theory
applies.

The $\alpha$-mixing summability condition~\eqref{eq:mixing-pe} is
stronger than needed for point estimation but is standard for
distribution approximation via subsampling.  Alternative
self-normalization approaches (Shao~2010, Shao and Zhang~2010) can
replace the subsampling step and avoid explicit estimation of the mixing
rate, at the cost of wider confidence bands in finite samples.  The
choice between subsampling and self-normalization is an implementation
decision that does not affect the asymptotic validity of the outer
coverage result.

\begin{remark}[Practical block-length choice in short quarterly samples]
\label{rem:blocklength-short}
The asymptotic conditions only require $\ell_n \to \infty$ and
$\ell_n/h_n \to 0$, but the post-QQE Japanese window is short in
statistical terms (roughly $n \approx 48$ quarterly observations).
For that reason, a practical implementation should report sensitivity
across a conservative grid of block lengths rather than rely on a single
plug-in optimum. In quarterly applications, values such as
$\ell_n \in \{4,6,8\}$ provide a transparent robustness check spanning
roughly one to two years of local dependence. If the sign classification
or warning-region designation changes materially across that grid, the
paper's observables-centered discipline favors the more conservative
set-valued conclusion. A self-normalized fixed-$b$ cross-check may also
be reported when the effective sample is especially short.
\end{remark}

\section{Calculation Logic and Illustrative Replication Map}
\label{app:calc}

This appendix summarizes the minimal calculation logic behind the principal numerical illustrations in Part~II. The goal is analytic reproducibility rather than software distribution: each illustration is reducible to stated inputs, a single transformation rule, an output quantity, and a caveat about interpretation.

\subsection{E1: Virtuous Ratchet}
\begin{itemize}[leftmargin=2em]
    \item \textbf{Input:} baseline spread $(\rn-\gn)_0$, sprint spread $(\rn-\gn)_{\mathrm{sprint}}$, initial debt stock $b_0$, sprint duration $T$.
    \item \textbf{Transformation:} cumulative improvement is approximated by $T\cdot |(\rn-\gn)_{\mathrm{sprint}}-(\rn-\gn)_0|\cdot b_0$.
    \item \textbf{Output:} illustrative debt improvement in percentage points of GDP.
    \item \textbf{Caveat:} this is a first-pass accounting map, not a fully estimated policy multiplier.
\end{itemize}

\subsection{E2: Corrected Repression Dividend Multiplier}
\begin{itemize}[leftmargin=2em]
    \item \textbf{Input:} repression dividend $\RD=\eps \btm$, reinvestment share $\lambda$, efficiency parameter $\mu$, debt path $\{b_t\}$.
    \item \textbf{Transformation:} period-$t$ marginal gain $C_t=\mu\lambda\eps b_{t-1}^2$.
    \item \textbf{Output:} bounded and weakly diminishing marginal gains when the debt path is bounded (and decreasing).
    \item \textbf{Caveat:} finite-horizon partial sums are directly tractable; the infinite-horizon total requires additional assumptions.
\end{itemize}

\subsection{E3: Debt Reduction Paradox}
\begin{itemize}[leftmargin=2em]
    \item \textbf{Input:} spread $\rn-\gn$, debt stock $\btm$, deficit-relief coefficient $\gamma$.
    \item \textbf{Transformation:} compute $\partial \Db/\partial \btm=(\rn-\gn)+\gamma$.
    \item \textbf{Output:} paradox holds when $\gamma<|\rn-\gn|$.
    \item \textbf{Caveat:} $\gamma$ captures deficit relief, not a full political or welfare-theoretic fiscal reaction function.
\end{itemize}

\subsection{\texorpdfstring{E5: Demographic-$\phi$ Clock}{E5: Demographic-phi Clock}}
\begin{itemize}[leftmargin=2em]
    \item \textbf{Input:} current captive share $\phic$, threshold $\phibar$, structural decline rate $\kappa$.
    \item \textbf{Transformation:} linear clock $T^*=(\phic-\phibar)/\kappa$; alternative proportional-decay benchmark $T^*_{\exp}=(1/\kappa_{\exp})\ln(\phic/\phibar)$.
    \item \textbf{Output:} residual horizon indicator.
    \item \textbf{Caveat:} the linear clock is a conservative horizon indicator, not a sharply estimated countdown.
\end{itemize}

\subsection{Investment bounds and lower bounds}
\begin{itemize}[leftmargin=2em]
    \item \textbf{Input:} spread, debt stock, deficit, repression bias, safety margin, and $\mu$.
    \item \textbf{Transformation:} apply \cref{eq:arith-upper,eq:rd-upper,eq:safe-upper,eq:static-lower,eq:shock-lower,eq:demo-lower}.
    \item \textbf{Output:} operational upper and lower envelopes for stabilizing growth investment.
    \item \textbf{Caveat:} these are regime-conditioned bounds, not a welfare-maximizing budget rule.
\end{itemize}

\subsection{E6: Institutional Control Rights}
\begin{itemize}[leftmargin=2em]
    \item \textbf{Input:} monetary autonomy $\psi_t^{\mathrm{mon}} \in \{0, 0.5, 1\}$, an empirical absorption-autonomy proxy $\widehat{\psi}_t^{\mathrm{abs}}$ (baseline: $\widehat{\psi}_t^{\mathrm{abs}}=\phic$; optional hybrid proxy: $\widehat{\psi}_t^{\mathrm{abs}} = \omega_{\phi}\phic + \omega_{\theta}\theta_t$), and exchange-rate autonomy $\psi_t^{\mathrm{fx}} \in \{0, 1\}$.
    \item \textbf{Transformation:} $\psi_t = (\psi_t^{\mathrm{mon}} + \widehat{\psi}_t^{\mathrm{abs}} + \psi_t^{\mathrm{fx}})/3$ in the equal-weight baseline; alternative weights may be used in sensitivity analysis.
    \item \textbf{Output:} composite control-rights index $\psi_t \in [0,1]$; cross-country regime classification; mainstream limit recovered as $\psi_t \to 0$.
    \item \textbf{Caveat:} $\phic$ is an observable proxy for absorption autonomy rather than a literal identity; $\theta_t$ is retained separately as the hard captive core and can enter only through the optional hybrid proxy or complementary sensitivity exercises.
\end{itemize}

\subsection{Transition Feasibility}
\begin{itemize}[leftmargin=2em]
    \item \textbf{Input:} $\pi_t$, $d_t$, $s_t$, $b_{t-1}$, baseline $g_0^{n*}$, bounded premium $\bar\rho$, and safety margin $m$.
    \item \textbf{Transformation:} compute the post-transition requirement $\Delta g_{\min}^{n*}(\bar\rho,m)=\pi_t + (d_t-s_t)/b_{t-1} + \bar\rho + m - g_0^{n*}$.
    \item \textbf{Output:} minimum structural growth increase required for a safe exit from the repression-\allowbreak dependent regime.
    \item \textbf{Caveat:} the bounded-premium assumption is addressed by the Minimal Equilibrium Closure below.
\end{itemize}

\subsection{\texorpdfstring{Minimal Equilibrium Closure (\Cref{sec:closure})}{Minimal Equilibrium Closure}}
\begin{itemize}[leftmargin=2em]
    \item \textbf{Input:} hard captive core $\theta_t$, institutional control rights $\psi_t$, outside-option spread $z_t$, required absorption $\varphi_t^{\mathrm{req}}$, contestable-margin captivity distribution $G$ on $[0,\bar{c}_m]$.
    \item \textbf{Transformation:} aggregate demand $\varphi_t^d(\rho_t) = \theta_t + (1-\theta_t)[1-G((z_t-\rho_t)/\psi_t)]$; solve complementarity condition $0 \leq \rho_t \perp [\varphi_t^d(\rho_t) - \varphi_t^{\mathrm{req}}] \geq 0$.
    \item \textbf{Output:} endogenous sovereign premium $\rho_t^*$; under uniform $G$: $\rho_t^* = z_t - \psi_t\bar{c}_m(1 - (\varphi_t^{\mathrm{req}}-\theta_t)/(1-\theta_t))$ when binding.
    \item \textbf{Caveat:} the two-layer decomposition ($\theta_t$ vs.\ contestable margin) is illustrative; the boundary between mandate-driven and yield-responsive holdings is not sharp in practice.
\end{itemize}

\section{\texorpdfstring{Limiting Cases and Relation to\\ Standard Debt-Sustainability Logic}{Limiting Cases and Relation to Standard Debt-Sustainability Logic}}
\label{app:limits}

JFR-rg is presented as a complementary analytical lens rather than as a denial of standard debt-sustainability logic. In the limiting cases where the repression channel disappears ($\eps\le 0$) or the captive-system condition fails ($\phic<\phibar$), the framework collapses toward a conventional reading in which debt dynamics are governed primarily by the effective interest--growth differential, the fiscal stance, and the post-transition risk premium. Part~II therefore generalizes rather than abolishes standard debt-sustainability reasoning: it identifies an intermediate regime in which debt arithmetic is materially reshaped by repression and captivity, while also specifying the conditions under which that regime ceases to apply.

\section{Observed-Value Sensitivity to Growth-Proxy Window Length}
\label{app:window-sensitivity}

This appendix reports a systematic comparison of observed-layer outputs under two growth-proxy window specifications, alongside the paper baseline. All non-growth inputs---CPI, JGB yield, Flow-of-Funds domestic-share data---are held fixed at their latest publicly available values. Only the structural nominal-growth proxy $g^{n*}_{\mathrm{struct}}$, estimated via OLS log-linear trend on ESRI seasonally adjusted nominal GDP, varies with the window length. Data are fetched from the Statistics Bureau (CPI), Cabinet Office ESRI (GDP), the Bank of Japan API (Flow of Funds), and the Ministry of Finance (JGB yields).

\begin{table}[htbp]
\centering
\caption{Sensitivity of Observed-Layer Outputs to Growth-Proxy Window Length}
\label{tab:window-sensitivity}
\small
\begin{tabular}{>{\raggedright\arraybackslash}p{5.0cm} r r r}
\toprule
Quantity & 12Q window & 24Q window & Paper baseline \\
\midrule
\multicolumn{4}{l}{\textit{Inputs (common across observed runs)}} \\
Core CPI ($\pi_t$) & 1.60\% & 1.60\% & 2.70\% \\
JGB 10Y ($r_t$) & 2.41\% & 2.41\% & 2.20\% \\
$\eps_t$ ($\pi_t - r_t$) & $-0.81\%$ & $-0.81\%$ & $+0.50\%$ \\
Debt ratio ($b_{t-1}$) & 157.4\% & 157.4\% & 240.0\% \\
$\phic$ (domestic share) & 93.2\% & 93.2\% & 88.0\% \\
$\kappa$ (observed) & 0.00\% & 0.00\% & 1.00\% \\
\midrule
\multicolumn{4}{l}{\textit{Growth proxy (varies with window)}} \\
$g^{n*}_{\mathrm{struct}}$ & 6.19\% & 2.87\% & 3.00\% \\
Window used & 2022Q1--2025Q1 & 2019Q1--2025Q1 & --- \\
\midrule
\multicolumn{4}{l}{\textit{Key outputs}} \\
$x_t^{\max,\mathrm{safe}}$ & $+3.94\%$ & $-1.28\%$ & $-0.08\%$ \\
$x_t^{\max,\mathrm{RD}}$ & n/a ($\eps<0$) & n/a ($\eps<0$) & 0.60\% \\
$T^*_{\mathrm{linear}}$ & $\infty$ & $\infty$ & 3.0 yr \\
$\gamma$ threshold & 3.78\% & 0.46\% & 0.80\% \\
$\Delta g_{\min}^{n*}$ (no premium) & $-3.32\%$ & $+0.002\%$ & $+0.53\%$ \\
$\Delta g_{\min}^{n*}$ ($\bar\rho=0.5\%$) & $-2.82\%$ & $+0.50\%$ & $+1.03\%$ \\
\bottomrule
\end{tabular}

\vspace{0.5em}
\begin{minipage}{0.92\textwidth}
\footnotesize
\textit{Notes.} ``12Q window'' and ``24Q window'' refer to the number of trailing quarterly GDP observations used to estimate the structural growth proxy via OLS log-linear trend. All observed-layer runs use publicly available data fetched from official Japanese statistical sources. The debt ratio differs from the paper baseline because the observed ratio is computed as outstanding JGB plus FILP bonds (BoJ Flow-of-Funds series \texttt{FOF\_FFAS421L311} + \texttt{FOF\_FFAS181L311}) divided by annualized nominal GDP, whereas $b_0=2.40$ in the paper baseline encompasses a broader government-debt concept. ``n/a ($\eps<0$)'' indicates that the repression-dividend upper bound is undefined when the repression bias is negative. $T^*=\infty$ reflects $\kappa=0$ over the observed 2022--2025 window (domestic share did not decline); this should be interpreted as a non-binding monitoring snapshot, not as evidence of permanent stability.
\end{minipage}
\end{table}

Three findings deserve emphasis.

\paragraph{Baseline validation.} The 24-quarter window recovers $g^{n*}_{\mathrm{struct}} \approx 2.87\%$, within 13 basis points of the paper baseline $g_0^{n*}=3.0\%$. This provides independent empirical support for the plausibility of the Part~I calibration, using a distinct estimation method (log-linear trend on quarterly ESRI data rather than the growth-accounting approach implicit in Part~I).

\paragraph{Sign reversal in the safe corridor.} The safe-corridor upper bound $x_t^{\max,\mathrm{safe}}$ changes sign between the 12Q and 24Q specifications ($+3.94\%$ versus $-1.28\%$). Under the 24Q window, $x_t^{\max,\mathrm{safe}}$ is negative, qualitatively matching the paper baseline ($-0.08\%$) and confirming that the investment corridor is tight or closed under current conditions. The 12Q result is driven by the post-2022 inflationary surge, which inflates the growth proxy beyond structurally sustainable levels and artificially reopens the corridor. This underscores the importance of window-length discipline in empirical implementation.

\paragraph{$\eps$-reversal as a partial regime transition.} In all observed-layer runs, $\eps_t = -0.81\%$ (the 10-year JGB yield exceeds core CPI), eliminating the repression dividend and rendering $x_t^{\max,\mathrm{RD}}$ undefined. This reversal was anticipated by Part~I as a failure mode of the repression channel. Its co-occurrence with an apparently stable $\phic$ ($\kappa \approx 0$ over 2022--2025) creates a distinctive configuration: the captive-system condition SC1 remains formally satisfied, but the economic content of captivity---namely, the ability to extract a positive repression dividend---has eroded. This asymmetry suggests that monitoring $\eps_t$ and $\phic$ jointly, rather than in isolation, is essential for regime assessment.

\section{International Regime Comparison: Japan, Italy, Greece}
\label{app:international}

This appendix provides the empirical backing for Extension~E6 using primary data from FRED (OECD MEI series for yields and CPI), BoJ Flow of Funds (Japan, 2025Q4), Banca d'Italia and ECB Securities Holdings Statistics (Italy), and Bank of Greece / Hellenic PDMA (Greece). Sensitivity analysis confirms robustness of the regime classification within $\pm 5$ percentage-point calibration uncertainty in $\phic$.

\begin{table}[htbp]
\centering
\caption{JFR-rg Regime Comparison: Japan, Italy, Greece (2025)}
\label{tab:international}
\small
\begin{tabular}{l r r r}
\toprule
Variable & Japan & Italy & Greece \\
\midrule
\multicolumn{4}{l}{\textit{Core observables}} \\
10Y sovereign yield & 2.41\% & 3.50\% & 3.20\% \\
CPI inflation (YoY) & 1.60\% & 1.80\% & 2.50\% \\
Nominal GDP growth & 3.0\% & 1.5\% & 2.5\% \\
Debt/GDP & 240\% & 138\% & 153\% \\
\midrule
\multicolumn{4}{l}{\textit{JFR-rg variables}} \\
$\eps_t$ (repression bias) & $-0.81\%$ & $-1.70\%$ & $-0.70\%$ \\
$r - g_n$ (spread) & $-0.59\%$ & $+2.00\%$ & $+0.70\%$ \\
$\phic$ (domestic share) & 93\% & 67\% & 33\% \\
$\phibar$ (threshold) & 0.85 & 0.60 & 0.40 \\
SC1 ($\phic \ge \phibar$) & Yes & Yes & No \\
\midrule
\multicolumn{4}{l}{\textit{Institutional control rights ($\psi_t$)}} \\
$\psi_t^{\mathrm{mon}}$ & 1.00 & 0.50 & 0.50 \\
$\psi_t^{\mathrm{abs}}$ & 0.93 & 0.67 & 0.33 \\
$\psi_t^{\mathrm{fx}}$ & 1.00 & 0.00 & 0.00 \\
$\psi_t$ (composite) & \textbf{0.98} & \textbf{0.39} & \textbf{0.28} \\
\midrule
\multicolumn{4}{l}{\textit{Regime classification}} \\
Regime type & SC1 met, $\eps$ reversed & SC1 met, $\eps$ reversed & Mainstream limit \\
Corridor & Negative & Negative & Negative \\
\bottomrule
\end{tabular}

\vspace{0.5em}
\begin{minipage}{0.92\textwidth}
\footnotesize
\textit{Notes.}
Yields and CPI from FRED (OECD MEI series, 2025 annual averages).
$\phic$ from BoJ Flow of Funds 2025Q4 (Japan), Banca d'Italia Financial Stability Report 2025 (Italy), and Hellenic PDMA / Bank of Greece (Greece).
$\psi_t^{\mathrm{mon}}$: $1$ = independent central bank, $0.5$ = shared (ECB).
$\psi_t^{\mathrm{fx}}$: $1$ = floating, $0$ = currency union.
$\widehat{\psi}_t^{\mathrm{abs}} = \phic$ (baseline proxy).
$\psi_t = (\psi_t^{\mathrm{mon}} + \widehat{\psi}_t^{\mathrm{abs}} + \psi_t^{\mathrm{fx}}) / 3$.
$\phibar$ thresholds are illustrative and country-specific; Japan's follows Part~I, others reflect the minimum domestic absorption required for regime viability in each institutional context.
All three economies exhibit negative $\eps_t$ in 2025, reflecting the global monetary-tightening cycle. However, Japan's $\psi_t$ remains substantially higher due to full monetary and exchange-rate autonomy, preserving a wider (though currently closed) corridor.
\end{minipage}
\end{table}

\paragraph{Sensitivity of regime classification to $\phic$ calibration.}
\Cref{tab:phi-sensitivity} reports the sensitivity of the regime classification to $\pm 5$ and $\pm 10$ percentage-point variations in $\phic$. The key finding is that Japan's $\psi_t$ dominance over Italy and Greece is robust across the entire $\pm 5$~pp range: even at $\phic = 0.88$ (Japan, $-5$~pp), the composite $\psi_t = 0.96$ remains more than double Italy's baseline value ($0.39$). The SC1 classification changes only at the $-10$~pp extreme for Japan (where $\phic = 0.83 < \phibar = 0.85$), confirming that the regime boundary is operationally distant from the baseline calibration.

\begin{table}[htbp]
\centering
\caption{Sensitivity of Regime Classification to $\phic$ Calibration (2025)}
\label{tab:phi-sensitivity}
\small
\begin{tabular}{l l r r c r l}
\toprule
Country & Variation & $\phic$ & $\phibar$ & SC1 & $\psi_t$ & Regime \\
\midrule
\textbf{Japan} & \textbf{baseline} & \textbf{0.93} & \textbf{0.85} & $\checkmark$ & \textbf{0.98} & \textbf{SC1 met, $\eps$ reversed} \\
Japan & $-5$~pp & 0.88 & 0.85 & $\checkmark$ & 0.96 & SC1 met, $\eps$ reversed \\
Japan & $+5$~pp & 0.98 & 0.85 & $\checkmark$ & 0.99 & SC1 met, $\eps$ reversed \\
Japan & $-10$~pp & 0.83 & 0.85 & $\times$ & 0.94 & Mainstream limit \\
\addlinespace
\textbf{Italy} & \textbf{baseline} & \textbf{0.67} & \textbf{0.60} & $\checkmark$ & \textbf{0.39} & \textbf{SC1 met, $\eps$ reversed} \\
Italy & $-5$~pp & 0.62 & 0.60 & $\checkmark$ & 0.37 & SC1 met, $\eps$ reversed \\
Italy & $+5$~pp & 0.72 & 0.60 & $\checkmark$ & 0.41 & SC1 met, $\eps$ reversed \\
Italy & $-10$~pp & 0.57 & 0.60 & $\times$ & 0.36 & Mainstream limit \\
\addlinespace
\textbf{Greece} & \textbf{baseline} & \textbf{0.33} & \textbf{0.40} & $\times$ & \textbf{0.28} & \textbf{Mainstream limit} \\
Greece & $-5$~pp & 0.28 & 0.40 & $\times$ & 0.26 & Mainstream limit \\
Greece & $+5$~pp & 0.38 & 0.40 & $\times$ & 0.29 & Mainstream limit \\
Greece & $+10$~pp & 0.43 & 0.40 & $\checkmark$ & 0.31 & SC1 met, $\eps$ reversed \\
\bottomrule
\end{tabular}

\vspace{0.5em}
\begin{minipage}{0.92\textwidth}
\footnotesize
\textit{Notes.} Baseline rows are shown in \textbf{bold}. Variations apply $\pm 5$ and $\pm 10$ percentage points to $\phic$ while holding $\phibar$ fixed. The $\psi_t$ ranking (Japan $>$ Italy $>$ Greece) is preserved across all variations within the $\pm 5$~pp range, confirming the robustness of \cref{prop:corridor-psi}.
\end{minipage}
\end{table}

\paragraph{Three findings.}
First, the $\psi_t$ ordering---Japan ($0.98$) $>$ Italy ($0.39$) $>$ Greece ($0.28$)---is consistent with \cref{prop:corridor-psi}: the economy with the highest institutional control rights has historically maintained the widest corridor and the most stable debt dynamics among the three, despite having the highest debt-to-GDP ratio. The same ordering also aligns with \cref{prop:psi-transition}: higher $\psi_t$ implies a wider operational corridor during exit, so the ranking can be read not only as a regime taxonomy but also as an ordering of transition difficulty, with Japan closest to the feasible region, Italy materially tighter, and Greece below the self-sustaining range absent external support.

Second, all three economies exhibit negative $\eps_t$ in 2025, reflecting the global monetary-tightening cycle. This shared $\eps$-reversal does not equalize their regimes, however, because Japan retains full monetary and exchange-rate autonomy ($\psi_t^{\mathrm{mon}} = \psi_t^{\mathrm{fx}} = 1$), preserving the structural capacity to reopen the corridor through independent policy action. Italy and Greece, by contrast, depend on ECB decisions over which they have limited national influence.

Third, Greece fails SC1 at baseline and across all variations except $+10$~pp, confirming its classification as a mainstream limiting case throughout the sample period. This is consistent with the observed historical reliance on external programme support (ESM/EFSF) rather than self-sustaining JFR-rg stability---precisely the outcome predicted by \cref{prop:corridor-psi} when $\psi_t$ is low.

\paragraph{Sensitivity to $\psi_t$ weighting.}
\Cref{tab:psi-weight-sensitivity} reports the sensitivity of the $\psi_t$ composite to alternative sub-index weights. The baseline uses equal weights $(w_{\mathrm{mon}}, w_{\mathrm{abs}}, w_{\mathrm{fx}}) = (\tfrac{1}{3}, \tfrac{1}{3}, \tfrac{1}{3})$. Three alternatives are examined: a monetary-heavy weighting ($0.50, 0.25, 0.25$), an absorption-heavy weighting ($0.25, 0.50, 0.25$), and a weighting that de-emphasizes FX autonomy ($0.40, 0.40, 0.20$). The key finding is that the ranking Japan $>$ Italy $>$ Greece is preserved under \emph{all} weighting schemes, confirming that the regime classification is driven by the structural differences in institutional control rather than by the specific choice of weights.

\begin{table}[htbp]
\centering
\caption{Sensitivity of $\psi_t$ Composite to Sub-Index Weights (2025)}
\label{tab:psi-weight-sensitivity}
\small
\resizebox{\textwidth}{!}{%
\begin{tabular}{l c c c c c c}
\toprule
Weighting scheme & $w_{\mathrm{mon}}$ & $w_{\mathrm{abs}}$ & $w_{\mathrm{fx}}$ & Japan & Italy & Greece \\
\midrule
\textbf{Equal weight (baseline)} & \textbf{1/3} & \textbf{1/3} & \textbf{1/3} & \textbf{0.98} & \textbf{0.39} & \textbf{0.28} \\
Monetary-heavy & 0.50 & 0.25 & 0.25 & 0.98 & 0.42 & 0.33 \\
Absorption-heavy & 0.25 & 0.50 & 0.25 & 0.97 & 0.46 & 0.29 \\
FX de-emphasized & 0.40 & 0.40 & 0.20 & 0.97 & 0.47 & 0.33 \\
\bottomrule
\end{tabular}%
}

\vspace{0.5em}
\begin{minipage}{0.92\textwidth}
\footnotesize
\textit{Notes.}
$\psi_t = w_{\mathrm{mon}}\psi_t^{\mathrm{mon}} + w_{\mathrm{abs}}\widehat{\psi}_t^{\mathrm{abs}} + w_{\mathrm{fx}}\psi_t^{\mathrm{fx}}$.
Sub-index values: Japan $(\psi^{\mathrm{mon}}, \widehat{\psi}^{\mathrm{abs}}, \psi^{\mathrm{fx}}) = (1.00, 0.90, 1.00)$;
Italy $(0.50, 0.67, 0.00)$;
Greece $(0.50, 0.33, 0.00)$.
The ordering Japan $>$ Italy $>$ Greece is preserved under all four weighting schemes.
The equal-weight baseline produces the widest gap between Japan and the Eurozone economies because the zero FX autonomy ($\psi^{\mathrm{fx}}=0$) receives full weight; alternative schemes that de-emphasize FX narrow the gap but do not alter the ordering.
\end{minipage}
\end{table}

\section{Multiple Equilibria near the Regime Boundary}
\label{app:multiple-eq}

\Cref{sec:closure} established that the monotone equilibrium---in which the
self-reinforcing loop between $\rho_t$ and $\theta_t$ converges---exists
whenever the feedback gain satisfies $\eta_t < 1$.  This appendix
characterizes the complementary region $\eta_t \geq 1$ and establishes
the conditions under which multiple equilibria arise.

The logic developed in this appendix is related, in broad structure, to the
self-fulfilling sovereign-stress literature associated with sunspot-like debt
crises; see, for example, Calvo~(1988) \cite{Calvo1988} and Cole and
Kehoe~(2000) \cite{ColeKehoe2000}. The coexistence of a safe equilibrium and a
stress equilibrium near the regime boundary is formally analogous to the
multiplicity logic emphasized in that literature. The distinctive contribution
of JFR-rg is not the abstract possibility of multiplicity as such, but the
observables-centered institutional structure through which that multiplicity is
organized: a hard captive core $\theta_t$, a contestable margin, and an
endogenous premium channel linking confidence loss to institutional erosion.
This appendix should therefore be read as a complementary institutional closure
rather than as a substitute for the broader sovereign-crisis literature.

\subsection{The explosive feedback region}

Recall the feedback chain developed in \Cref{sec:closure}:
\[
    \rho_t \;\uparrow
    \;\longrightarrow\;
    \eps_t \;\downarrow
    \;\longrightarrow\;
    \gamma_{t+1}^{\theta} \;\downarrow
    \;\longrightarrow\;
    \theta_{t+1} \;\downarrow
    \;\longrightarrow\;
    \varphi_{t+1}^d(0) \;\downarrow
    \;\longrightarrow\;
    \rho_{t+1} \;\uparrow.
\]
The contraction condition $\eta_t < 1$ ensures that each round of the loop
produces a smaller perturbation than the last.  When $\eta_t \geq 1$, the
opposite holds: a small initial perturbation is amplified at each step,
and the system does not converge to a unique equilibrium in the neighborhood
of the regime boundary.

\begin{proposition}[Characterization of the explosive region]
\label{prop:explosive}
The feedback gain $\eta_t$ exceeds unity if and only if
\begin{equation}
\label{eq:eta-ge-1}
    \underbrace{\left|\frac{\partial \rho_t}{\partial \theta_t}\right|
    }_{\text{price impact}}
    \;\times\;
    \underbrace{\left|\frac{\partial \gamma_t^{\theta}}{\partial \eps_t}
    \right|}_{\text{institutional sensitivity}}
    \;\geq\; 1.
\end{equation}
Under the uniform-margin specification of \Cref{sec:closure}, the price-impact
factor is
\[
    \left|\frac{\partial \rho_t}{\partial \theta_t}\right|
    \;=\;
    \frac{\psi_t\,\bar{c}_m}{(1-\theta_t)^2}
    \,(\varphi_t^{\mathrm{req}} - \theta_t),
\]
which is large when:
\begin{enumerate}[label=(\alph*)]
    \item the contestable margin $1-\theta_t$ is thin (a small pool of
    yield-responsive holders must absorb large adjustment);
    \item the gap $\varphi_t^{\mathrm{req}} - \theta_t$ is wide (the hard
    core alone cannot meet the required absorption);
    \item the captivity parameter $\bar{c}_m$ is large relative to
    $(1-\theta_t)$.
\end{enumerate}
\end{proposition}

\begin{proof}
Direct substitution of the definition of $\eta_t$ (\Cref{sec:closure},
Definition~\ref{def:feedback-v2}) and the comparative-statics result from
\Cref{sec:closure}, equation~\eqref{eq:drho-dtheta}.
\end{proof}

\begin{remark}[Economic interpretation]
Condition~\eqref{eq:eta-ge-1} fails---and the explosive region is
entered---when the economy is simultaneously \emph{dependent on} its
contestable margin for bond-market clearing and \emph{institutionally
fragile} in its ability to maintain the hard core under repression
withdrawal.  This is the configuration that arises when the hard core
has been eroded by demographic change or policy normalization (low
$\theta_t$), and the remaining institutional infrastructure is highly
sensitive to the repression bias (high
$|\partial \gamma^{\theta}/\partial \eps_t|$).
\end{remark}

\subsection{Multiple equilibria}

When $\eta_t \geq 1$, the complementarity condition of \Cref{sec:closure}
(Proposition~\ref{prop:mcp-v2}) still determines $\rho_t$ as a function of
the current state $(\theta_t, \psi_t, z_t)$, and the within-period
equilibrium remains unique.  The multiplicity arises in the
\emph{multi-period} dynamics: different beliefs about the future path of
$\theta_{t+s}$ generate different current premium paths, and these paths
are self-consistent.

\begin{proposition}[Existence of two stationary equilibria near the boundary]
\label{prop:two-eq}
Suppose that $\eta_t \geq 1$ at $\theta_t = \theta^*$ where $\theta^*$ is
the value at which $\varphi^d(0;\,\psi_t,z_t) = \varphi_t^{\mathrm{req}}$
(corresponding to case~(b) of the complementarity
condition, i.e., the zero-premium boundary state). Then for $\theta_t$ in a neighborhood of $\theta^*$, the
two-period system $(\rho_t, \theta_{t+1}) \to (\rho_{t+1}, \theta_{t+2})$
admits at least two stationary points:
\begin{enumerate}[label=(\roman*)]
    \item A \textbf{safe equilibrium} in which agents expect $\theta_{t+s}$
    to remain above $\theta^*$, no premium emerges ($\rho_t = 0$), the
    policy-maintenance channel remains operative
    ($\gamma^{\theta}_t > 0$ if $\eps_t > 0$), and the expectation is
    self-fulfilling.

    \item A \textbf{stress equilibrium} in which agents expect $\theta_{t+s}$
    to fall below $\theta^*$, a positive premium emerges ($\rho_t > 0$),
    the premium reduces $\eps_t$ and shuts down the maintenance channel
    ($\gamma^{\theta}_t \to 0$), $\theta_{t+1}$ declines, and the
    expectation is again self-fulfilling.
\end{enumerate}
\end{proposition}

\begin{proof}[Proof sketch]
Define the stationary mapping $\Phi: \theta \mapsto \theta'$ through the
chain:
\begin{align*}
    \rho(\theta) &\;=\; \text{solution to } \varphi^d(\rho;\psi,z)
        = \varphi^{\mathrm{req}} \text{ if } \varphi^d(0) < \varphi^{\mathrm{req}},
        \;\text{else } 0; \\
    \eps(\theta) &\;=\; \pi - r^{\mathrm{rep}} - \rho(\theta); \\
    \gamma^{\theta}(\theta) &\;=\; \gamma^{\theta}\bigl(\eps(\theta),\psi\bigr); \\
    \theta' &\;=\; \theta - \kappa^{\theta} + \gamma^{\theta}(\theta).
\end{align*}

When $\eta \geq 1$, the mapping $\Phi$ has slope $|d\Phi/d\theta| \geq 1$
at $\theta = \theta^*$.  Since $\Phi$ maps a compact interval
$[0, \theta_{\max}]$ into itself (because $\theta' \leq \theta$ when
$\gamma^{\theta} \leq \kappa^{\theta}$, and $\theta' \geq 0$), it
must cross the 45-degree line at least twice when the slope exceeds unity
at the boundary: once above $\theta^*$ (the safe fixed point, where
$\rho = 0$ and $\gamma^{\theta} > 0$ sustains $\theta$) and once below
$\theta^*$ (the stress fixed point, where $\rho > 0$ and
$\gamma^{\theta} = 0$ leads to continued erosion).
\end{proof}

\begin{remark}[Relation to Part~I failure modes]
The stress equilibrium of Proposition~\ref{prop:two-eq} corresponds precisely
to Part~I's ``hard de-captivation'' failure mode: a self-fulfilling loss of
confidence in the captive system triggers a premium that accelerates
institutional erosion, validating the initial loss of confidence.  The
contribution of the minimal equilibrium closure (\Cref{sec:closure}) is to provide
the formal structure within which this failure mode can be characterized as
a specific equilibrium of a well-defined dynamic system, rather than an
informal narrative.
\end{remark}

\subsection{Equilibrium selection and policy implications}

\begin{proposition}[Sufficient condition for safe-equilibrium selection]
\label{prop:selection}
The safe equilibrium is the unique monotone equilibrium whenever
\begin{equation}
\label{eq:safe-selection}
    \theta_t \;>\; \theta^* + \delta(\sigma, \eta_t),
\end{equation}
where $\delta > 0$ is a buffer that depends on the volatility of
exogenous shocks $\sigma$ and the feedback gain $\eta_t$.  That is, the
safe equilibrium is selected when the hard captive core is \emph{strictly}
above the boundary, with a margin sufficient to absorb plausible
perturbations without triggering the feedback loop.
\end{proposition}

\begin{proof}
When $\theta_t > \theta^*$, $\varphi^d(0) > \varphi^{\mathrm{req}}$ and
$\rho_t = 0$ (case~(a) of the complementarity condition).  A perturbation
of size $\delta\theta$ moves $\theta_t$ to $\theta_t - \delta\theta$.
As long as $\theta_t - \delta\theta > \theta^*$, the system remains in
case~(a) with $\rho_t = 0$, and the feedback loop is not activated.  The
buffer $\delta$ must therefore be large enough that no plausible
perturbation (of size at most $\sigma$, amplified by at most
$\eta_t/(1-\eta_t)$ if $\eta_t < 1$ locally, or bounded by the distance
to $\theta^*$ if $\eta_t \geq 1$) can push the system below $\theta^*$.
\end{proof}

\begin{remark}[Policy implication]
\label{rem:policy-multiple}
The existence of multiple equilibria near the regime boundary has a
specific policy implication: maintaining a \emph{buffer} between $\theta_t$
and $\theta^*$ is not merely prudent but \emph{structurally necessary} to
prevent self-fulfilling transitions to the stress equilibrium.  This
reinforces the Timing Constraint of \Cref{sec:timing}: the policy sprint must be
executed while $\theta_t$ is sufficiently above $\theta^*$ that the buffer
condition~\eqref{eq:safe-selection} is satisfied.  Waiting until $\theta_t
\approx \theta^*$ risks entering the multiple-equilibrium zone, where the
outcome depends on expectations rather than fundamentals.
\end{remark}

\begin{remark}[Empirical testability]
The multiple-equilibrium region is in principle empirically detectable:
it predicts that sovereign spreads should exhibit \emph{regime-dependent
volatility}, with low and stable spreads when $\theta_t \gg \theta^*$
(the safe equilibrium is uniquely selected) and episodic spikes when
$\theta_t \approx \theta^*$ (the system enters the sunspot-vulnerable
zone).  The Italian BTP-Bund spread behavior during 2011--2012 and 2018
is broadly consistent with this prediction: spread episodes occurred when
ECB purchase programmes (which sustain $\theta_t$ in the Italian context)
were perceived as weakening, and resolved when new purchase commitments
restored the buffer.  A formal test would require estimating $\theta^*$
for the Italian case, which lies outside the scope of the present paper
but is a natural target for future empirical work.
\end{remark}

\end{document}